\documentclass[10pt]{article}
\usepackage[a4paper,margin=2cm]{geometry}
\usepackage{microtype}
\usepackage[utf8]{inputenc}
\usepackage[dvipsnames]{xcolor}
\usepackage{gosip}
\usepackage{enumerate}
\usepackage{hyperref}
\usepackage{array}
\usepackage{todonotes}
\usepackage{soul}
\usepackage{tikz}
\usetikzlibrary{calc}

\usepackage{algorithm}
\usepackage[noend]{algpseudocode}

\DeclareMathOperator{\cost}{cost}
\DeclareMathOperator{\sol}{sol}

\def \simplify {\textsc{Simplify}\xspace}
\def \randomcover {\textsc{RandomCover}\xspace}

\def \req {\ensuremath{\mathcal{T}}}

\def \djcut {\textsc{Disjunctive Multicut}\xspace}
\def \mdt {\textsc{Triple Multicut}\xspace}
\def \spc {\textsc{Split Paired Cut}\xspace}

\newcommand{\rel}[1]{\textnormal{\sf #1}\xspace}

\title{Parameterized Complexity of Equality MinCSP\thanks{
  The first author was supported by
  the Wallenberg AI, Autonomous Systems and Software Program (WASP) funded
  by the Knut and Alice Wallenberg Foundation.
}}
\author{
  George Osipov\thanks{Link{\"o}ping University, Sweden, \texttt{george.osipov@pm.me}} \and
  Magnus Wahlstr{\"o}m\thanks{Royal Holloway, University of London, UK, \texttt{Magnus.Wahlstrom@rhul.ac.uk}}
}

\newif\iflong
\newif\ifshort

\longtrue

\iflong
\else
  \shorttrue
\fi

\begin{document}

\begin{titlepage}
  \maketitle

  \begin{abstract}
    We study the parameterized complexity of MinCSP for so-called
    \emph{equality languages}, i.e., for finite languages over an
    infinite domain such as $\mathbb{N}$, where the relations are defined via
    first-order formulas whose only predicate is $=$.
    This is an important class of languages that forms the starting
    point of all study of infinite-domain CSPs under the commonly used
    approach pioneered by Bodirsky, i.e., languages defined as reducts
    of finitely bounded homogeneous structures.
    Moreover, MinCSP  over equality languages forms a natural class of
    optimisation problems in its own right, covering such problems as 
    \textsc{Edge Multicut}, \textsc{Steiner Multicut} and
    (under singleton expansion) \textsc{Edge Multiway Cut}.
    We classify $\textsc{MinCSP}(\Gamma)$ for every finite equality language $\Gamma$,
    under the natural parameter, 
    as either FPT, W[1]-hard but admitting a constant-factor FPT-approximation,
    or not admitting a constant-factor FPT-approximation unless FPT=W[2].
    In particular, we
    describe an FPT case that slightly generalises \textsc{Multicut}, and
    show a constant-factor FPT-approximation for \textsc{Disjunctive Multicut},
    the generalisation of \textsc{Multicut} where the ``cut requests''
    come as disjunctions over $d = O(1)$ individual cut requests $s_i \neq t_i$.
    We also consider \emph{singleton expansions} of equality languages,
    i.e., enriching an equality language with the capability for
    assignment constraints $(x=i)$, $i \in \NN$, for either a finite
    or infinitely many constants $i$, and fully characterize the
    complexity of the resulting MinCSP.
  \end{abstract}

  \thispagestyle{empty}
\end{titlepage}

\iflong
  \thispagestyle{empty}
  \tableofcontents
  \newpage
\fi

\setcounter{page}{1}
\section{Introduction}
\label{sec:intro}

Let $D$ be a fixed domain, and let $\Gamma$ be a finite set of finitary relations over $D$.
$\Gamma$ is referred to as a \emph{constraint language}.
A \emph{constraint} over $\Gamma$ is a pair $(R,X)$, less formally written $R(X)$, 
where $R \in \Gamma$ is a relation of some arity $r$ and $X=(x_1,\ldots,x_r)$ is a tuple of variables.
It is \emph{satisfied} by an assignment $\alpha$ if $(\alpha(x_1),\ldots,\alpha(x_r)) \in R$. 
For a constraint language $\Gamma$, the \emph{constraint satisfaction problem} over $\Gamma$,
\csp{\Gamma}, is the problem where an instance $I$ is a collection of
constraints over $\Gamma$, on some set of variables $V$, and the question
is if there is an assignment such that all constraints in $I$ are satisfied. 
In the optimization variant \mincsp{\Gamma}, the input also contains
an integer $k$ and the question is whether there is an assignment
such that all but at most $k$ constraints are satisfied.
Less formally, a constraint language $\Gamma$ determines the ``type of constraints''
allowed in an instance of \csp{\Gamma} or \mincsp{\Gamma},
and varying the constraint language defines problems of varying
complexity (such as $k$-SAT, \textsc{$k$-Colouring}, \textsc{$st$-Min Cut}, etc.). 
After decades-long investigations, \emph{dichotomy theorems} have been
established for these problems: for every constraint language over a
finite domain, \csp{\Gamma} and \mincsp{\Gamma} is either in P or NP-complete,
and the characterizations are known~\cite{Bulatov17CSP,Zhuk20CSP,ThapperZ16JACM,KolmogorovKR17}.
For fixed cases, such as the Boolean domain $D=\{0,1\}$, 
\emph{parameterized} dichotomies are also known, characterizing 
every problem \mincsp{\Gamma} as either FPT or W[1]-hard~\cite{KimKPW23flow3},
and similarly for approximate FPT algorithms~\cite{bonnet2016mincsp}.
This work represents significant advancements of our understanding of
tractable and intractable computational problems (classical or parameterized).

But as highlighted by Bodirsky~\cite{Bodirsky-Hab,BodirskyM17Dagstuhl},
there are also many problems from a range of application domains
that do not lend themselves to a formulation in the above CSP framework,
yet which can be formulated via CSPs over structures with \emph{infinite} domains. 
Unfortunately, CSPs with fixed templates over infinite domains are
not as well-behaved as over finite domains; it is known that
the problem $\csp{\Gamma}$ over an infinite domain can have any 
computational complexity (including being intermediate),
making any dichotomy impossible~\cite{BodirskyG08Nondich,BodirskyM17Dagstuhl}.
There are also questions of how an arbitrary infinite-domain
relation would be represented. The approach used by Bodirsky,
which is the standard approach for the study of infinite-domain CSPs,
is to consider a language $\Gamma$ as a \emph{reduct of a finitely
  bounded homogeneous structure}. Less technically, consider
a structure, for example $(\QQ,<)$ or $(\ZZ,<)$, and
let $\Gamma$ be a finite language where every relation in $\Gamma$
has a quantifier-free first-order definition over the structure; i.e.,
$\Gamma$ is a \emph{first-order reduct} of the structure.\footnote{The
  definition can be assumed to be quantifier-free
  since these structures admit quantifier elimination.}
For such languages a dichotomy is plausible, and many cases
have been settled, including \emph{temporal} CSPs, i.e.,
first-order reducts of $(\QQ,<)$~\cite{BodirskyK10Temporal};
\emph{discrete temporal} CSPs, i.e., first-order reducts of
$(\ZZ,<)$~\cite{BodirskyMM18discrete}; CSPs over the universal random
graph~\cite{BodirskyP15graphs}; and many more. 

Our goal is to study the parameterized complexity of MinCSPs over such structures. 
Many important problems in parameterized complexity, which are not well handled by
CSP optimization frameworks over finite-domain CSPs,
can be expressed very simply in this setting. 
For example, the MinCSP with domain $\QQ$ and the single relation $<$
is equivalent to the \textsc{Directed Feedback Arc Set} problem, i.e., 
given a digraph $D$ and an integer $k$, find a set $X$
of at most $k$ arcs from $D$ such that $D-X$ is acyclic.
(Here, the vertices of $D$ become variables, the arcs constraints, and
the topological order of $D-X$ becomes an assignment which violates
at most $|X|$ constraints.)
Other examples include \textsc{Subset Directed Feedback Arc Set},
which corresponds to $\mincsp{<,\leq}$, and
\textsc{Symmetric Directed Multicut} which corresponds
to $\mincsp{\leq,\neq}$.
The former is another important FPT problem~\cite{chitnis2015directed},
while FPT status of the latter is open~\cite{EibenRW22ipec}.

The structure we study in this paper is $(\NN,=)$.
The relations definable over this structure are called
\emph{equality constraint languages}. Here, $\NN$ is an arbitrary,
countably infinite domain; first-order reducts of $(\NN,=)$
are simply relations definable by a quantifier-free first-order
formula whose only predicate is $=$. Equivalently, relations in an equality
language accept or reject an assignment to their arguments purely
based on the partition that the assignment induces.
Since every first-order formula is allowed to use equality in this
framework, equality languages are contained in \emph{every} other
class of languages studied in the framework. Hence, characterizing the
complexity of equality languages is a prerequisite for studying any
other structure. 

Moreover, the setting also covers problems that are important in their
own right, as it captures undirected \emph{graph separation} problems. 
In particular, \textsc{(Vertex/Edge) Multicut} is defined as follows.
The input is a graph $G$, an integer $k$, and a set of \emph{cut
  requests} $\cT \subseteq \binom{V(G)}{2}$, and the task is to find a
set $X$ of at most $k$ vertices, respectively edges, such that for
every cut request $st \in \cT$, vertices $s$ and $t$ are in different
connected components in $G-X$. \textsc{Multicut} is FPT
parameterized by $k$ --- a breakthrough result, settling a long-open
question~\cite{marx2014fixed,BousquetDT18}. 
As with the above examples, there appears to be no natural way of capturing
\textsc{Multicut} as a finite-domain CSP optimization problem.\footnote{The 
  \emph{bipedal form} of \textsc{Multicut} encountered as a processing
  step in the solution~\cite{marx2014fixed} can be captured as a
  MinCSP over domain 3, but that is hardly a natural or immediate
  capture given the amount of work involved in reaching this stage.}
However, it naturally corresponds to $\mincsp{=,\neq}$ over domain $\NN$,
where edges correspond to soft $=$-constraints and cut requests to
crisp $\neq$-constraints. 
Another classic problem is \textsc{Multiway Cut}, which is the special
case of \textsc{Multicut} where the cut requests are $\cT=\binom{T}{2}$
for a set $T$ of \emph{terminal vertices} in the graph.
\textsc{Multiway Cut} was among the first graph separation problems
shown to be FPT~\cite{Marx06mwc}, and remains a relevant problem,
e.g., for the question of \emph{polynomial kernelization}~\cite{kratsch2020representative,Wahlstrom22talg}.
While \textsc{Multiway Cut} is not directly captured by an equality CSP,
it is captured by the \emph{singleton expansion} of the setting, i.e.,
adding ``assignment constraints'' (see later). 

We characterize \mincsp{\Gamma} for any equality language $\Gamma$
as being in P or NP-hard, FPT or W[1]-hard, and in terms of admitting
constant-factor FPT approximations. We also characterize 
the complexity for singleton expansions of $\Gamma$.

\subsection{Related work}

Bodirsky and Kara~\cite{bodirsky2008complexity} characterized
$\csp{\Gamma}$ as in P or NP-hard for every equality language $\Gamma$.
Bodirsky, Chen and Pinsker~\cite{BodirskyCP10equality}
characterized the structure of equality languages up to pp-definitions
(\emph{primitive positive definitions}, see Section~\ref{sec:prelim});
these are too coarse to preserve the parameterized complexity of a
problem, but their results are very useful as a guide to our search.
For much more material on CSPs over infinite domains, see
Bodirsky~\cite{Bodirsky-book21}.
Singleton expansions (under different names) are discussed by
Bodirsky~\cite{Bodirsky-book21} and Jonsson~\cite{Jonsson18constants}.
We have taken the term from Barto et al.~\cite{BartoKW17}.

Many variations on cut problems have been considered, and have been
particularly important in parameterized complexity~\cite{cygan2015parameterized}
(see also~\cite{crespelle2020survey}). We cover \textsc{Multicut},
\textsc{Multiway Cut} and \emph{Steiner cut}. 
Given a graph $G$ and a set of terminals $T \subseteq V(G)$,
a Steiner cut is an edge cut in $G$ that separates $T$,
i.e., 
a cut $Z$ such that some pair of vertices in $T$
is disconnected in $G - Z$.
\textsc{Steiner Cut} is the problem of finding a minimum Steiner cut.
This can clearly be solved in polynomial time; in fact, using advanced
methods, it can even be computed in \emph{near-linear} time~\cite{LiP20mincut,ChenKLPGS22maxflow}.
\textsc{Steiner Multicut} is the generalization where the input
contains a set $\cT=\{T_1,\ldots,T_t\}$ of terminal sets and the task
is to find a smallest-possible cut that separates every set $T_i$. Clearly, this
is NP-hard if $t \geq 3$. Bringmann et al.~\cite{bringmann2016parameterized}
considered parameterized variations of this and showed, among other results,
that \textsc{Edge Steiner Multicut} is FPT even for terminal sets
$T_i$ of unbounded size, if the parameter includes both $t$ and the
cut size $k$. On the other hand, parameterized by $k$ alone,
\textsc{Steiner Multicut} is W[1]-hard for terminal sets of size
$|T_i|=3$. 

Other parameterized CSP dichotomies directly
relevant to our work are the dichotomies for Boolean MinCSP
as having constant factor FPT-approximations or being W[1]-hard to approximate~\cite{bonnet2016mincsp}
(with additional results in a later preprint~\cite{Bonnet2016mincspCoRR})
and the recent FPT/W[1]-hardness dichotomy~\cite{KimKPW23flow3}.

The area of FPT approximations has seen significant activity in recent years,
especially regarding lower bounds on FPT approximations~\cite{FeldmannSLM20survey,KarthikS22ccc,BhattacharyyaBE21jacm,LinRSW23soda}.
In particular, we will need that there is no constant-factor
FPT-approximation for \textsc{Nearest Codeword} in Boolean codes
unless FPT=W[1]~\cite{Bonnet2016mincspCoRR,BhattacharyyaBE21jacm},
or for \textsc{Hitting Set} unless FPT=W[2]~\cite{LinRSW23soda}.
Lokshtanov et al.~\cite{lokshtanov2021fpt} considered fast FPT-approximations
for problems whose FPT algorithms are slow; in particular,
our result for \textsc{Steiner Multicut} 
builds on their algorithm giving an $O^*(2^{O(k)})$-time
2-approximation for \textsc{Vertex Multicut}. 

\subsection{Our results}

We study the classical and parameterized complexity of
$\mincsp{\Gamma}$ for every finite equality language $\Gamma$,
as well as for singleton expansions over equality languages. 
We consider both exact FPT-algorithms and constant-factor
FPT-approximations. We provide our results in overview here; for
details, see the body of the paper.

Unsurprisingly, $\mincsp{\Gamma}$ for an equality language $\Gamma$ is NP-hard except in trivial cases,
since $\mincsp{=,\neq}$ already corresponds to \textsc{Edge Multicut}. 
Specifically, \mincsp{\Gamma} is in P if $\Gamma$ is \emph{constant}, 
in which case every relation in $\Gamma$ contains the tuple $(1,\ldots,1)$,
or \emph{strictly negative},
in which case every relation in $\Gamma$ contains every tuple $(1,\ldots,r)$ 
where all values are distinct
(proper definitions of the terms are found in 
\iflong Section~\ref{sec:equalangs}\fi
\ifshort Section~\ref{sec:short}\fi
).
In all other cases, \mincsp{\Gamma} is NP-hard.

\iflong
\begin{theorem}[Theorem~\ref{thm:eq-mincsp-npn-dichotomy}~and~Corollary~\ref{cor:poly-inapprox}]
  \label{ithm:mincsp-p}
  Let $\Gamma$ be an equality constraint language. Then 
  $\mincsp{\Gamma}$ is in P if $\Gamma$ is constant or strictly negative.
  Otherwise, \mincsp{\Gamma} is  NP-hard and has no
  constant-factor approximation under the Unique Games Conjecture.
\end{theorem}
\fi

\ifshort
\begin{theorem}
  \label{ithm:mincsp-p}
  Let $\Gamma$ be an equality constraint language. Then 
  $\mincsp{\Gamma}$ is in P if $\Gamma$ is constant or strictly negative.
  Otherwise, \mincsp{\Gamma} is  NP-hard and has no
  constant-factor approximation under the Unique Games Conjecture.
\end{theorem}
\fi

For our FPT results, we introduce the following generalization of
\textsc{Vertex Multicut}.

\pbDefP{Vertex Multicut with Deletable Triples (aka \mdt)}
{A graph $G$, a collection $\cT \subseteq \binom{V(G)}{3}$ of vertex triples, and integer $k$.}
{$k$.}
{Are there subsets $Z_V \subseteq V(G)$ and $Z_{\cT} \subseteq \cT$ such that
  $|Z_V| + |Z_{\cT}| \leq k + 1$ and every connected component of $G - X_V$
  intersects every triple in $\cT \setminus X_{\cT}$ in at most one vertex?}

Note that this is a proper generalization of \textsc{Vertex Multicut}. On the
one hand, for any cut request $uv$ in a \textsc{Vertex Multicut} instance
we can create a triple $uvz$ for an auxiliary vertex $z$ not
connected to the rest of the graph. On the other hand, there is no
apparent way to implement triples $uvw$ in \textsc{Vertex Multicut}
with the condition that the whole triple can be ignored at unit cost. 
We show that \mdt is FPT.

\begin{theorem}[Theorem~\ref{thm:triple-multicut-fpt}]
  \mdt is fixed-parameter tractable.
\end{theorem}

For the FPT cases of \mincsp{\Gamma}, let
$\rel{NEQ}_3$ be the ternary relation which contains all tuples with
three distinct values, and let a \emph{split} constraint be a constraint $R$
of some arity $p+q$ for $p, q \geq 0$, defined (up to argument order) by
\[
  R(x_1,\ldots,x_p,y_1,\ldots,y_q) \equiv \bigwedge_{i, j \in [p]} (x_i=x_j)
  \land \bigwedge_{i \in [p], j \in [q]} (x_i \neq y_j).
\]
We note that \mincsp{\Gamma} with split constraints reduces to
\textsc{Vertex Multicut}. A split constraint $R(u_1,\ldots,u_p,v_1,\ldots,v_q)$
can be represented by introducing a new vertex $c$,
adding edges $cu_i$ for every $i \in [p]$ and
cut requests $cv_j$ for every $i \in [q]$.
Furthermore, a constraint $\rel{NEQ}_3(u,v,w)$ naturally corresponds
to a triple $uvw \in \cT$. Thus \mincsp{\Gamma} reduces to \mdt
if every relation is either split or $\rel{NEQ}_3$, and hence is FPT.
We show that all other cases are W[1]-hard and get the following.

\iflong
\begin{theorem}[from Theorem~\ref{thm:fpt-class}]
  Let $\Gamma$ be an equality constraint language that is not constant
  or strictly negative. Then \mincsp{\Gamma} is FPT if every relation in
  $\Gamma$ is either split or $\rel{NEQ}_3$, and W[1]-hard otherwise.
\end{theorem}
\fi
\ifshort
\begin{theorem}
  Let $\Gamma$ be an equality constraint language that is not constant
  or strictly negative. Then \mincsp{\Gamma} is FPT if every relation in
  $\Gamma$ is either split or $\rel{NEQ}_3$, and W[1]-hard otherwise.
\end{theorem}
\fi

Next, we describe the cases with constant-factor FPT-approximations. 
Again, we introduce a new problem to capture this.
Let $G$ be a graph.
A subset $L \subseteq \binom{V(G)}{2}$ of pairs
is a \emph{request list}, and
a set of vertices $X \subseteq V(G)$ \emph{satisfies} $L$ if
there is a pair $st \in L$ separated by $X$.
For a graph $G$ and a collection of request lists $\cL$,
let $\cost(G, \cL)$ be the minimum size of a set 
$X \subseteq V(G)$ that satisfies all lists in $\cL$.

\pbDefP{\djcut}
{A graph $G$,
  a collection $\cL$ of request lists, each of size at most $d$,
  and an integer $k$.}
{$k$.}
{Is $\cost(G, \cL) \leq k$?}

Note that \textsc{Steiner Multicut} is the special case of
\textsc{Disjunctive Multicut} where each request list is
$L=\binom{T_i}{2}$ for some terminal set $T_i$.
Our main algorithmic contribution is an FPT-approximation for \djcut.

\begin{theorem}[Theorems~\ref{thm:djcut-fpta} and~\ref{thm:steiner-fpa}]
  \label{ithm:fpa-algs}
  Let $d \in \NN$ be a constant. \textsc{Disjunctive Multicut}
  with request lists of length at most $d$ has a constant-factor
  FPT-approximation parameterized by $k$. 
  \textsc{Steiner Multicut} where every terminal set $T_i$
  has $|T_i| \leq d$ has an $O^*(2^{O(k)})$-time 2-approximation.
\end{theorem}

This precisely describes the FPT-approximable cases of \mincsp{\Gamma}:
For every equality constraint language $\Gamma$ such that \csp{\Gamma} is in P,
either \mincsp{\Gamma} reduces to \textsc{Disjunctive Multicut} in an
immediate way (up to a constant-factor approximation loss), 
implying a constant-factor FPT-approximation,
or there is a cost-preserving reduction from \textsc{Hitting Set}
to \mincsp{\Gamma}.
We refer to the latter as \mincsp{\Gamma} being \textsc{Hitting Set}-hard.

\iflong
\begin{theorem}[from Theorem~\ref{thm:fpta-class}]
  Let $\Gamma$ be an equality constraint language such that \csp{\Gamma} is in P.
  Then either \mincsp{\Gamma} reduces to \djcut
  and has a constant-factor FPT-approximation,
  or \mincsp{\Gamma} is \textsc{Hitting Set}-hard.
\end{theorem}
\fi
\ifshort
\begin{theorem}
  Let $\Gamma$ be an equality constraint language such that \csp{\Gamma} is in P.
  Then either \mincsp{\Gamma} reduces to \djcut
  and has a constant-factor FPT-approximation,
  or \mincsp{\Gamma} is \textsc{Hitting Set}-hard.
\end{theorem}
\fi

\paragraph{Singleton expansion.}

In addition to the above (main) results, we also investigate the effect of adding
constants to an equality language motivated by the problem \textsc{Multiway Cut}.
More precisely, for an equality language $\Gamma$, 
we investigate the effect of adding some
number of unary singleton relations $\{(i)\}$ to $\Gamma$.
This is equivalent to allowing ``assignment constraints'' $(x=i)$ in
$\mincsp{\Gamma}$. We consider adding either a finite number of
singletons, or every singleton relation.
For an equality language $\Gamma$ and an integer $c \in \NN$,
$c \geq 1$, we define $\Gamma_c^+=\Gamma \cup \{\{(i) \mid i \in [c]\}$
as the language $\Gamma$ with $c$ different singletons added,
and let $\Gamma^+$ denote $\Gamma$ with every singleton $\{(i)\}$, $i \in \NN$
added. 
\textsc{Edge Multiway Cut} corresponds to $\mincsp{\Gamma^+}$
over the language $\Gamma=\{=\}$,
and \textsc{$s$-Edge Multiway Cut}, the special case with $s$ terminals,
corresponds to \mincsp{\Gamma_s^+}.
By a \emph{singleton expansion of $\Gamma$} we refer to either
the language $\Gamma'=\Gamma^+$ or $\Gamma'=\Gamma_c^+$ for some $c \in \NN$.

As the first step of the characterization, 
we observe that if $\Gamma$
can express $=$ and $\neq$, then the singleton expansion adds no
power, i.e., $\mincsp{\Gamma^+}$ reduces back to $\mincsp{\Gamma}$
by introducing variables $c_1, \dots, c_m$ for arbitrarily many constants,
adding constraints $c_i \neq c_j$ whenever $i \neq j$,
and using constraints $x = c_i$ in place of assignments $x = i$.
For the rest of the characterization, we thus study the cases that
either cannot express equality, or cannot express disequality. 
We defer the explicit characterization of the cases to the main text,
but in summary we get the following result.
We say that a language is \emph{positive conjunctive} if every
relation $R \in \Gamma$ can be defined as a conjunction of clauses
$(x_i=x_j)$. In the below, %
\emph{is equivalent to} refers to the there being 
cost-preserving reductions in both directions (see Section~\ref{sec:prelim}).

\begin{theorem} \label{ithm:singleton-expansion-cases}
  Let $\Gamma$ be an equality constraint language and let $\Gamma'$
  be a singleton expansion of $\Gamma$. Then one of the following cases
  applies.
  \begin{itemize}
  \itemsep0em
  \item \mincsp{\Gamma'} is equivalent to \mincsp{\Gamma}
  \item \mincsp{\Gamma'} is trivial, i.e., always satisfiable
  \item \mincsp{\Gamma'} is equivalent to the MinCSP over a singleton
    expansion of the empty language $\Delta=\emptyset$, in which case
    \mincsp{\Gamma'} is in P
  \item \mincsp{\Gamma'} is equivalent to \mincsp{\Delta} for a Boolean
    language $\Delta$
  \item $\Gamma$ is strictly negative, \mincsp{\Gamma'} is NP-hard but
    FPT and has a constant-factor approximation
  \item \mincsp{\Gamma'} is equivalent to \mincsp{\Delta'}
    where $\Delta'$ is a singleton expansion
    of a positive conjunctive language $\Delta$
  \item \mincsp{\Gamma'} is \textsc{Hitting Set}-hard,
    $\Gamma$ is Horn and \csp{\Gamma'} is in P
  \item \csp{\Gamma'} is NP-hard
  \end{itemize}
\end{theorem}

Note the distinction between \emph{$\Gamma$ is positive conjunctive}
and \emph{\mincsp{\Gamma'} is equivalent to \mincsp{\Delta'} 
where $\Delta$ is positive conjunctive}. 
This distinction is the main subtlety of the result.
Consider a relation $R(x_1,\ldots,x_r) \equiv (|\{x_1,\ldots,x_r\}| \neq r)$.
For $\{R\}_c^+$ with $c<r$ there is never a need to use more than $c$ distinct values
in an assignment, hence $R$ becomes ``effectively trivial''. 
But the language $\{R\}_c^+$ for $c \geq r$ is intractable.
Finally, we note the cases of singleton expansions of a positive conjunctive language.
In particular, every such case reduces to
\textsc{Multiway Cut} up to a constant-factor loss.

\begin{theorem} \label{ithm:pos-conj}
  Let $\Gamma$ be a positive conjunctive language and $\Gamma'$ a
  singleton expansion of $\Gamma$ with at least three added singleton
  relations. Then $\mincsp{\Gamma'}$ is NP-hard but has a
  constant-factor approximation. Furthermore, $\mincsp{\Gamma'}$ is
  FPT if $\Gamma$ is split, otherwise W[1]-hard. 
\end{theorem}

\paragraph{Roadmap.}
Section~\ref{sec:prelim} contains technical preliminaries.
\ifshort
  Section~\ref{sec:short} contains an overview of the classification proof.
\else
  Section~\ref{sec:equalangs} properly defines equality constraint
  languages and gives the classification with proofs deferred to later sections.  
  Section~\ref{sec:reductions} contains hardness reductions 
  and reductions from \mincsp{\Gamma} to various cut problems.
\fi
Section~\ref{sec:triple-multicut} gives the FPT algorithm for \mdt.
Section~\ref{sec:djcut-and-steiner} gives the FPT approximation algorithms.
\iflong
  Section~\ref{sec:singleton-expansion} contains the classification for
  singleton expansions. 
  Section~\ref{sec:discuss} concludes the paper. 
\fi

\section{Preliminaries}
\label{sec:prelim}

\paragraph{Graph Separation.}

Let $G$ be an undirected graph.
Denote the vertex set of $G$ by $V(G)$ and the edge set by $E(G)$.
For a subset of edges/vertices $X$ in $G$, let $G - X$ denote 
the graph obtained by removing the elements of $X$ from $G$,
i.e. $G - X = (V(G), E(G) \setminus X)$ if $X \subseteq E(G)$ and
$G - X = G[V(G) \setminus X]$ if $X \subseteq V(G)$.
A \emph{cut request} is a pair of vertices $st \in \binom{V(G)}{2}$,
and an $st$-cut/$st$-separator is a subset of edges/vertices $X$
such that $G - X$ contains no path connecting $s$ and $t$.
We write that $X$ \emph{fulfills} $st$ if $X$ is an $st$-cut/$st$-separator.
We implicitly allows the inputs to cut problems such as
\textsc{Multiway Cut} and \textsc{Multicut} 
to contain undeletable edges/vertices:
for edges, with a parameter of $k$, we can include $k + 1$ parallel copies;
for vertices, we can replace a vertex $v$ with a clique
of size $k+1$, where every member of the clique
has the same neighbourhood as $v$.

\paragraph{Parameterized Deletion.}

A {\em parameterized} problem is a subset of $\Sigma^* \times \NN$,
where $\Sigma$ is the input alphabet.
The parameterized complexity class FPT contains problems decidable in 
$f(k) \cdot n^{O(1)}$ time, where $f$ is a computable function and 
$n$ is the bit-size of the instance.
Let $L_1, L_2 \subseteq \Sigma^* \times \NN$ be two parameterized problems.
A mapping $F: \Sigma^* \times \NN \rightarrow \Sigma^* \times \NN$
is an \emph{FPT-reduction} from $L_1$ to $L_2$ if
\begin{itemize}
  \itemsep0em
  \item $(x, k) \in  L_1$ if and only if $F((x, k)) \in L_2$, 
  \item the mapping can be computed in $f(k) \cdot n^{O(1)}$ time for some computable function $f$, and 
  \item there is a computable function $g : \NN \rightarrow \NN$  
  such that for all $(x,k) \in \Sigma^* \times \NN$, if $(x', k') = F((x, k))$, then $k' \leq g(k)$.
\end{itemize}
The classes W[1] and W[2] contains all problems that are 
FPT-reducible to \textsc{Clique} and \textsc{Hitting Set}, respectively,
parameterized by the solution size.
These problems are not in FPT under the standard assumptions
FPT$\neq$W[1] and FPT$\neq$W[2].
For a thorough treatment of parameterized complexity
we refer to~\cite{cygan2015parameterized}.

\paragraph{Constraint Satisfaction.}

Fix a \emph{domain} $D$.
A relation $R$ of \emph{arity} $r$ is a subset of tuples in $D^r$, i.e. $R \subseteq D^r$.
We write $=$ and $\neq$ to denote the binary equality and disequality relations over $D$,
i.e. $\{(a,b) \in D^2 : a = b\}$ and $\{(a,b) \in D^2 : a \neq b\}$, respectively.
A \emph{constraint language} $\Gamma$ is a set of relations over a domain $D$.
A \emph{constraint} is defined by a relation $R$
and a tuple of variables $\bx = (x_1, \dots, x_r)$, where $r$ is the arity of $R$.
It is often written as $R(\bx)$ or $R(x_1,\dots,x_r)$.
An assignment $\alpha : \{x_1, \dots, x_r\} \to D$ \emph{satisfies} 
the constraint if $\alpha(\bx) = (\alpha(x_1), \dots, \alpha(x_r)) \in R$,
and \emph{violates} the constraint if $\alpha(\bx) \notin R$.

\pbDef{Constraint Satisfaction Problem for $\Gamma$ (\csp{\Gamma})}
{An instance $I$, where $V(I)$ is a set of variables
  and $C(I)$ is a multiset of constraints using relations from $\Gamma$.}
{Is there an assignment $\alpha : V(I) \to D$ that satisfies all constraints in $C(I)$?}

\textsc{MinCSP} is an optimization version of the problem seeking
an assignment that minimizes the number of violated constraints.
In this constraints are allowed to be \emph{crisp} and \emph{soft}.
The \emph{cost of assignment $\alpha$} in an instance $I$ of \textsc{CSP}
is infinite if it violates a crisp constraint, and
equals the number of violated soft constraints otherwise.
The \emph{cost of an instance $I$} denoted by $\cost(I)$
is the minimum cost of any assignment to $I$.

\pbDef{\mincsp{\Gamma}}
{An instance $I$ of \csp{\Gamma} and an integer $k$.}
{Is $\cost(I) \leq k$?}

Next, we recall a useful notion that captures local reductions between CSPs.

\begin{definition} \label{def:pp-definition}
  Let $\Gamma$ be a constraint language over $D$
  and $R \subseteq D^r$ be a relation.
  A \emph{primitive positive definition (pp-definition)} 
  of $R$ in $\Gamma$ is an instance
  $C_R$ of \csp{\Gamma, =} with 
  primary variables $\bx$, auxiliary variables $\by$ and
  the following properties:
  \begin{enumerate}[(1)]
  \itemsep0em
  \item if $\alpha$ satisfies $C_R$, then it satisfies $R(\bx)$,
  \item if $\alpha$ satisfies $R(\bx)$, then
    there exists an extension of $\alpha$ to $\by$ that satisfies $C_R$.
  \end{enumerate}
\end{definition}

Informally, pp-definitions can be used to simulate $R$
using the relations available in $\Gamma$ and equality:
every constraint using $R$ can be replaced
by a gadget based on the pp-definition,
resulting in an equivalent instance.
The type of reductions captured by pp-definitions is however
incompatible with \textsc{MinCSP}
because the reductions do not preserve assignment costs.
\iflong
For example, consider the double-equality relation
$R = \{ (a,a,b,b) : a, b \in D \}$ and its pp-definition 
$C_R = \{x_1 = x_2, x_3 = x_4\}$.
The cost of assignment 
$(1,2,1,2)$ is one in $R(x_1,x_2,x_3,x_4)$ but two in $C_R$.
\fi
This motivates the following definition.

\begin{definition} \label{def:implementation}
  Let $\Gamma$ be a constraint language over $D$
  and $R \subseteq D^r$ be a relation.
  An \emph{implementation} of $R$ in $\Gamma$ is a pp-definition of $R$ 
  with primary variables $\bx$, 
  auxiliary variables $\by$ and an additional property:
  if $\alpha$ violates $R(\bx)$,
  there exists an extension of $\alpha$ to $\by$ of cost one.
\end{definition}

Although pp-definitions do not preserve costs, 
they can be used to simulate crisp constraints
in \textsc{MinCSP} instances.

\begin{proposition}[Proposition~5.2~in~\cite{kim2022flow}] \label{prop:ppdef-impl}
  Let $\Gamma$ be a constraint language over a domain $D$ and $R$ be a relation over $D$.
  Then the following hold.
  \begin{enumerate}
  \itemsep0em
  \item If $\Gamma$ pp-defines $R$, then there is an FPT-reduction
    from \mincsp{\Gamma, R} restricted to instances with 
    only crisp $R$-constraints to \mincsp{\Gamma, =}.
  \item If $\Gamma$ implements $R$, then there is an FPT-reduction
    from \mincsp{\Gamma, R} to \mincsp{\Gamma, =}.
  \end{enumerate}
\end{proposition}

\paragraph{Approximation.}

A minimization problem over an alphabet $\Sigma$
is a triple $(\cI, \sol, \cost)$, where 
\iflong
\begin{itemize}
  \itemsep0em
  \item $\cI \subseteq \Sigma^*$ is the set of instances,
  \item $\sol : \cI \to \Sigma^*$ is a function such that maps instances $I \in \cI$
    to the sets of solutions $\sol(I)$, and
  \item $\cost : \cI \times \Sigma^* \to \ZZ_{\geq 0}$
    is a function that takes an instance $I \in \cI$ and a solution
    $X \in \sol(I)$ as input,
    and returns a non-negative integer cost of the solution.
\end{itemize}
\fi
\ifshort
$\cI \subseteq \Sigma^*$ is the set of instances,
$\sol : \cI \to \Sigma^*$ is a function such that maps instances $I \in \cI$
to the sets of solutions $\sol(I)$, and
$\cost : \cI \times \Sigma^* \to \ZZ_{\geq 0}$
is a function that takes an instance $I \in \cI$ and a solution
$X \in \sol(I)$ as input,
and returns a non-negative integer cost of the solution.
\fi
Define $\cost(I) := \min \{ \cost(I, X) : X \in \cost(I) \}$.
A constant-factor approximation algorithm with factor $c \geq 1$
takes an instance $x \in \cI$ and an integer $k \in \NN$, and returns
`yes' if $\cost(I) \leq k$ and `no' if $\cost(I) > c \cdot k$.
A \emph{cost-preserving reduction} from a problem
$A = (\cI_A, \sol_A, \cost_A)$ to $B = (\cI_B, \sol_B, \cost_B)$
is a pair of functions polynomial-time computable functions $F$ and $G$ such that
\iflong
\begin{itemize}
  \itemsep0em
  \item for every $I \in \cI_A$, we have $F(I) \in \cI_B$ with $\cost_A(I) = \cost_B(F(I))$, and
  \item for every $I \in \cI_A$ and $Y \in \sol(F(I))$,
  we have $G(I, Y) \in \sol(I)$, and $\cost_A(I, G(I, Y)) \leq \cost_B(F(I), Y)$.
\end{itemize}
\fi
\ifshort 
(1) for every $I \in \cI_A$, we have $F(I) \in \cI_B$ with $\cost_A(I) = \cost_B(F(I))$, and
(2) for every $I \in \cI_A$ and $Y \in \sol(F(I))$,
we have $G(I, Y) \in \sol(I)$, and $\cost_A(I, G(I, Y)) \leq \cost_B(F(I), Y)$.
\fi
If there is a cost-preserving reduction from $A$ to $B$,
and $B$ admits a constant-factor polynomial-time/fpt 
approximation algorithm,
then $A$ also admits a constant-factor polynomial-time/fpt 
approximation algorithm.

\ifshort
  \section{Classification Overview}
  \label{sec:short}
  
  \begin{table}
    \begin{center}
      \begin{tabular}{| c | m{5cm} | l | l |}
  \hline
  Name & CNF Formula & Tuples & Complexity \\
  \hline
  \hline
  $\rel{EQ}_3$ &
  $({x_1 = x_2}) \land ({x_2 = x_3}) \land ({x_1 = x_3})$ & 
  $(1,1,1)$ & 
  \textcolor{Green}{FPT} \\
  \hline

  \iflong
  --- &
  $({x_1 = x_2})$ & 
  $(1,1,1), (1,1,2)$ & 
  \textcolor{Green}{FPT} \\
  \hline
    
  --- &
  $({x_1 \neq x_3}) \land ({x_2 \neq x_3})$ & 
  $(1,1,2), (1,2,3)$ & 
  \textcolor{Green}{FPT} \\
  \hline
  \fi
  
  $\rel{NEQ}_3$ &
  $({x_1 \neq x_2}) \land ({x_2 \neq x_3}) \land ({x_1 \neq x_3})$ & 
  $(1,2,3)$ & 
  \textcolor{Green}{FPT} \\
  \hline

  \iflong
  --- &
  $({x_2 \neq x_3})$ & 
  $(1,1,2), (1,2,1), (1,2,3)$ & 
  \textcolor{Green}{FPT} \\
  \hline
  \fi
  --- &
  $({x_1 = x_2}) \land ({x_1 \neq x_3}) \land ({x_2 \neq x_3})$ & 
  $(1,1,2)$ & 
  \textcolor{Green}{FPT} \\
  \hline
    
  $\rel{ODD}_3$ &
  $({x_1 = x_2} \lor {x_1 \neq x_3}) \land
  ({x_1 = x_2} \lor {x_2 \neq x_3}) \land ({x_1 \neq x_2} \lor {x_2 \neq x_3})$ & 
  $(1,1,1), (1,2,3)$ & 
  \textcolor{Red}{\textsc{Hitting Set}-hard} \\
  \hline

  \iflong
  --- &
  $({x_1 = x_2} \lor {x_1 \neq x_3}) \land ({x_1 = x_2} \lor {x_2 \neq x_3})$ & 
  $(1,1,1), (1,1,2), (1,2,3)$ & 
  \textcolor{Red}{\textsc{Hitting Set}-hard} \\
  \hline
      
  --- &
  $({x_1 \neq x_2} \lor {x_2 = x_3})$ & 
  excludes $(1,1,2)$ & 
  \textcolor{Red}{\textsc{Hitting Set}-hard} \\
  \hline
  \fi
  
  $\rel{NAE}_3$ &
  $({x_1 \neq x_2} \lor {x_2 \neq x_3})$ & 
  excludes $(1,1,1)$ & 
  \textcolor{Red}{W[1]-hard}, \textcolor{Blue}{FPA} \\
  \hline

  \hline

  $R^{\lor}_{\neq,\neq}$ &
  $\displaystyle ({x_1 \neq x_2} \lor {x_3 \neq x_4}) 
  \land ({x_1 \neq x_3}) \land ({x_1 \neq x_4}) \land ({x_2 \neq x_3}) \land ({x_2 \neq x_4})$ & 
  $(1,2,3,3)$, $(1,1,2,3)$, $(1,2,3,4)$ & 
  \textcolor{Red}{W[1]-hard}, \textcolor{Blue}{FPA} \\
  \hline

  $R^{\land}_{=,=}$ &
  $({x_1 = x_2}) \land ({x_3 = x_4})$ & 
  $(1,1,1,1)$, $(1,1,2,2)$ & 
  \textcolor{Red}{W[1]-hard}, \textcolor{Blue}{FPA} \\
  \hline


  $R^{\land}_{\neq,\neq}$ &
  $({x_1 \neq x_2}) \land ({x_3 \neq x_4})$ & 
  -- \textit{too many too list here} -- & 
  \textcolor{Red}{W[1]-hard}, \textcolor{Blue}{FPA} \\
  \hline


  $R^{\land}_{=,\neq}$ &
  $({x_1 = x_2}) \land ({x_3 \neq x_4})$ & 
  $(1,1,1,2), (1,1,2,1), (1,1,2,3)$ & 
  \textcolor{Red}{W[1]-hard}, \textcolor{Blue}{FPA} \\
  \hline


\end{tabular}

    \end{center}
    \caption{Selected Horn relations $R$ and the complexity of
      $\mincsp{R,=,\neq}$.
      FPA refers to fixed-parameter approximation.}
    \label{tab:horn-names}
  \end{table}  
  
  We now give an overview of the complexity dichotomy.
  Details are deferred to the full paper.
  We begin with a definition of the relevant language classes.
  Recall that an \emph{equality language} is a constraint language over $\NN$
  whose relations can be defined via Boolean formulas over the equality predicate.
  More precisely, for a set of variables $X=\{x_1,\ldots,x_n\}$,
  let a \emph{positive literal} be a term $(x_i=x_j)$
  and a \emph{negative literal} a term $(x_i \neq x_j)$, $i, j \in [n]$.
  A \emph{clause} is a disjunction of literals.
  Then every equality relation has a CNF definition as a conjunction
  of clauses. A relation (respectively language) is \emph{Horn}
  if it (respectively every relation in the language) has a CNF
  definition where every clause has at most one positive literal,
  \emph{negative} if positive literals only occur in singleton
  clauses $(x_i=x_j)$,
  \emph{strictly negative} if there are no positive literals,
  and \emph{conjunctive} if all clauses are
  singletons. Note that split relations and $\rel{NEQ}_3$ are both conjunctive. 

  Bodirsky and Kara~\cite{bodirsky2008complexity}
  showed that for an equality language $\Gamma$,
  \csp{\Gamma} is in P if $\Gamma$ is Horn or constant, and NP-hard otherwise.
  We note that \mincsp{\Gamma} is trivial if $\Gamma$ is constant or
  strictly negative, and also show that if $\Gamma$ is Horn but not
  constant or strictly negative then $\Gamma$ implements the relations $=$ and $\neq$.
  Hence we focus on this case and assume that $\Gamma$ is Horn and $=, \neq \; \in \Gamma$. 
  Theorem~\ref{ithm:mincsp-p} then follows since \textsc{Edge Multicut} reduces to \mincsp{=,\neq}.
  For the remaining steps of the classification, we show that
  \begin{enumerate}
    \itemsep0em
  \item $\mincsp{\Gamma, =, \neq}$ admits a constant-factor
    fpt-approximation if $\Gamma$ is negative, otherwise it is
    \textsc{Hitting Set}-hard;
    \item $\mincsp{R, =, \neq}$ is W[1]-hard if $R$ is negative but not conjunctive;
    \item $\mincsp{R, =, \neq}$ is W[1]-hard if $R$ is conjunctive but neither split nor $\rel{NEQ}_3$;
    \item $\mincsp{\Gamma, =, \neq}$ is in FPT if $\Gamma$ is conjunctive and all relations in 
    $\Gamma$ are split or $\rel{NEQ}_3$.
  \end{enumerate}
  Table~\ref{tab:horn-names} lists a number of Horn relations and the
  associated complexity of \textsc{MinCSP}.
  
  Towards hardness of approximation, we recall a result of
  Bodirsky, Chen and Pinsker~\cite{BodirskyCP10equality}
  that every equality language that is not negative can define
  $\rel{ODD}_3$ (see Table~\ref{tab:horn-names}).
  We show that \mincsp{\rel{ODD}_3,=,\neq} is \textsc{Hitting Set}-hard.
  
  \begin{lemma}
    There is a cost-preserving reduction from
    \textsc{Hitting Set} to \mincsp{\rel{ODD}_3,=,\neq}
    where every $\rel{ODD}_3$-constraint is crisp.
  \end{lemma}
  \begin{proof}[Proof Sketch]
    Let the input be $(V,\cE,k)$, $V = \{1,\dots,n\}$.
    Create an instance $(I, k)$ of $\mincsp{\rel{ODD}_3, =, \neq}$
    starting from variables $x_1,\dots,x_n$ and $z$, with
    soft constraints $x_i = z$ for all $i \in [n]$.
    For every set $e = \{a_1, \dots, a_{\ell}\} \in \cE$,
    add auxiliary variables $y_2, \dots, y_{\ell}$
    and crisp constraints
    $\rel{ODD}_3(x_{a_1}, x_{a_2}, y_2)$,
    $\rel{ODD}_3(y_{i-1}, x_{a_i}, y_{i})$ for all $3 \leq i \leq \ell$, and
    $x_{a_1} \neq y_\ell$.
    These constraints are satisfiable if and only if not all variables
    $x_i$, $i \in e$ are equal,
    so to satisfy them
    it is sufficient to break a soft constraint $x_{a_i} = z$.
    Thus, $X \subseteq V$ is a hitting set
    if and only if 
    $I - \{x_i=z : i \in X\}$ is consistent.
  \end{proof}
  
  As noted, this implies that \mincsp{\Gamma} is W[1]-hard to even
  approximate in FPT time. Now assume that $\Gamma$ is negative.
  Then every relation $R \in \Gamma$ is defined by a conjunction of
  positive singleton clauses and strictly negative clauses.
  For a constant-factor approximation,
  by~\cite[Lemma~10]{bonnet2016mincsp}
  we may split these definitions
  into separate constraints, so we may assume that an instance of
  \mincsp{\Gamma} uses constraint types $(x_i=x_j)$
  and $(x_1 \neq y_1 \lor \ldots \lor x_r \neq y_r)$, of lengths
  $r \leq d$ for some constant $d$. Then \mincsp{\Gamma} reduces to
  \djcut, and since $d=O(1)$
  we get an FPT approximation via Theorem~\ref{ithm:fpa-algs}.
  This settles the cases where \mincsp{\Gamma} has a constant-factor
  FPT approximation.

  \begin{figure}
    \centering
    \begin{tikzpicture}
  \def \radius {2}
  \def \step {-360/7}

  \draw (0,0) circle (\radius);
  
  \coordinate (v5) at ({\radius*cos(270 + \step*0)},{\radius*sin(270 + \step*0)});
  \coordinate (v6) at ({\radius*cos(270 + \step*1)},{\radius*sin(270 + \step*1)});
  \coordinate (v0) at ({\radius*cos(270 + \step*2)},{\radius*sin(270 + \step*2)});
  \coordinate (v1) at ({\radius*cos(270 + \step*3)},{\radius*sin(270 + \step*3)});
  \coordinate (v2) at ({\radius*cos(270 + \step*4)},{\radius*sin(270 + \step*4)});
  \coordinate (v3) at ({\radius*cos(270 + \step*5)},{\radius*sin(270 + \step*5)});
  \coordinate (v4) at ({\radius*cos(270 + \step*6)},{\radius*sin(270 + \step*6)});

  \draw[red, thick]  (v0) -- (v3);
  \draw[red, dashed] (v1) -- (v4);
  \draw[red, dashed] (v2) -- (v5);
  \draw[red, dashed] (v3) -- (v6);
  \draw[red, thick]  (v4) -- (v0);
  \draw[red, thick]  (v5) -- (v1);
  \draw[red, thick]  (v6) -- (v2);

  \filldraw[black] (v5) circle (1pt) node[anchor={-270}] {$v_5$};
  \filldraw[black] (v6) circle (1pt) node[anchor={-270+\step*1}] {$v_6$};
  \filldraw[black] (v0) circle (1pt) node[anchor={-270+\step*2}] {$v_0$};  
  \filldraw[black] (v1) circle (1pt) node[anchor={-270+\step*3}] {$v_1$};  
  \filldraw[black] (v2) circle (1pt) node[anchor={-270+\step*4}] {$v_2$};  
  \filldraw[black] (v3) circle (1pt) node[anchor={-270+\step*5}] {$v_3$};  
  \filldraw[black] (v4) circle (1pt) node[anchor={-270+\step*6}] {$v_4$};  

\end{tikzpicture}
    \caption{An illustration of a choice gadget for $t = 3$.
      Black arcs correspond to equality constraints,
      dashed red edges to soft disequality constraints, and
      the bold red edge to a crisp disequality constraint.}
    \label{fig:wheel}
  \end{figure}

  Towards delineating FPT and W[1]-hard cases, assume that $\Gamma$ is negative but not
  conjunctive, and let $R \in \Gamma$ be a relation such that every
  CNF definition of $R$ contains a non-singleton clause. 
  We show that in this case
  $\{R,=,\neq\}$ pp-defines $R^{\lor}_{\neq,\neq}$ or $\rel{NAE}_3$,
  and that $\mincsp{R', =, \neq}$ is W[1]-hard with
  crisp $R'$-constraints, where $R'$ is either $R^{\lor}_{\neq,\neq}$ or $\rel{NAE}_3$.
  The problem \mincsp{\rel{NAE}_3,=} with crisp $\rel{NAE}_3$-constraints
  naturally corresponds to \textsc{Steiner Multicut}, which is W[1]-hard~\cite{bringmann2016parameterized}.
  For the remaining proofs, we will need a \emph{choice gadget}. 
  Let $S = \{s_1, \dots, s_t\}$ be a set.
  Define an instance $W(S)$ of $\csp{=,\neq}$ as follows.
  Introduce $2t+1$ variables $v_0, \dots, v_{2t}$
  and identify indices modulo $2t+1$, e.g. $v_0 = v_{2t+1}$.
  Connect variables in a cycle of equalities, i.e.
  add soft constraints $v_i = v_{i+1}$ for all $0 \leq i \leq 2t$.
  The \emph{forward partner} of a variable $v_i$ is $f(v_i) := v_{i+t}$,
  i.e. the variable that is $t$ steps ahead of $v_i$ on the cycle.
  Add soft constraints $v_i \neq f(v_i)$ for all $0 \leq i \leq 2t$,
  and make the constraint $v_0 \neq v_t$ crisp.
  See Figure~\ref{fig:wheel} for an illustration.
  
  \begin{lemma} \label{lem:wheel-gadget}
    Let $S$ be a set of size at least two and $W(S)$ be the choice gadget.
    Then $\cost(W(S)) = 3$.
    Moreover, if $X \subseteq W(S)$, $|X| = 3$, $W(S) - X$ is consistent
    and $X$ contains $v_i \neq f(v_i)$,
    then $X$ also contains $v_{i-1} = v_{i}$ and $f(v_{i}) = f(v_{i+1})$.
  \end{lemma}
  \begin{proof}[Proof Sketch]
  Since $v_0 \neq v_t$ is a crisp constraint, we need to remove
  one link from the top and one from the the bottom
  chain of equality constraints connecting $v_0$ and $v_t$.
  After this deletion, at least one chain of length 
  $\ceil{\frac{2t + 1 -2}{2}} = t$ remains, which connects
  a vertex $v_i$ with its forward partner $f(v_i)$.
  Thus, we either need to disconnect $v_i$ and $f(v_i)$,
  or remove the constraint $v_i \neq f(v_i)$.
  Observe that it is sufficient to delete
  $v_{i-1} = v_{i}$, $f(v_{i}) = f(v_{i+1})$ and 
  $v_{i} \neq f(v_{i})$ to satisfy $W(S)$.
  Moreover, if a deletion set $X$ of size $3$ contains
  $v_{i} \neq f(v_{i})$, the only pair of vertices
  $v_{j}$ and $f(v_{j})$ that may remain connected 
  in $W(S) - X$ is the pair with $j = i$.
  Out of the two compatible choices,
  only deleting $v_{i-1} = v_{i}$ and $f(v_{i}) = f(v_{i+1})$ 
  leaves no constraint unsatisfied.
  \end{proof}
  
  Intuitively, deleting $v_{i-1} = v_{i}$, $f(v_{i}) = f(v_{i+1})$ and 
  $v_{i} \neq f(v_{i})$ from $W(S)$
  corresponds to choosing element $s_i$ from the set $S$.
  We note a simple reduction from
  \textsc{Multicoloured Independent Set} to 
  \mincsp{R^{\lor}_{\neq,\neq},=}.
  Let the input be a graph $G$ with $k$ colour classes $V(G)=V_1 \cup \ldots V_k$. 
  Create a choice gadget for every $V_i$ without soft $\neq$-constraints
  (but keeping the crisp one), and set the budget to $2k$.
  For every pair $u \in V_i$, $v \in V_j$ for $i \neq j$ such that $uv \in E(G)$,
  add a crisp constraint $R^\lor_{\neq,\neq}(u,f(u),v,f(v))$. 
  Due to the budget, $u=f(u)$ holds for at least one $u \in V_i$ for
  every $V_i$, corresponding to a selection,
  and the $R$-constraint prevents that $u=f(u)$ and $v=f(v)$,
  thereby blocking the combination of $u \in V_i$ and $v \in V_j$.
  We defer the details.
  
  Having shown that \mincsp{\Gamma,=,\neq} is W[1]-hard 
  if $\Gamma$ is non-conjunctive, it remains to show hardness for
  languages which are conjunctive but not $\rel{NEQ}_3$ or split.
  We use the following W[1]-hard problem (see~\cite[Lemma~6.1]{dabrowski2023almost}).
  
  \pbDefP{Split Paired Cut}
  {Graphs $G_1, G_2$, vertices $s_1, t_1 \in V(G_1)$, $s_2, t_2 \in V(G_2)$,
    a family of disjoint edge pairs $\cP \subseteq E(G_1) \times E(G_2)$,
    and an integer $k$.}
  {$k$.}
  {Is there a subset $X \subseteq \cP$ of size at most $k$ 
    such that for both $i \in \{1,2\}$,
    $\{e_i : \{e_1, e_2\} \in X \}$ is an $st$-cut in $G_i$?}
  
  Assuming $\Gamma$ is conjunctive, associate with a relation $R(x_1,\ldots,x_r)$ a graph
  where for each pair $i, j  \in [r]$ there is a blue edge $x_ix_j$
  if $i \neq j$ and $R$ implies $x_i=x_j$,
  and a red edge $x_ix_j$ if $R$ implies $x_i \neq x_j$. 
  Then the graph of a split relation $R(X,Y)$,
  $X=\{x_1,\ldots,x_p\}$ and $Y=\{y_1,\ldots,y_q\}$
  is a blue clique on $X$ and a red biclique of $(X,Y)$,
  and the graph of $\rel{NEQ}_3$ is a red triangle.
  The remaining cases are precisely the cases
  when the graph contains two disjoint edges $x_1x_2$, $x_3x_4$
  with no blue edges between their endpoints (such as
  $R^\land_{=,=}$, $R^\land_{=,\neq}$ or $R^\land_{\neq,\neq}$).   
  There is a curious parallel to the characterization of
  FPT cases of \textsc{Boolean MinCSP}, in terms of 2CNF-definable
  relations whose \emph{Gaifman graph} is
  $2K_2$-free~\cite{KimKPW23flow3}.
  We show hardness for all such cases. If both $x_1x_2$ and $x_3x_4$
  are blue, a reduction is immediate from \textsc{Split Paired Cut}.
  We sketch the reduction for a more interesting case of 
  \mincsp{R^\land_{\neq,\neq},=,\neq}, and note that the reduction for 
  \mincsp{R^\land_{=,\neq},=,\neq} is a variant of the above, and 
  hence omitted in this version of the paper.
 
  \begin{lemma}
    \mincsp{R,=} is W[1]-hard for $R=R_{\neq,\neq}^\land$.
  \end{lemma}
  \begin{proof}[Proof sketch]
    Let $(G_1, G_2, s_1, t_1, s_2, t_2, \cP, k)$ be an instance of \spc.
    By the construction of~\cite[Lemma~5.7]{kim2020solving},
    assume $k = 2\ell$, and $\cF_i$ for $i \in \{1,2\}$ are 
    $s_i t_i$-maxflows in $G_i$ partitioning 
    $E(G_i)$ into $k$ pairwise edge-disjoint paths.
    Construct an instance $(I, k')$ of \mincsp{R, =, \neq} 
    with $k'= 5\ell$ as follows.
    Start by creating a variable for every vertex in $V(G_1) \cup
    V(G_2)$ with the same name as the vertex.
    For each $i \in \{1,2\}$, consider a path $P \in \cF_i$, and 
    let $p$ be the number of edges on $P$.
    Create a choice gadget $W(P)$ for every $P$ with variables 
    $v^{P}_0, \dots, v^{P}_p$ following the path, and
    fresh variables $v^{P}_{j}$ for $p < j \leq 2p$ added to the instance.
    Observe that variables may appear on several paths in $\cF_i$.
    In particular, $v^P_{0} = s_i$ and $v^P_{p} = t_i$ for every $P \in \cF_i$,
    so we have crisp constraints $s_i \neq t_i$.
    Furthermore, since $\cF_i$ partitions $E(G_i)$, 
    the construction contains a copy of graphs $G_1$ and $G_2$
    with equality constraints for edges.
    Now we pair up edges according to $\cP$.
    For every pair $\{e_1, e_2\} \in \cP$,
    let $P \in \cF_1$ and $Q \in \cF_2$ be the paths such that $e_1 \in P$ and $e_2 \in Q$, and 
    suppose $e_1 = v^{P}_{i-1} v^{P}_{i}$ and $e_2 = v^{Q}_{j-1} v^{Q}_{j}$.
    Pair up soft constraints $v^{P}_i \neq f(v^{P}_i)$ and $v^{Q}_j \neq f(v^{Q}_j)$,
    i.e. replace individual constraints with one soft constraint
    $R(v^{P}_{i}, f(v^{P}), v^{Q}_{j}, f(v^{Q}_{j}))$.
    Finally, if an edge $uv \in E(G)$ does not appear in any pair of $\cP$,
    make constraint $u = v$ crisp in $I$.
    This completes the construction.

    To sketch correctness, observe that $I$ contains $2\ell$ choice gadgets.
    By Lemma~\ref{lem:wheel-gadget}, a solution is required
    to delete a pair of equality constraints from each plus another constraint. 
    Choice gadgets do not share equality constraints,
    so the remaining budget after deleting the pairs of 
    equality constraints is $\ell = 5\ell - 2 \cdot 2\ell$.
    Thus, the remaining $\ell$ deletions in the $2\ell$ gadgets 
    have to synchronize according to the $R$-constraints
    (to remove two soft disequality constraints at cost one).
    Furthermore, crisp constraints force 
    the selection to form an $s_1t_1$-cut and an $s_2t_2$-cut. 
    We omit the details of the verification.
  \end{proof}
  
  For all remaining cases, as noted, every $R \in \Gamma$ is either
  split or $\rel{NEQ}_3$, and there is a simple reduction from
  \mincsp{\Gamma} to \mdt. Since the latter is FPT
  (Theorem~\ref{thm:triple-multicut-fpt})
  the classification is complete.

  We omit the details of the classification for singleton expansion
  cases. The proofs are somewhat more intricate, explicitly using the
  algebraic method as employed by Bodirsky, Chen and Pinsker~\cite{BodirskyCP10equality},
  but ultimately both the hardness proofs and algorithmic cases are
  less interesting than for the above.  
\fi

\iflong
  \section{Classifications for Equality Constraint Languages}
  \label{sec:equalangs}

  This section offers a bird's eye view of complexity classifications
  of \textsc{MinCSP} over equality constraint languages,
  including our main results -- parameterized complexity
  and parameterized approximation classifications for
  equality constraint languages and their singleton expansions.
  The aim is to develop introduce necessary definitions,
  present some basic observations, and
  give an overview of the proof strategies.
  All technical proofs are delegated to subsequent sections.

  Fix an infinite (countable) domain, e.g. the set of natural numbers $\NN$.
  A set of relations $\Gamma$ over $\NN$ is an \emph{equality constraint 
    language} if all its relations are preserved by every permutation of the domain.
  Syntactically, equality constraint relations can be defined by Boolean formulas 
  using the equality relation, conjunction, disjunction, and negation, e.g.
  \begin{equation*}
    R(x_1, x_2, x_3) \equiv (x_1 = x_2 \land x_2 \neq x_3) \lor
    (x_2 = x_3 \land x_1 \neq x_2) \lor
    (x_1 = x_3 \land x_2 \neq x_3).
  \end{equation*}
  is an equality constraint relation.
  Atomic formulas $x_i = x_j$ and $x_i \neq x_j$ are referred to as \emph{positive}
  and \emph{negative literals}, respectively.
  We only consider proper relations, i.e. relations that are neither empty nor complete.
  Moreover, we assume that relations do not have redundant arguments,
  i.e. if the arity of $R$ is $r$,
  every definition of $R$ is a formula on $r$ variables.
  The relation-defining formulas can be converted into conjunctive normal form (CNF),
  and we refer to the conjunctions in CNF formulas as \emph{clauses}.

  There is another way to define equality constraint relations.
  Recall that these relations are closed under all automorphisms of $\NN$.
  One can think of the orbits of tuples in a relation
  under the action of automorphisms as partitions of indices,
  and define the relation by the list of non-isomorphic tuples in it.
  With this in mind, we invoke a definition.

  \begin{definition} \label{def:tuple-refine}
    Let $\ba, \bb \in \NN^r$ be two tuples.
    We say that $\ba$ \emph{refines} $\bb$ if $\ba_i = \ba_j$ implies $\bb_i = \bb_j$
    for all $i,j \in \range{1}{r}$.
    Furthermore, if there exist indices $i,j \in \range{1}{r}$ such that
    $\ba_i \neq \ba_j$ and $\bb_i = \bb_j$, 
    then $\ba$ \emph{strictly refines} $\bb$.
  \end{definition}

  For example, tuple $(1, 1, 2, 3, 4)$ strictly refines $(5, 5, 5, 6, 7)$.
  Refinement is a partial order on the tuples.

  \subsection{Polynomial-Time Complexity of Equality CSP}
  \label{ssec:equalangs-csp}

  Polynomial-time complexity of CSP for equality constraint languages
  was classified by Bodirsky and K{\'a}ra~\cite{bodirsky2008complexity}.
  To describe the dividing line between tractable and NP-hard problems,
  we recall the necessary definitions.

  \begin{definition}
    Let $R$ be an equality constraint relation.
    \begin{itemize}
    \item $R$ is \emph{constant} if it contains the tuple $(1, \dots, 1)$.
    \item $R$ is \emph{Horn} if it is definable by a CNF formula with at most
      one positive literal in each clause.
    \end{itemize}
    An equality constraint language is constant/Horn
    if all its relations are constant/Horn, respectively.
  \end{definition}

  \begin{example}
    Consider two relations defined by the formulas 
    $(x_1 = x_2 \lor x_3 = x_4)$ and
    $(x_1 = x_2 \lor x_3 \neq x_4) \land (x_1 \neq x_3)$.
    The first formula is not Horn since it has two positive literals 
    in a single clause, while the second relation is Horn.
    It is not hard to show that the relation 
    defined by the first formula is in fact not Horn,
    i.e. it admits no equivalent Horn formulation.
    The first formula defines a constant relation, while the second does not.
  \end{example}

  Every constraint using a constant relation
  is satisfies by setting all variables to the same value, 
  so every instance of CSP for constant languages is trivially consistent.
  To check if an instance of CSP over a Horn language
  is consistent, one can first propagate positive unit clauses $(x = y)$
  by identifying variables $x$ and $y$,
  removing falsified literals $x \neq x$ from every clause,
  and repeating the procedure until either an empty clause is derived
  or there are no more positive unit clauses to propagate.
  If an empty clause is derived, then we have a no-instance.
  Otherwise, we obtain a Horn instance without positive unit clauses,
  which is satisfiable by any assignment of distinct values to all variables.
  Bodirsky and K{\'a}ra~\cite{bodirsky2008complexity} proved that
  constant and Horn are the only polynomial-time solvable cases.

  \begin{theorem}[Theorem~1~in~\cite{bodirsky2008complexity}]
    \label{thm:eq-csp-dichotomy}
    Let $\Gamma$ be an equality constraint language.
    Then \csp{\Gamma} is solvable in polynomial time if 
    $\Gamma$ is either constant or Horn, and is NP-complete otherwise.
  \end{theorem}

  \subsection{Polynomial-Time Complexity of Equality MinCSP}
  \label{ssec:equalangs-mincsp}

  \textsc{MinCSP} for equality constraint languages is NP-hard even for very simple languages: indeed,
  \mincsp{=,\neq} is NP-hard by a simple reduction from \textsc{Edge Multicut}.
  As it turns out, every \textsc{MinCSP} for an equality constraint language 
  is either solvable in polynomial time by a trivial algorithm or NP-hard.
  Towards proving this, we define another class of equality languages.

  \begin{definition}
    An equality constraint relation is \emph{strictly negative} if it admits a CNF definition
    with only negative literals.
    An equality constraint language is strictly negative if all its relations are strictly negative.
  \end{definition}
  Strictly negative relations are a special case of Horn relations.
  The following lemma implies NP-hardness of \textsc{MinCSP}
  for almost all equality constraint languages.

  \begin{lemma}[Section~\ref{ssec:expressive}] 
  \label{lem:neither-const-nor-neg}
    Let $R$ be an equality constraint relation.
    \begin{enumerate}
    \item If $R$ is not constant, then $R$ implements $x_1 \neq x_2$.
    \item If $R$ is Horn but not strictly negative, then $R$ implements $x_1 = x_2$.
    \end{enumerate}
  \end{lemma}
  We use Lemma~\ref{lem:neither-const-nor-neg} to prove NP-hardness.
  In fact, we prove a slightly stronger result.

  \begin{lemma}[Section~\ref{ssec:multicut-variants}]
  \label{lem:multicut-to-mincsp}
    Let $\Gamma$ be an equality constraint language that implements
    binary equality and pp-defines binary disequality relations.
    Then there is a cost-preserving reduction from 
    \textsc{Edge Multicut} to \mincsp{\Gamma}.
  \end{lemma}

  The classification follows by observing that if $\Gamma$
  is neither Horn nor constant, then even $\csp{\Gamma}$ is NP-hard.

  \begin{theorem}
    \label{thm:eq-mincsp-npn-dichotomy}
    Let $\Gamma$ be an equality constraint language.
    \mincsp{\Gamma} is solvable in polynomial time if
    $\Gamma$ is constant or strictly negative,
    and is NP-hard otherwise.
  \end{theorem}

  \subsection{Parameterized Complexity of Equality MinCSP}

  We classify parameterized complexity of \mincsp{\Gamma} 
  for finite equality constraint language $\Gamma$. 
  By Theorem~\ref{thm:eq-mincsp-npn-dichotomy},
  we may focus on Horn equality languages $\Gamma$ that are neither constant nor strictly negative.
  To present the results, we need more definitions.

  \begin{definition}
    Let $R$ be an equality constraint relation.
    \begin{itemize}
    \item $R$ is \emph{negative} if it is definable by
      a CNF formula with positive literals occurring only in singleton clauses.
    \item $R$ is \emph{conjunctive} if it is definable by
      a CNF formula without disjunction.
    \item $R$ is \emph{split} if it is definable by a CNF formula
      $\bigwedge_{p,p' \in P} (x_p = x_{p'}) \land \bigwedge_{p \in P, q \in Q} (x_p \neq x_q)$,
      where $R$ is a relation of arity $p + q$ and $P \uplus Q$ is a partition of the indices $\{1,\dots,p+q\}$.
    \end{itemize}
  \end{definition}

  \begin{table}
    \begin{center}
      
    \end{center}
    \caption{Several Horn relations $R$ and the complexity of $\mincsp{R,=,\neq}$.
      The first part of the tables contains all ternary Horn relations up to permutation of indices.}
    \label{tab:horn-names}
  \end{table}  

  We remark that negative relations are Horn,
  conjunctive relations are negative, and
  split relations are conjunctive.
  We also define several important relations, see Table~\ref{tab:horn-names}.
  In particular, we need the `not-all-equal'
  relation ${\rel{NAE}_3(x_1, x_2, x_3) \equiv (x_1 \neq x_2 \lor x_2 \neq x_3)}$
  to state the main theorem.

  \begin{theorem} \label{thm:fpt-class}
    Let $\Gamma$ be a finite Horn equality constraint language.
    Assume $\Gamma$ is neither constant nor strictly negative.
    \begin{enumerate}
    \item \label{thm:fpt-class:not-ess-neg}
      If $\Gamma$ is not negative, then \mincsp{\Gamma} is W[2]-hard.
    \item \label{thm:fpt-class:ess-neg}
      If $\Gamma$ is negative and contains any relation that is neither split nor $\rel{NEQ}_3$, 
      then \mincsp{\Gamma} is W[1]-hard.
    \item \label{thm:fpt-class:split-nae}
      If $\Gamma$ contains only split relations and $\rel{NEQ}_3$, then \mincsp{\Gamma} is in FPT.
    \end{enumerate}
  \end{theorem}

  We present the proof of Theorem~\ref{thm:fpt-class} in four parts,
  one for part~\ref{thm:fpt-class:not-ess-neg},
  two parts for~\ref{thm:fpt-class:ess-neg},
  and one part for~\ref{thm:fpt-class:split-nae},
  The second point is split into two cases 
  based on whether $\Gamma$ is conjunctive or not.
  Throughout the section,
  we assume by Lemma~\ref{lem:neither-const-nor-neg} that $\Gamma$ implements $=$ and $\neq$ 

  \paragraph{Case~\ref{thm:fpt-class:not-ess-neg}: Not negative.}

  For non-negative languages, we recall a result of~\cite{BodirskyCP10equality}.

  \begin{theorem}[Theorem~67~of~\cite{BodirskyCP10equality}]
    \label{thm:not-essen-neg}
    If $R$ is a Horn equality constraint relation
    that is not negative,
    then \mincsp{R, =, \neq} pp-defines $\rel{ODD}_3$.
  \end{theorem}

  Combined with the following lemma and the fact that \textsc{Hitting Set} is W[2]-hard, 
  this yields Theorem~\ref{thm:fpt-class}.\ref{thm:fpt-class:not-ess-neg}.

  \begin{lemma}[Section~\ref{ssec:odd3-nae3}]
    \label{lem:hitting-set-to-odd3}
    There is a polynomial-time reduction that takes an instance $(V, \cE, k)$ of
    \textsc{Hitting Set}, and produces an instance $(I, k)$ of $\mincsp{\rel{ODD}_3, =, \neq}$
    where every $\rel{ODD}_3$-constraint is crisp.
    Furthermore,
    $(V, \cE, k)$ is a yes-instance if and only if
    $(I, k)$ is a yes-instance.
  \end{lemma}
  \begin{proof}[Proof Sketch]
    Let $V = \{1,\dots,n\}$, and
    construct an instance $(I, k)$ of $\mincsp{\rel{ODD}_3, =, \neq}$
    by introducing variables $x_1,\dots,x_n$ and $z$, and 
    add soft constraints $x_i = z$ for all $i \in [n]$.
    For every subset $e = \{a_1, \dots, a_{\ell}\} \in \cE$,
    introduce auxiliary variables $y_2, \dots, y_{\ell}$
    and crisp constraints
    $\rel{ODD}_3(x_{a_1}, x_{a_2}, y_2)$,
    $\rel{ODD}_3(y_{i-1}, x_{a_i}, y_{i})$ for all $3 \leq i \leq \ell$, and
    $x_{a_1} \neq y_\ell$.
    Correctness follows by observing that
    crisp constraints introduced for each $e \in \cE$ imply that
    not all variables $x_{a_1}, \dots, x_{a_\ell}$ are equal.
    Moreover, to satisfy these constraints,
    it is sufficient to break a soft $x_{a_i} = z$ constraint.
    Thus, $X \subseteq V$ is a hitting set
    if and only if 
    $I - \{x_i : i \in X\}$ is consistent.
  \end{proof}

  \paragraph{Case~\ref{thm:fpt-class:ess-neg}a: Negative, Not Conjunctive}

  Recall ternary `not-all-equals' $\rel{NAE}_3$ and
  the relation $R^{\lor}_{\neq,\neq}$
  from Table~\ref{tab:horn-names}.
  Our hardness results rely on the following two lemmas.

  \begin{lemma}[Section~\ref{ssec:split-paired-cut}] 
    \label{lem:disjneqneq-hard}
    $\mincsp{R^{\lor}_{\neq, \neq}, =}$ is W[1]-hard
    even restricted to instances with only crisp $R^{\lor}_{\neq, \neq}$-constraints.
  \end{lemma}

  \begin{lemma}[Section~\ref{ssec:odd3-nae3}]
    \label{lem:nae3-hard}
    $\mincsp{\rel{NAE}_3, =}$ is W[1]-hard
    even restricted to instances with only crisp $\rel{NAE}_3$-constraints.
  \end{lemma}

  Armed with Lemmas~\ref{lem:not-conj-nae-3-or-disjneqneq}~and~\ref{lem:nae3-hard},
  we show that every negative non-conjunctive language 
  pp-defines $\rel{NAE}_3$ or $R^{\lor}_{\neq,\neq}$,
  thus proving
  Theorem~\ref{thm:fpt-class}.\ref{thm:fpt-class:ess-neg}
  for all negative non-conjunctive languages.

  \begin{lemma}[Section~\ref{ssec:expressive}]
    \label{lem:not-conj-nae-3-or-disjneqneq}
    Let $R$ be a negative non-conjunctive 
    equality constraint relation.
    Then $\{R, =, \neq\}$ 
    pp-defines $R^{\lor}_{\neq,\neq}$ or $\rel{NAE}_3$.
  \end{lemma}

  \paragraph{Case~\ref{thm:fpt-class:ess-neg}b: negative Conjunctive, neither Split nor $\rel{NEQ}_3$}

  Let $R$ be a conjunctive relation of arity $r$. 
  Define edge-coloured graph $G_R$ with vertices $\{1,\dots,r\}$.
  Add a blue edge \textcolor{blue}{$ij$} whenever $R$ implies \textcolor{blue}{$x_i = x_j$},
  and a red edge \textcolor{red}{$ij$} whenever $R$ implies \textcolor{red}{$x_i \neq x_j$}.
  Note that blue edges in $G_R$ form cliques because equality relation is transitive,
  and red edges connect all members of one clique to all members of another.
  The graph $G_R$ for a split relation $R$ with indices $\{1,\dots,r\}$ 
  partitioned into $P \uplus Q$ 
  consists of a clique of blue edges on $P$
  and a biclique of red edges on $(P, Q)$.
  The graph $G_R$ for $R = \rel{NEQ}_3$ is a red triangle.
  If $R$ is neither split nor $\rel{NEQ}_3$,
  then $G_R$ contains two independent edges,
  which we assume are $\{1,2\}$ and $\{3,4\}$ without loss of generality,
  and every edge $uv \in E(G_R)$ with $u \in \{1,2\}$ and $v \in \{3,4\}$ is red.
  By considering only $x_1,x_2,x_3,x_4$ as primary variables and
  depending on the colours of edges $\{1,2\}$ and $\{3,4\}$,
  the projection of $R$ onto $\{1,2,3,4\}$ is one of the relations
  \begin{align*}
    &(\textcolor{blue}{x_1 = x_2}) \land (\textcolor{blue}{x_3 = x_4}) \land 
      \textstyle \bigwedge_{i \in A, j \in B} (x_i \neq x_j), \\
    &(\textcolor{blue}{x_1 = x_2}) \land (\textcolor{red}{x_3 \neq x_4}) \land 
      \textstyle \bigwedge_{i \in A, j \in B} (x_i \neq x_j), \textnormal{or} \\
    &(\textcolor{red}{x_1 \neq x_2}) \land (\textcolor{red}{x_3 \neq x_4}) \land 
      \textstyle \bigwedge_{i \in A, j \in B} (x_i \neq x_j)
  \end{align*}
  for some $A \subseteq \{1,2\}$ and $B \subseteq \{3,4\}$.
  We refer to the relations defined by the formulas above as
  \emph{$(=,=)$-relations}, \emph{$(=,\neq)$-relations} and \emph{$(\neq,\neq)$-relations}, respectively.
  Examples of such relations with $A = B = \emptyset$ are 
  $R^{\land}_{=,=}$, $R^{\land}_{=,\neq}$ and $R^{\land}_{\neq,\neq}$
  from Table~\ref{tab:horn-names}.

  \begin{observation}
    \label{obs:double-conj}
    If $R$ is a conjunctive equality constraint relation, and
    $R$ is neither split nor $\rel{NEQ}_3$, then
    $R$ implements an $(=,=)$-relation, an $(=,\neq)$-relation
    or a $(\neq,\neq)$-relation.
  \end{observation}

  We prove that $\mincsp{R, =, \neq}$ is W[1]-hard if
  $R$ is either an $(=,=)$-relation, an $(=,\neq)$-relation or a $(\neq,\neq)$-relation.
  The reductions in all three cases are from \spc (defined in Section~\ref{sec:reductions}),
  and we reuse some gadgets in the proofs, but the results are better presented as three separate lemmas.

  \begin{lemma}[Section~\ref{ssec:split-paired-cut}]
    \label{lem:eq-eq-hard}
    If $R$ is an $(=,=)$-relation, then $\mincsp{R, =, \neq}$ is W[1]-hard.
  \end{lemma}

  \begin{lemma}[Section~\ref{ssec:split-paired-cut}]
    \label{lem:eq-neq-hard}
    If $R$ is an $(=,\neq)$-relation, then $\mincsp{R, =, \neq}$ is W[1]-hard.
  \end{lemma}

  \begin{lemma}[Section~\ref{ssec:split-paired-cut}]
    \label{lem:neq-neq-hard}
    If $R$ is a $(\neq,\neq)$-relation, then $\mincsp{R, =, \neq}$ is W[1]-hard.
  \end{lemma}

  Observation~\ref{obs:double-conj} and 
  Lemmas~\ref{lem:eq-eq-hard},~\ref{lem:eq-neq-hard}~and~\ref{lem:neq-neq-hard}
  complete the proof of Theorem~\ref{thm:fpt-class}.\ref{thm:fpt-class:ess-neg}.

  \paragraph{Case~\ref{thm:fpt-class:split-nae}: Split and $\rel{NEQ}_3$}

  If $\Gamma$ only contains split relations and $\rel{NEQ}_3$,
  we show that the problem is in FPT via a  
  reduction from \mincsp{\Gamma} to \mdt.

  \begin{lemma}[Section~\ref{ssec:multicut-variants}]
    \label{lem:split-nae3-to-mdt}
    Let $\Gamma$ be an equality constraint language
    where every relation is either split or $\rel{NEQ}_3$.
    There is a polynomial-time reduction that
    takes an instance $(I, k)$ of \mincsp{\Gamma}
    and produces an instance $(G, \cT, k)$ of \mdt such that
    $(I, k)$ is a yes-instance if and only if $(G, \cT, k)$ is a yes-instance.
  \end{lemma}

  We solve \mdt by reducing it to \mincsp{\Gamma'}
  for a certain Boolean constraint language $\Delta$
  (i.e. the domain of $\Delta$ is \{0,1\}).
  Then we show that \mincsp{\Delta} is fixed-parameter tractable 
  using the full classification by~\cite{KimKPW23flow3}.

  \begin{theorem}[Theorem~\ref{thm:triple-multicut-fpt} in Section~\ref{sec:triple-multicut}]
    \mdt is fixed-parameter tractable.
  \end{theorem}

  Combining Lemma~\ref{lem:split-nae3-to-mdt} with Theorem~\ref{thm:triple-multicut-fpt}
  completes the proof of 
  Theorem~\ref{thm:fpt-class}.\ref{thm:fpt-class:split-nae}.

  \subsection{Approximation of Equality MinCSP}
  \label{sec:fpt-approx}

  Under the Unique Games Conjecture (UGC) of Khot~\cite{khot2002power},
  \textsc{Edge Multicut} is NP-hard to approximate within any constant~\cite{chawla2006hardness}.
  By Lemmas~\ref{lem:neither-const-nor-neg}~and~\ref{lem:multicut-to-mincsp}, 
  there is a cost-preserving
  reduction from \textsc{Edge Multicut} to \mincsp{\Gamma}
  whenever $\Gamma$ is Horn and \mincsp{\Gamma} is NP-hard.

  \begin{corollary}
  \label{cor:poly-inapprox}
    Let $\Gamma$ be a Horn equality constraint language. 
    If \mincsp{\Gamma} is NP-hard, then, assuming UGC, 
    it is NP-hard to approximate \mincsp{\Gamma} in 
    polynomial time within any constant.
  \end{corollary}

  Motivated by this hardness result, we study
  constant-factor approximation algorithms
  running in fpt time.
  Let $\Gamma$ be a Horn equality constraint language.
  By Lemma~\ref{lem:hitting-set-to-odd3},
  if $\Gamma$ is not negative, then
  it admits a cost-preserving reduction
  from \textsc{Hitting Set}.
  Unless FPT=W[2], \textsc{Hitting Set} does not admit 
  a constant-factor fpt approximation~\cite{lin2022constant},
  so there is little hope for obtaining
  constant-factor fpt approximation algorithms
  for \mincsp{\Gamma} when $\Gamma$ is not
  negative.
  It turns out that the hardness result is tight
  for equality \textsc{MinCSP}s.

  \begin{theorem} \label{thm:fpta-class}
    Let $\Gamma$ be a Horn equality constraint language.
    \mincsp{\Gamma} admits a constant-factor 
    approximation in fpt time if $\Gamma$ is negative,
    and is \textsc{Hitting Set}-hard otherwise.
  \end{theorem}

  We remark that the approximation factor depends on the language $\Gamma$.
  The following fact is used to obtain the approximation algorithm.

  \begin{lemma}[See e.g. Lemma~10~of~\cite{bonnet2016mincsp}]
    \label{lem:pp-def-preserves-apx}
    Let $\Gamma$ be a constraint language that pp-defines a relation $R$.
    If \mincsp{\Gamma,=} admits constant-factor fpt-approximation,
    then \mincsp{\Gamma,R} also admits constant-factor fpt-approximation.
  \end{lemma}

  The final approximation factor in the lemma above depends on 
  the number of constraints in the pp-definition of $R$ in $\Gamma$.
  Define relations 
  $R^{\neq}_d(x_1,y_1,\dots,x_d,y_d) \equiv \bigvee_{i=1}^{d} x_{i} \neq y_{i}$
  for all $d \in \NN$.
  Observe that every negative relation
  admits a (quantifier-free) pp-definition in 
  $\{R^{\neq}_d, =\}$ for some $d \in \NN$:
  the pp-definition contains a constraint for each clause,
  and $d$ is upper-bounded by the number of literals in a largest clause.
  By Lemma~\ref{lem:pp-def-preserves-apx},
  showing an fpt-approximation algorithm for $\mincsp{R^{\neq}_d, =}$
  is sufficient to prove Theorem~\ref{thm:fpta-class}.

  We solve $\mincsp{R^{\neq}_d, =}$ by a cost-preserving 
  reduction to \djcut.

  \begin{lemma}[Section~\ref{ssec:multicut-variants}]
    \label{lem:essen-neg-lmcut}
    There is a polynomial-time algorithm that takes an instance $I$ of
    $\csp{R^{\neq}_d, =}$ as input, and produces a graph $G$ and
    a collection of request lists $\cL$ such that 
    $\max_{L \in \cL} |L| = d + 1$ and $\cost(I) = \cost(G, \cL)$.
  \end{lemma}

  We prove that \djcut has a constant-factor fpt-approximation.

  \begin{theorem}[Theorem~\ref{thm:djcut-fpta}~in~Section~\ref{sec:djcut-and-steiner}]
    For every constant $d$, \djcut
    with request lists of length at most $d$ and
    admits an $f(d)$-factor fpt-approximation algorithm
    for some function $f$.
  \end{theorem}

  We remark that we do not optimize the function $f$ or the running time.
  Note that \textsc{Steiner Multicut} is a special case of \djcut where each list is a clique.
  This additional structure allows us to obtain a much simpler algorithm for this case.

  \begin{theorem}[Section~\ref{sec:djcut-and-steiner}] 
    \label{thm:steiner-fpa}
    \textsc{Steiner Multicut} with requests of constant size
    is $2$-approximable in $O^*(2^{O(k)})$ time.
  \end{theorem}

\fi

\iflong
  \section{Reductions}
  \label{sec:reductions}

  We present several reductions grouped into four parts.
  Section~\ref{ssec:expressive} contains some
  pp-definition and implementation results. proving
  Lemmas~\ref{lem:neither-const-nor-neg}~and~\ref{lem:not-conj-nae-3-or-disjneqneq}.
  Section~\ref{ssec:split-paired-cut} contains W[1]-hardness
  proofs by reduction from \textsc{Split Paired Cut} to
  $\mincsp{R, =, \neq}$ where
  $R$ is an $(=, =)$-relation (Lemma~\ref{lem:eq-eq-hard}),
  $R$ is a $(\neq,\neq)$-relation (Lemma~\ref{lem:neq-neq-hard}),
  $R$ is an $(=,\neq)$-relation (Lemma~\ref{lem:neq-neq-hard}), and
  $R = R^\lor_{\neq,\neq}$ (Lemma~\ref{lem:disjneqneq-hard}).
  Section~\ref{ssec:odd3-nae3} provides
  W[1]- and W[2]-hardness proofs for
  $\mincsp{\rel{NAE}_3, =, \neq}$ and $\mincsp{\rel{ODD}_3, =, \neq}$,
  respectively, proving Lemmas~\ref{lem:nae3-hard}~and~\ref{lem:hitting-set-to-odd3}.
  Finally, 
  in Section~\ref{ssec:multicut-variants}
  we present a reduction from \textsc{Edge Multicut} to $\mincsp{=, \neq}$ (Lemma~\ref{lem:multicut-to-mincsp})
  and reduction from \textsc{MinCSP} problems to
  \mdt and \djcut, supporting the positive results in the classification
  of exact fpt and fpt-approximation complexity
  (specifically, Lemmas~\ref{lem:split-nae3-to-mdt}~and~\ref{lem:essen-neg-lmcut}).

  \subsection{Expressive Power of Some Relations}
  \label{ssec:expressive}

  Let $R$ be an equality constraint relations,
  and $\phi_R$ be a CNF formula that defines $R$.
  We say that $\phi_R$ is \emph{reduced} if
  removing any clause or literal alters
  the defined relation.

  \begin{observation} \label{obs:minimal-formula}
    Let $\phi_R$ be a reduced definition of an equality constraint relation $R$.
    Suppose $\phi_R$ contains clause $C = \bigvee_{s=1}^{t} x_{i_s} \odot_s x_{j_s}$
    where $\odot_s \in \{=, \neq\}$.
    If $\phi_R$ is reduced, then for every $1 \leq u \leq t$,
    there is a tuple in $R$ that
    satisfies literal $x_{i_u} \odot_u x_{j_u}$,
    and violates $x_{i_s} \odot_s x_{j_s}$ for all $s \neq u$.    
  \end{observation}

  We can now prove Lemma~\ref{lem:neither-const-nor-neg}
  which states that a non-constant language implements $x_1 \neq x_2$, and
  a Horn, but not strictly negative language implements $x_1 = x_2$.

  \begin{proof}[Proof of Lemma~\ref{lem:neither-const-nor-neg}]
    First, suppose $R \in \NN^r$ is not constant and 
    let $\ba = (\ba_1, \dots, \ba_r)$ be a least refined tuple in $R$,
    i.e. $\ba$ does not strictly refine any tuple in $R$.
    Since $R$ is not constant, the number of 
    distinct entries in $\ba$ is at least two.
    By permuting indices, assume that $\ba_1 \neq \ba_2$ and
    consider the constraint $R(x_{\ba_1}, \dots, x_{\ba_r})$.
    For example, if $\ba = (1,2,2,3,2,3)$, then we consider
    $R(x_1, x_2, x_2, x_3, x_2, x_3)$.
    Suppose $\bb$ is a tuple of values that satisfies the constraint.
    Observe that $\ba_i = \ba_j$ implies $\bb_i = \bb_j$ 
    since $x_{\ba_i}$ and $x_{\ba_j}$ denote the same variable.
    Hence, $\ba$ refines $\bb$.
    By the choice of $\ba$, the refinement is not strict,
    hence $\ba_1 \neq \ba_2$ implies $\bb_1 \neq \bb_2$.

    Now suppose $R$ is Horn and not strictly negative,
    and $\phi_R$ is a reduced CNF definition of $R$.
    Since $R$ is not strictly negative,
    $\phi_R$ contains a clause $C$ with a positive literal.
    By permuting indices, assume $C$ is
    $(x_1 = x_2 \lor \bigvee_{s=1}^{t} x_{i_s} \neq x_{j_s})$.
    By Observation~\ref{obs:minimal-formula},
    there is a tuple $\ba \in R$
    with $\ba_1 = \ba_2$ and $\ba_{i_s} = \ba_{j_s}$ for all $1 \leq s \leq t$.
    Consider an instance $R(\bx)$,
    where $\bx$ is a tuple of variables
    such that $\bx_{i_s} = \bx_{j_s}$ for all $1 \leq s \leq t$,
    while all other variables are distinct.
    Let $\bx_1,\bx_2$ be the primary variables,
    and all remaining variables be auxiliary.
    Since $\ba \in R$, constraint $R(\bx)$ is consistent.
    Moreover, by identifying $\bx_{i_s}$ and $\bx_{j_s}$ for all $s$,
    we falsify every literal in $C$ except for $\bx_1 = \bx_2$,
    hence $R(\bx)$ implies $\bx_1 = \bx_2$.
  \end{proof}

  We use a consequence of Proposition~68~of~\cite{BodirskyCP10equality}.

  \begin{proposition} \label{prop:projection}
    Let $R$ be a negative relation of arity $r$,
    and let $1 \leq i_1, \dots, i_t \leq r$ be a set of indices.
    The projection of $R$ onto indices $i_1, \dots, i_r$
    is a negative relation.
  \end{proposition}

  \begin{proof}[Proof of Lemma~\ref{lem:not-conj-nae-3-or-disjneqneq}]
      Let $\phi_R$ be a CNF definition of $R$
      with the minimum number of literals.
      Then $\phi_R$ contains a clause
      $C = \bigvee_{s=1}^{t} x_{i_s} \neq x_{j_s}$ with $t \geq 2$.
      Define formula $\phi' = \phi_R \land \bigwedge_{s=3}^{t} (x_{i_s} = x_{j_s})$,
      and relation $R'$ obtained by projecting all tuples satisfying $\phi'$ onto $i_1, j_1, i_2, j_2$.
      By Proposition~\ref{prop:projection}, $R'$ is essentially negative.
      Note that $\phi'$ implies $x_{i_1} \neq x_{j_1} \lor x_{i_2} \neq x_{j_2}$.
      Furthermore, $\phi_R$ implies 
      $C' = R'(x_{i_1}, x_{i_2}, x_{j_1}, x_{j_2}) \lor \bigvee_{s=3}^{t} x_{i_s} \neq x_{j_s}$.
      By minimality of $\phi_R$, no clause of the formula $C'$ subsumes the clause $C$,
      so $R'(x_{i_1}, x_{i_2}, x_{j_1}, x_{j_2})$ implies 
      neither $x_{i_1} \neq x_{j_1}$ nor $x_{i_2} \neq x_{j_2}$.
      We proceed with two cases based on 
      the cardinality of $\{i_1, j_1, i_2, j_2\}$.

      If $|\{i_1, j_1, i_2, j_2\}| = 3$, then $R'$ is an essentially ternary relation.
      Without loss of generality, assume $j_1 = i_2$ and note that $(1, 1, 1) \notin R'$.
      Since $R'$ is negative, $R'(x_1, x_2, x_3)$ implies 
      $x_1 \neq x_2$, $x_2 \neq x_3$, or $\rel{NAE}_3(x_1, x_2, x_3)$.
      The first two formulas are ruled out by minimality of $\phi_R$, 
      hence $R = \rel{NAE}_3$.

      If $|\{i_1, j_1, i_2, j_2\}| = 4$, let indices
      $p$ and $q$ range over $\{i_1, j_1\}$ and $\{i_2, j_2\}$, respectively.
      If $\phi_R$ implies $x_p = x_q$ for some $p$ and $q$,
      then we reduce to the previous case since 
      $R'$ is an essentially ternary relation.
      Otherwise, note that $(1, 1, 2, 2) \notin R'$,
      $R'$ is negative and $\phi_R$ does not imply $x_p = x_q$ for any $p, q$,
      hence $R'(x_1, x_2, x_3, x_4)$ implies $x_1 \neq x_2$,
      $x_3 \neq x_4$ or $x_1 \neq x_2 \lor x_3 \neq x_4$.
      The first two formulas are ruled out by minimality of $\phi_R$.
      Thus, $R(x_1, x_2, x_3, x_4) \land \bigwedge_{p,q} (x_p \neq x_q)$
      is a pp-definition of $R^{\lor}_{\neq,\neq}$.
   \end{proof}

\fi
\iflong
  
  \subsection{Hardness from Split Paired Cut}
  \label{ssec:split-paired-cut}
  
  We start with the following problem.
  
  \pbDefP{Split Paired Cut}
  {Graphs $G_1, G_2$, vertices $s_1, t_1 \in V(G_1)$, $s_2, t_2 \in V(G_2)$,
    a family of disjoint edge pairs $\cP \subseteq E(G_1) \times E(G_2)$,
    and an integer $k$.}
  {$k$.}
  {Is there a subset $X \subseteq \cP$ of size at most $k$ 
    such that for both $i \in \{1,2\}$,
    $\{e_i : \{e_1, e_2\} \in X \}$ is an $st$-cut in $G_i$?}
  
  \spc is W[1]-hard~(see Lemma~6.1~in~\cite{dabrowski2023almost}).
  There is a simple reduction
  from \spc to \mincsp{R, =, \neq} where
  $R$ is a $(=,=)$-relation.
  
  \begin{proof}[Proof of Lemma~\ref{lem:eq-eq-hard}]
    Let $(G_1, G_2, s_1, t_1, s_2, t_2, \cP, k)$ be an instance of \spc.
    Construct an instance $(I, k)$ of \mincsp{R, =, \neq} as follows.
    Let $V(I) = V(G_1) \cup V(G_2)$.
    For every edge $uv \in E(G_1) \cup E(G_2)$,
    add crisp constraint $u = v$ to $I$.
    Add crisp constraints $s_1 \neq t_1$ and $s_2 \neq t_2$ to $I$.
    Finally, pair up equality constraints according to
    the pairs in $\cP$.
    For every pair $\{u_1 v_1, u_2 v_2\} \in \cP$,
    remove constraints $u_1 = v_1$ and $u_2 = v_2$ from $I$
    and add a soft constraint $R(u_1, v_1, u_2, v_2)$.
    This completes the construction.
    Note that all unpaired edges correspond to crisp equality constraints in $I$.
    We proceed with the correctness proof.

    For one direction, assume $X \subseteq \cP$ is 
    a solution to $(G_1, G_2, s_1, t_1, s_2, t_2, \cP, k)$.
    Define $X' \subseteq C(I)$ that contains
    $R(u_1, v_1, u_2, v_2)$ for every pair $\{u_1 v_1, u_2 v_2\}$ in $X$.
    Note that $|X'| = |X| \leq k$.
    We claim that $I - X'$ is satisfied by assignment $\alpha$
    defined as follows.
    Let $\alpha(s_1) = 1$, $\alpha(t_1) = 2$,
    $\alpha(s_2) = 4$ and $\alpha(t_2) = 5$.
    Propagate these values to variables connected to
    $s_1, t_1, s_2, t_2$ by equality constraints;
    for the remaining variables $v$, set
    $\alpha(v) = 3$ if $v \in V(G_1)$ and 
    $\alpha(v) = 6$ if $v \in V(G_2)$.
    Since $\bigcup X$ is an $s_1 t_1$-cut and an $s_2 t_2$-cut,
    assignment $\alpha$ is well-defined.
    It satisfies $\alpha(v_1) \neq \alpha(v_2)$ for all
    $v_1 \in V(G_1)$ and $v_2 \in V(G_2)$.
    Furthermore, it satisfies crisp constraints 
    $s_1 \neq t_1$ and $s_2 \neq t_2$,
    hence $I - X'$ is consistent.

    For the other direction, let $Z$ be a solution to $(I, k)$.
    By construction, only $R$-constraints are soft in $I$,
    hence $Z$ only contains $R$-constraints.
    Define $Z' \subseteq \cP$ containing
    $\{u_1 v_1, u_2 v_2\}$ for all $R(u_1, v_1, u_2, v_2)$ in $Z$.
    Note that $|Z'| = |Z| \leq k$.
    Since $s_1 \neq t_1$ and $s_2 \neq t_2$ are crisp in $I$,
    $s_i$ and $t_i$ are not connected by equality constraints in $I - Z$,
    thus $\bigcup Z'$ is an $s_i t_i$-cut in $G_i$,
    and it is a union of $k$ pairs in $\cP$ by definition.
    Hence, $Z'$ is a solution to the instance of \spc.
  \end{proof}
\fi
\iflong
  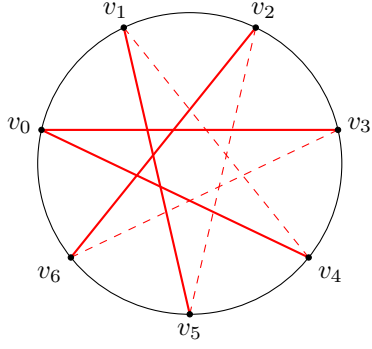
\begin{figure}
    \centering
    \begin{tikzpicture}
  \def \radius {2}
  \def \step {-360/7}

  \draw (0,0) circle (\radius);
  
  \coordinate (v5) at ({\radius*cos(270 + \step*0)},{\radius*sin(270 + \step*0)});
  \coordinate (v6) at ({\radius*cos(270 + \step*1)},{\radius*sin(270 + \step*1)});
  \coordinate (v0) at ({\radius*cos(270 + \step*2)},{\radius*sin(270 + \step*2)});
  \coordinate (v1) at ({\radius*cos(270 + \step*3)},{\radius*sin(270 + \step*3)});
  \coordinate (v2) at ({\radius*cos(270 + \step*4)},{\radius*sin(270 + \step*4)});
  \coordinate (v3) at ({\radius*cos(270 + \step*5)},{\radius*sin(270 + \step*5)});
  \coordinate (v4) at ({\radius*cos(270 + \step*6)},{\radius*sin(270 + \step*6)});

  \draw[red, thick]  (v0) -- (v3);
  \draw[red, dashed] (v1) -- (v4);
  \draw[red, dashed] (v2) -- (v5);
  \draw[red, dashed] (v3) -- (v6);
  \draw[red, thick]  (v4) -- (v0);
  \draw[red, thick]  (v5) -- (v1);
  \draw[red, thick]  (v6) -- (v2);

  \filldraw[black] (v5) circle (1pt) node[anchor={-270}] {$v_5$};
  \filldraw[black] (v6) circle (1pt) node[anchor={-270+\step*1}] {$v_6$};
  \filldraw[black] (v0) circle (1pt) node[anchor={-270+\step*2}] {$v_0$};  
  \filldraw[black] (v1) circle (1pt) node[anchor={-270+\step*3}] {$v_1$};  
  \filldraw[black] (v2) circle (1pt) node[anchor={-270+\step*4}] {$v_2$};  
  \filldraw[black] (v3) circle (1pt) node[anchor={-270+\step*5}] {$v_3$};  
  \filldraw[black] (v4) circle (1pt) node[anchor={-270+\step*6}] {$v_4$};  

\end{tikzpicture}
    \caption{An illustration of a choice gadget for $t = 3$.
      Black arcs represent equality constraints of cost 2,
      dashed red edges -- soft disequality constraints, and
      bold red edges -- crisp disequality constraints.}
    \label{fig:wheel}
  \end{figure}
  
  Further reductions in this section share a choice gadget.
  Let $S = \{s_1, \dots, s_t\}$ be a set.
  Define an instance $W(S)$ of $\csp{=,\neq}$ as follows.
  Introduce $2t+1$ variables $v_0, \dots, v_{2t}$.
  In what follows, indices are identified modulo $2t+1$, e.g. $v_0 = v_{2t+1}$.
  Connect variables in a double-cycle of equalities, i.e.
  add two copies of soft constraint $v_i = v_{i+1}$ for all $0 \leq i \leq 2t$.
  We regard $v_i = v_{i+1}$ as single constraints as having cost two.
  The \emph{forward partner} of a variable $v_i$ is $f(v_i) := v_{i+t}$,
  i.e. the variable that is $t$ steps ahead of $v_i$ on the cycle.
  Add constraints $v_i \neq f(v_i)$ for all $0 \leq i \leq 2t$,
  making them soft if $1 \leq i \leq t$ and crisp otherwise.
  Note that $v_0 \neq v_t$ is crisp.
  See Figure~\ref{fig:wheel} for an illustration.
  
  \begin{lemma} \label{lem:wheel-gadget}
    Let $S$ be a set of size at least two and $W(S)$ be the choice gadget.
    Then $\cost(W(S)) = 5$ and every optimal solution deletes
    $v_{i-1} = v_{i}$, $v_{i} \neq f(v_i)$ and $f(v_{i}) = f(v_{i+1})$
    for some $i \in \range{1}{t}$.
  \end{lemma}
  \begin{proof}
    First, we claim that every optimal solution
    consists of two equality constraints and a disequality constraint.
    Note that $v_0$ and $v_t$ are connected
    by two disjoint paths of equality constraints, so
    one constraint of cost two has to be deleted from each. 
    After this, the cycle is split into two paths,
    and the longest of them has at least
    $\ceil{\frac{2t + 1 - 2}{2}} = t$ edges.
    Then there is a pair of variables $v_i$ and $f(v_i)$
    still connected on the longer path.
    To see that deleting two equalities and a disequality suffices,
    note that $\{v_{i-1} = v_{i}, v_{i} \neq f(v_i), f(v_{i}) = f(v_{i+1}) \}$
    for any $1 \leq i \leq t$ is a solution.

    To show every optimal solution is of the form above,
    observe that the crisp constraint $v_0 \neq v_t$ implies that
    every solution has to delete
    $v_{j-1} = v_{j}$ for some $1 \leq j \leq t$.
    By construction, there are crisp constraints
    $f(v_{j}) \neq v_{j-1}$ and $f(v_{j+1}) \neq v_{j}$ in $W(S)$, and
    there are paths of equality constraints connecting
    $f(v_{j}), v_{j-1}$ and $f(v_{j+1}), v_{j}$ in $W(S) - \{v_{j-1} = v_{j}\}$
    which intersect only in $f(v_j) = f(v_{j+1})$, so 
    this constraint has to be deleted.
    Now $W(S) - \{ v_{j-1} = v_j, f(v_{j}) = f(v_{j+1}) \}$
    contains a path connecting $v_j$ and $f(v_j)$,
    and the remaining budget of one is only sufficient
    to delete the soft constraint $v_j \neq f(v_j)$.
  \end{proof}

  We interpret deleting $v_{i-1} = v_{i}$, $v_{i} \neq f(v_i)$ and $f(v_{i}) = f(v_{i+1})$ 
  from $W(S)$ as choosing element $s_i$ from the set $S$.
  For the next proofs, we remark that by
  the construction in Lemma~5.7~of~\cite{kim2020solving},
  we may assume that graphs $(G_1, G_2)$ in an instance of \spc 
  come with two maxflows $\cF_1$ and $\cF_2$ that partition
  $E(G_1)$ and $E(G_2)$, respectively,
  into $k$ pairwise edge-disjoint paths.
  We are now ready to show that $\mincsp{R, =, \neq}$ is W[1]-hard
  if $R$ is a $(\neq, \neq)$-relation.

  \begin{proof}[Proof of Lemma~\ref{lem:neq-neq-hard}]
    Let $(G_1, G_2, s_1, t_1, s_2, t_2, \cP, k)$ 
    be an instance of \spc.
    Assume $k = 2\ell$, and $\cF_i$ for $i \in \{1,2\}$ are 
    $s_i t_i$-maxflows in $G_i$ partitioning 
    $E(G_i)$ into $k$ pairwise edge-disjoint paths.
    Construct an instance $(I, k')$ of \mincsp{R, =, \neq} 
    with $k'= 9\ell$ as follows.
    Start by creating a variable for every vertex in $V(G_1) \cup V(G_2)$ with the same name.
    For each $i \in \{1,2\}$, consider a path $P \in \cF_i$, and 
    let $p$ be the number of edges on $P$.
    Create a choice gadget $W(P)$ for every $P$ with variables 
    $v^{P}_0, \dots, v^{P}_p$ following the path, and
    fresh variables $v^{P}_{j}$ for $p < j \leq 2p$ added to the instance.
    Observe that variables may appear on several paths in $\cF_i$.
    In particular, $v^P_{0} = s_i$ and $v^P_{p} = t_i$ for every $P \in \cF_i$,
    so we have crisp constraints $s_i \neq t_i$.
    Furthermore, since $\cF_i$ partitions $E(G_i)$, 
    the construction contains a copy of graphs $G_1$ and $G_2$
    with equality constraints for edges.
    Now we pair up edges according to $\cP$.
    For every pair $\{e_1, e_2\} \in \cP$,
    let $P \in \cF_1$ and $Q \in \cF_2$ be the paths such that $e_1 \in P$ and $e_2 \in Q$, and 
    suppose $e_1 = v^{P}_{i-1} v^{P}_{i}$ and $e_2 = v^{Q}_{j-1} v^{Q}_{j}$.
    Pair up soft constraints $v^{P}_i \neq f(v^{P}_i)$ and $v^{Q}_j \neq f(v^{Q}_j)$,
    i.e. replace individual constraints with one soft constraint
    \[ R(v^{P}_{i}, f(v^{P}), v^{Q}_{j}, f(v^{Q}_{j})). \]
    Finally, if an edge $uv \in E(G)$ does not appear in any pair of $\cP$,
    make constraint $u = v$ crisp in $I$.
    This completes the construction.

    For one direction, suppose $X \subseteq \cP$ is a solution to 
    $(G_1, G_2, s_1, t_1, s_2, t_2, \cP, k)$.
    Define $X' \subseteq C(I)$ as follows.
    For every pair $\{e_1, e_2\} \in X$,
    let $P \in \cF_1$ and $Q \in \cF_2$ be the paths such that $e_1 \in P$ and $e_2 \in Q$, and 
    suppose $e_1 = v^{P}_{i-1} v^{P}_{i}$ and $e_2 = v^{Q}_{j-1}, v^{Q}_{j}$.
    Add constraints $v^{P}_{i-1} = v^{P}_{i}$, $f(v^{P}_{i}) = f(v^{P}_{i+1})$,
    $v^{Q}_{j-1} = v^{Q}_{j}$, $f(v^{Q}_{j}) = f(v^{Q}_{j+1})$ and
    $R(v^{P}_{i}, f(v^{P}_i), v^{Q}_{j} f(v^{Q}_{j}))$ to $X'$.
    Note that $X'$ only contains soft constraints,
    and $|X| = 2\ell$ implies $|X'| = (2 \cdot 4 + 1)\ell = 9\ell$.
    Moreover, by construction of $X'$ and Lemma~\ref{lem:wheel-gadget}, 
    instances $W(P) - X'$ are consistent for every $P \in \cF_1 \cup \cF_2$.
    We claim that $I - X'$ is consistent.
    Consider the graph $G' = G - \bigcup X$.
    Observe that $X_i = \{e_i : \{e_1, e_2\} \in X\}$ is an $s_i t_i$-cut,
    $|X_i| = |X| = 2\ell$ and $s_i t_i$-maxflow in $G_i$ equals $2\ell$. 
    Hence, $G'$ contains exactly four components which contain
    $s_1$, $t_1$, $s_2$ and $t_2$, respectively.
    Define assignment that maps the components of $G'$ to distinct values.
    This assignment clearly satisfies all crisp constraints.
    It also satisfies constraints constraints $R(v^{P}_{i}, f(v^{P}_i), v^{Q}_{j} f(v^{Q}_{j}))$
    that remain in $I - X'$ for any $(\neq,\neq)$-relation $R$,
    since it satisfies 
    $v^{P}_{i} \neq v^{Q}_{j}$, 
    $f(v^{P}_{i}) \neq v^{Q}_{j}$,
    $v^{P}_{i} \neq f(v^{Q}_{j})$ and
    $f(v^{P}_{i}) \neq f(v^{Q}_{j})$.
    We conclude that $I - X'$ is consistent.

    For the opposite direction,
    suppose $Z \subseteq C(I)$ is a solution to $(I, k')$.
    Define the set of edges $Z' \subseteq E(G)$
    with an edge $v^{P}_{i-1}v^{P}_{i}$ for every constraint $v^{P}_{i} \neq f(v^{P}_{i})$ in $Z$.
    We claim that $Z'$ is a union of $\ell$ pairs.
    We have $2\ell$ gadgets $W(P)$ in $I$ for each path $P \in \cF_1 \cup \cF_2$,
    so by Lemma~\ref{lem:wheel-gadget}, 
    $Z$ contains two equality constraints from every gadget $W(P)$,
    and deleting each costs two.
    Since the gadgets are constraint-disjoint,
    this amounts to $8\ell$ equality constraints in total.
    The remaining budget is $\ell$,
    so the remaining constraints in $Z$ are $(\neq,\neq)$-constraints.
    By construction,
    $Z'$ is then a union of $\ell$ pairs.
    Finally, note that $\{e_i : \{e_1, e_2\} \in Z'\}$ are 
    $s_i t_i$-cuts because $s_i \neq t_i$ are crisp constraints in $I$,
    hence $Z'$ is a solution to the instance of \spc.
  \end{proof}
\fi
\iflong
  Now we prove that
  \mincsp{R, =, \neq} for an $(=,\neq)$-relation $R$
  is W[1]-hard by providing a hybrid of the previous two reductions.

  \begin{proof}[Proof of Lemma~\ref{lem:eq-neq-hard}]
    Let $(G_1, G_2, s_1, t_1, s_2, t_2, \cP, k)$ 
    be an instance of \spc.
    Assume $k = 2\ell$, and $\cF_2$ is an 
    $s_2 t_2$-maxflow in $G_2$ that partitions
    $E(G_2)$ into $k$ pairwise edge-disjoint paths.
    Construct an instance $(I, k')$ of \mincsp{R, =, \neq} as follows.
    Start by creating a variable for every vertex in $V(G_1) \cup V(G_2)$ with the same name.
    For every edge $uw \in E(G_1)$, add constraint $u = w$ in $I$.
    For every path $P \in \cF_2$ with $p$ edges,
    create a choice gadget $W(P)$ with variables 
    $v^{P}_0, \dots, v^{P}_p$ following the path, and
    fresh variables $v^{P}_{j}$ for $p < j \leq 2p$ added to $I$.
    Now we pair up edges according to $\cP$.
    For every pair $\{e_1, e_2\} \in \cP$,
    let $e_1 = uv$ and $e_2 = v^{P}_{j-1} v^{P}_{j}$,
    where $P$ is the path in $\cF_2$ such that $e_2 \in P$.
    Pair up constraints $u = w$ and $v^{P}_j \neq f(v^{P}_j)$,
    i.e. remove individual constraints and add one soft constraint
    $R(u, w, v^{Q}_{j}, f(v^{Q}_{j}))$.
    Finally, if an edge $uv \in E(G)$ does not appear in any pair of $\cP$,
    make constraint $u = v$ crisp in $I$,
    and set $k' = 5\ell$.
    This completes the construction.

    Correctness proof is analogous to the proofs of
    Lemmas~\ref{lem:eq-eq-hard}~and~\ref{lem:neq-neq-hard}
    and follows from the observation that
    satisfying all choice gadgets $W(P)$ for $P \in \cF_2$
    requires deleting $5\ell$ constraints,
    and the choices exactly determine
    which $R$-constraints are picked.
  \end{proof}

  Finally, we prove that $\mincsp{R^{\lor}_{\neq, \neq}, =, \neq}$ is W[1]-hard
  by reduction from the following problem.

  \pbDefP{Multicoloured Independent Set (MIS)}
  {Graph $G$, partition of $V(G) = V_1 \uplus \dots \uplus V_k$, and integer $k$.}
  {$k$.}
  {Is there an independent set in $G$ with one vertex from each $V_i$?}

  \begin{proof}[Proof of Lemma~\ref{lem:disjneqneq-hard}]
    Let $(G, V_1 \uplus \dots \uplus, V_k, k)$ be an instance of \textsc{MIS}.
    Enumerate vertices in each set $V_i$.
    Construct an instance $(I, k')$ of \mincsp{R^{\lor}_{\neq, \neq}, =, \neq} as follows.
    Create a choice gadget $W_i = W(V_i)$ for each $i \in [k]$.
    Let the variables in the gadget be $x^i_0, \dots, x^i_{2|V_i|}$.
    Add crisp disequality constraint $x^i_0 \neq x^i_{|V_i|}$.
    For every edge $uv \in E(G)$, assume
    $u$ is vertex $j$ of $V_i$,
    $v$ is vertex $j'$ of $V_{i'}$, and
    add a crisp constraint
    \[
      R^\lor_{\neq, \neq}(x^i_j, f(x^i_j), x^{i'}_{j'}, f(x^{i'}_{j'})). 
    \]
    Finally, set the budget to $k' = 5k$.
    This completes the reduction.

    For one direction, suppose $X \subseteq V(G)$ is
    a solution to $(G, V_1 \uplus \dots \uplus, V_k, k)$.
    Construct a subset $X'$ of constraints in $I$
    by adding $x^i_{j-1} = x^i_j$, $f(x^i_j) = f(x^i_{j+1})$ and
    $x^{i}_j \neq f(x^{i}_j)$ to $X'$
    whenever $X$ contains vertex $j$ from $V_i$.
    Note that $|X| = k$ implies $|X'| = 4k$
    because equalities are present in two copies.
    Since $X'$ separates $x^i_0$ and $x^i_{|V_i|}$
    in every gadget, we can define an assignment
    that maps all vertices in $\{x^i_0$, $x^i_{|V_i|} : i \in [k]\}$
    to distinct values, and propagate through equality constraints.
    We claim that this assignment satisfies $I - X'$.
    Observe that, by the choice of $X'$,
    equality constraints in $W_i - X'$ imply $x^i_j = f(x^i_j)$ for exactly one value of $j$.
    Suppose for contradiction that a crisp constraint
    $R^\lor_{\neq, \neq}(x^i_j, f(x^i_j), x^{i'}_{j'}, f(x^{i'}_{j'}))$ is violated.
    Then we have $x^{i}_{j} = f(x^{i}_{j})$ and $x^{i'}_{j'} = f(x^{i'}_{j'})$,
    hence vertex $j$ of $V_i$ and vertex $j'$ of $V_{i'}$ are in $X$.
    However, by construction of $I$,
    these vertices are connected by an edge in $G$,
    contradicting the fact that $X$ is an independent set.

    For the opposite direction, let $Z \subseteq C(I)$
    be a solution to $(I, k')$.
    By Lemma~\ref{lem:wheel-gadget},
    after an optimal solution is deleted from $W(P)$,
    there remains exactly one pair $x^i_j$ and $f(x^i_j)$ connected 
    by a path of equalities in $I - Z$.
    Construct $Z' \subseteq V(G)$ by adding
    vertex $j$ of $V_i$ to $Z'$ if $x^i_j, f(x^i_j)$ is picked.
    Note that $|Z| = 5k$ implies $|Z'| = k$ by construction.
    It remains to show that $Z'$ is an independent set.
    For the sake of contradiction, suppose $u, v \in Z'$,
    $u$ is vertex $j'$ of $V_{i'}$,
    $v$ is vertex $j$ of $V_i$,
    and $uv \in E(G)$.
    By construction, $I$ contains a crisp constraint
    $R^\lor_{\neq, \neq}(x^i_j, f(x^i_j), x^{i'}_{j'}, f(x^{i'}_{j'}))$.
    However, $I - Z$ implies
    $x^{i}_{j} = f(x^{i}_{j})$ and $x^{i'}_{j'} = f(x^{i'}_{j'})$,
    so we arrive at a contradiction.
  \end{proof}
\fi

\iflong
  \subsection[Hardness of ODD3 and NAE3]{Hardness for $\rel{ODD}_3$ and $\rel{NAE}_3$}
  \label{ssec:odd3-nae3}

  We prove that if an equality constraint language $\Gamma$
  pp-defines $\rel{ODD}_3$ or $\rel{NAE}_3$,
  then $\mincsp{\Gamma, =, \neq}$ is W[2]- and W[1]-hard, respectively.
  The reductions are from \textsc{Hitting Set} and \textsc{Steiner Multicut}.
  We start with a cost-preserving reduction
  that takes an instance of \textsc{Hitting Set}
  and produces in polynomial time an instance
  of \mincsp{\rel{ODD}_3, =, \neq} with strict
  $\rel{ODD}_3$-constraints.

  \begin{proof}[Proof of Lemma~\ref{lem:hitting-set-to-odd3}]
    Let $(V, \cE, k)$ be an instance of \textsc{Hitting Set}.
    Assume $V = \{1,\dots,n\}$. 
    Construct an instance $(I, k)$ of $\mincsp{\rel{ODD}_3, =, \neq}$ as follows.
    First, introduce variables $x_1,\dots,x_n$ and $z$, and 
    add soft constraints $x_i = z$ for all $i \in [n]$.
    For every subset $e = \{a_1, \dots, a_{\ell}\} \in \cE$,
    introduce auxiliary variables $y_2, \dots, y_{\ell}$
    and the following crisp constraints:
    \begin{enumerate}
    \item $\rel{ODD}_3(x_{a_1}, x_{a_2}, y_2)$,
    \item $\rel{ODD}_3(y_{i-1}, x_{a_i}, y_{i})$ for all $3 \leq i \leq \ell$, and
    \item $x_{a_1} \neq y_\ell$.
    \end{enumerate}
    This completes the reduction.

    To show correctness, first assume $X$ is a solution to $(I, k)$.
    Define a subset $X' \subseteq [n]$ containing all indices $i \in [n]$
    such that $x_i = z$ is in $X$.
    We claim that $X'$ intersects every set in $\cE$.
    For the sake of contradiction, assume $e \cap X' = \emptyset$ for some subset 
    $e = \{a_1, \dots, a_{\ell}\} \in \cE$.
    Then $I - X$ contains constraints $x_{a_i} = z$ for all $i \in [\ell]$,
    implying that $x_{a_i} = x_{a_j}$ for all $a_i,a_j \in e$.
    Now consider the crisp constraints introduced for $e$ in $I$.
    Constraint $\rel{ODD}_3(x_{a_1}, x_{a_2}, y_2)$ together with $x_{a_1} = x_{a_2} = z$
    implies that $y_{2} = z$.
    Constraints $\rel{ODD}_3(y_{i-1}, x_{a_i}, y_{i})$
    imply that $y_{i} = z$ for $3 \leq i \leq \ell$.
    However, then $x_{a_1} = z = y_\ell$, which is a contradiction.

    For the other direction, assume $Z \subseteq [n]$ intersects every subset in $\cE$.
    By renaming elements, we may assume that $Z = \{1,\dots,k'\}$ for some $k' \leq k$.
    We define an assignment $\alpha_Z : V(I) \to \NN$ as follows.
    First, let $\alpha_Z(x_i) = i$ if $1 \leq i \leq k'$, $\alpha_Z(x_j) = 0$ if $j > k'$,
    and $\alpha_Z(z) = 0$.
    By definition, $\alpha_Z$ satisfies all but $k'$ soft constraints $x_i = z$.
    Now we extend $\alpha_Z$ to auxiliary variables so that it satisfies all crisp constraints. 
    Consider $e = \{a_1, \dots, a_\ell\} \in \cE$.
    Since $Z$ intersects $e$, there is an index $1 \leq i < \ell$ 
    such that $\alpha_Z(x_{a_i}) \neq \alpha_Z(x_{a_{i+1}})$.
    Let $i$ be the minimal such index, and notice that
    $\alpha_Z(x_1) = \dots = \alpha_Z(x_i)$.
    Set $\alpha_Z(y_j) = \alpha_Z(x_j)$ for $1 \leq j \leq i$.
    For variables $y_j$ with $j > i$, we choose values that are
    pairwise distinct and different from those assigned to $x_1,\dots,x_n$,
    for instance $\alpha_Z(y_j) = k + j$.
    To check that $\alpha_Z$ satisfies all crisp constraints
    corresponding to subset $e$, 
    note that the constraints using $\rel{ODD}_3$ and 
    containing variables $y_j$ for $1 \leq j \leq i$ in the scope
    are satisfied because all variables are assigned to the same value,
    while in the remaining constraints all variables are assigned distinct values.
    Hence, the reduction is correct.
  \end{proof}

  Now we show that if an equality constraint language 
  $\Gamma$ pp-defines $\rel{NAE}_3$, 
  then there is a reduction from \textsc{Steiner Multicut} to 
  $\mincsp{\Gamma, =, \neq}$.

  \begin{proof}[Proof of Lemma~\ref{lem:nae3-hard}]
    Let $(G, \cT, k)$ be an instance of
    \textsc{(Edge) Steiner Multicut},
    where $\cT = (T_1, \dots, T_k)$ and $|T_i| = 3$ for all $i$.
    Create an instance $(I, k)$ of $\mincsp{\rel{NAE}_3, =}$ as follows.
    Introduce a variable for every vertex in $V(G)$.
    Add a soft binary equality constraint $u = v$ for every edge $\{u,v\} \in E(G)$, and
    a crisp constraint $\rel{NAE}_3(x_i,y_i,z_i)$ for every subset $T_i = \{x_i,y_i,z_i\}$.
    This completes the reduction.

    To argue correctness, assume first that $Z$ is a solution to $(I, k)$.
    Note that the only soft constraints in $I$ are binary equalities,
    so the set $S_Z$ of edges in $G$ corresponding to constraints in $Z$ is well-defined.
    We claim that $S_Z$ is a solution for $(G, \cT, k)$.
    Consider a triple $\{x,y,z\} \in \cT$.
    If all vertices $x,y,z$ are in the same connected component of $G - S_Z$,
    then binary equality constraints in $I - Z$ force $x$, $y$ and $z$
    to take the same value, violating a crisp constraint $\rel{NAE}_3(x,y,z)$,
    which leads to a contradiction.

    For the other direction, assume $S$ is a solution to $(G, \cT, k)$,
    and $Z_S$ is the set of corresponding binary equality constraints in $I$.
    Consider an assignment 
    $\alpha : V(G) \to \NN$ that is constant on the connected
    components of $G - S$ and assigns distinct values to distinct components.
    We claim that $\alpha$ satisfies $I - Z_S$, therefore $I - Z_S$ is consistent.
    All binary equality constraints in $I - Z_S$ are satisfied by $\alpha$
    since the assigned value is the same for any pair of connected variables.
    All constraints $\rel{NAE}_3(x,y,z)$ are satisfied because $S$ separates
    at least one pair of variables in $\{x,y,z\}$, and variables
    in that pair are assigned different values by $\alpha$.
    Thus, the reduction is correct.
  \end{proof}

  \subsection{Reductions and Multicut Variants}
  \label{ssec:multicut-variants}

  We start by showing that if $=$ and $\neq$ are available in 
  an equality constraint language $\Gamma$, 
  then \textsc{Edge Multicut} reduces to \mincsp{\Gamma}.

  \begin{proof}[Proof of Lemma~\ref{lem:multicut-to-mincsp}]
    By Lemma~\ref{lem:neither-const-nor-neg}, $\Gamma$ implements both $=$ and $\neq$.
    Let $(G, \req, k)$ be an arbitrary instance of \textsc{Edge Multicut},
    and create an instance $(I, k)$ of \mincsp{=,\neq} 
    with $V(G)$ as the set of variables, 
    soft constraints $u = v$ for every edge $uv \in E(G)$, and 
    crisp constraints $s \neq t$ for every cut request $st \in \req$.
    Clearly, the reduction requires polynomial time,
    and leaves the parameter is unchanged.
    If $X$ is a solution to $(G, \req, k)$,
    let $X' = \{ u = v : uv \in X \}$ and observe that
    an assignment mapping distinct connected components of $G - X$
    to distinct values satisfies $I - X'$.
    On the other hand, if $Z$ is a solution to $I$,
    let $Z' = \{ uv : u = v \in Z \}$ and suppose there is 
    path in $G - Z'$ contains no path connecting a terminal pair $st$.
    Then $I - Z$ contains a path of $=$-constraints connecting $s$ and $t$,
    and a crisp constraint $s \neq t$, contradicting that $I - Z$ is consistent. 
  \end{proof}

  For algorithmic purposes,
  we need reduction from \textsc{MinCSP} to
  several variants of \textsc{Vertex Multicut}.
  The first one reduces \mincsp{\Gamma}
  with split and $\rel{NEQ}_3$ relations to \mdt.

  \begin{proof}[Proof of Lemma~\ref{lem:split-nae3-to-mdt}]
    Enumerate variables in $V(I)$ as $x_1, \dots, x_n$.
    Introduce a vertex $v_i$ in $G$ for every variable $x_i \in V(I)$.
    Introduce a dummy vertex $w$.
    Consider a constraint $c \in C(I)$ that uses a split relation of arity $r$,
    and let $P \uplus Q$ be the partition of $\{1,\dots,r\}$.
    Introduce vertex $z_c$ in $G$,
    connect it by edges to $x_p$ for all $p \in P$,
    and add triples $z_c v_q w$ to $\cT$ for all $q \in Q$.
    Note that $w$ is disconnected from all other vertices,
    so the triple only requests to separate $z_c$ and $x_q$.
    For every constraint 
    of the form $\rel{NEQ}_3(x_h, x_i, x_j)$ in $I$,
    add triple $v_h v_i v_j$ to $\cT$.
    Finally, make vertices $v_i$ undeletable by replacing 
    each $v_i$ with copies $v_i^{(1)}, \dots, v_i^{(k+1)}$
    that have the same neighbourhood $N(v_i)$.
    To avoid cumbersome notation, we will regard the copies 
    $v_i^{(1)}, \dots, v_i^{(k+1)}$ as a single undeletable vertex $v_i$
    in the remainder of the proof.
    Observe that, by construction,
    vertices $v_i$ and $v_j$ are connected in $G$
    if and only if the constraints in $I$ imply $x_i = x_j$.

    For one direction, assume $(I, k)$ is a yes-instance, and let $X$ be a solution.
    Let $Z_V$ contain vertices $z_c$ for all constraints $c \in X$
    that use a split relation,
    and let $Z_{\cT}$ contains triples $v_h v_i v_j$ for all
    constraints of the form $\rel{NEQ}_3(v_h, v_i, v_j)$ in $X$.
    Clearly, $|Z_V| + |Z_{\cT}| \leq |X| \leq k$.
    We claim that $(Z_V, Z_{\cT})$ is a solution to $(G, \cT, k)$, i.e.
    every connected component of $G - Z_V$ intersects
    every triple in $\cT \setminus Z_{\cT}$ in at most one vertex.
    Suppose for the sake of contradiction that there is a triple
    $v_h v_i v_j \in \cT \setminus Z_{\cT}$ such that 
    $v_i$ and $v_j$ are connected in $G - Z_V$.
    Then the constraints in $I - X$ imply $x_i = x_j$.
    Moreover, since $v_h v_i v_j \notin Z_{\cT}$,
    constraint $\rel{NEQ}_3(x_h, x_i, x_j)$ is present in $I - X$,
    leading to a contradiction.

    For the opposite direction, assume $(G, \cT, k)$ is a yes-instance,
    and $(Z'_V, Z'_{\cT})$ is a solution.
    Note that $|Z'_V| \leq k$ implies that $Z'_V$ does not contain 
    undeletable vertices $v_i$ for any $i \in [n]$
    (or more precisely, for every undeletable vertex $v_i$, 
    at least one copy of $v_i$ is untouched by $Z'_V$).
    Moreover, we can assume that $Z'_{\cT}$ does not contain triples
    involving the dummy variable $w$:
    every such triple is of the form $z_c v_q w$,
    so we can replace it by deleting $z_c$ instead.
    In other words, $(Z'_V \cup \{z_c\}, Z'_{\cT} \setminus \{z_c v_q w\})$
    is still a solution of the same size.
    Define $X' \subseteq C(I)$ as
    $X'= \{ c : z_c \in Z'_V \} \cup \{ \rel{NEQ}_3(x_h, x_i, x_j) : v_h v_i v_j \in Z'_{\cT} \}$
    and pick any assignment $\alpha : V(I) \to \NN$ such that
    $\alpha(x_i) = \alpha(x_j)$ if and only $v_i$ and $v_j$ are connected in $G - Z'_V$.
    We claim that $\alpha$ satisfies $I - X'$.
    Consider a constraint $c$ in $I - X'$.
    First, suppose $c$ uses a split relation.
    Since $c \notin X'$, we have $z_c \in V(G) \setminus Z'_V$.
    If $c$ implies $x_i = x_j$, then $v_i$ and $v_j$ are connected through
    $z_c$ in $G - Z'_V$, hence $\alpha(v_i) = \alpha(v_j)$.
    If $c$ implies $x_i \neq x_j$, 
    assume by symmetry that
    $z_c$ is connected to $x_i$ in $G - Z'_V$, and 
    there is a triple $z_c x_j w$ in $\cT$.
    By our assumption, $z_c x_j w \notin Z'_T$,
    hence $z_c$ and $x_j$ are disconnected and
    $\alpha(x_i) = \alpha(z_c) \neq \alpha(x_j)$.
    Finally, suppose $c$ is of the form $\rel{NEQ}_3(x_h, x_i, x_j)$.
    Then $v_h v_i v_j$ is in $\cT \setminus Z'_{\cT}$, so
 $v_h$, $v_i$ and $v_j$ appear in distinct components of $G - Z'_V$,
    and $\alpha$ assigns distinct values to them.
  \end{proof}

  Another reduction used to obtain constant-factor fpt-approximation
  for negative equality constraint languages
  starts with an instance
  of $\csp{R^{\neq}_d, =}$, where
  $R^{\neq}_d(x_1, y_1, \dots, x_d, y_d) \equiv \bigvee_{i=1}^{d} x_{i} \neq y_{i}$,
  and produces an instance of \djcut with the same cost.

  \begin{proof}[Proof of Lemma~\ref{lem:essen-neg-lmcut}]
    Introduce a crisp vertex $x_v$ in $G$ for every variable $v$ in $I$ 
    and a soft vertex $z_c$ for every constraint $c$ in $I$.
    For every constraint $c$ in $I$ of the form $u = v$, 
    add edges $x_u z_c$ and $x_v z_c$ in $G$.
    For every constraint $c$ in $I$ of the form 
    $R^{\neq}_d(u_1, v_1, \dots, u_d, v_d)$,
    create a list request 
    $\{x_{u_1} x_{v_1}, \dots, x_{u_d} x_{v_d}, z_c z_c\}$ in $\cL$. 
    Observe that every list has $d+1$ requests.
    This completes the construction.
    Clearly, it requires polynomial time.

    To show that the reduction preserves costs,
    assume $X$ is a solution to $I$, and 
    let $X' = \{z_c : c \in X\}$, noting that $|X'| = |X|$.
    We claim that $X'$ satisfies all lists in $\cL$.
    Observe that all edges in $G$ correspond
    to equality constraints in $I$, which implies
    that $x_u$ and $x_v$ are connected in $G - X'$
    if and only if the constraints in $I - X$ imply that $u = v$.
    Now, consider a list request
    $\{x_{u_1} x_{v_1}, \dots, x_{u_d} x_{v_d}, z_c z_c\}$ in $\cL$. 
    If the constraint $c$ is in $X$, then $z_c$ is in $X'$ and 
    $X'$ satisfies the request.
    Otherwise, there is $i$ such that the constraints in $I - X$ 
    are consistent with $u_i \neq v_i$,
    hence $x_{u_i}$ and $x_{v_i}$ are disconnected in $G - X'$.

    For the opposite direction,
    assume $Y \subseteq V^1(G)$ satisfies $\cL$
    and let $Y' = \{c : z_c \in Y\}$. 
    Observe that $|Y| = |Y'|$.
    We claim that $I - Y'$ is consistent.
    Similarly to the previous part of the proof,
    observe that $x_u$ and $x_v$ are connected in $G - Y$
    if and only if the constraints in $I - Y'$ imply $u = v$.
    Then, for every request list
    $\{x_{u_1} x_{v_1}, \dots, x_{u_d} x_{v_d}, z_c z_c\}$ in $\cL$,
    either $z_c \in Y$, in which case $c \in Y'$,
    or there is $i$ such that $x_{u_i}$ and $x_{v_i}$
    are disconnected in $G - Y$,
    in which case $I - Y'$
    is consistent with $u_i \neq v_i$.
    Hence, $I - Y'$ is satisfied by any assignment
    that maps variables $u$ and $v$
    to the same value if and only if
    $x_u$ and $x_v$ are connected in $G - Y$.
  \end{proof}

\fi 

\section{Triple Multicut}
\label{sec:triple-multicut}

We show that \mdt is in FPT, thus proving
Theorem~\ref{thm:triple-multicut-fpt}.
The algorithm works by a reduction to \textsc{Boolean MinCSP},
i.e. \mincsp{\Delta} for a constraint language
$\Delta$ over the binary domain $\{0,1\}$.
Parameterized complexity of \textsc{Boolean MinCSP} was 
completely classified by~\cite{KimKPW23flow3}.
As a result of our reduction, we obtain an instance
where $\Delta$ is bijunctive, i.e. every relation
in $\Delta$ can be defined by a Boolean formula in CNF
with at most two literals in each clause.
Define the \emph{Gaifman graph} of a bijunctive relation $R$
with vertices $\{1, \dots, r\}$,
where $r$ is the arity of $R$, 
and edges $ij$ for every pair of indices such that
$R(x_1, \dots, x_r)$ implies a 2-clause involving $x_i$ and $x_j$.
A graph is $2K_2$-free if no four vertices
induce a subgraph with two independent edges.

\begin{theorem}[Theorem~1.2~of~\cite{KimKPW23flow3}]
  \label{thm:bij-mincsp}
  Let $\Delta$ be a finite bijunctive Boolean constraint language
  such that the Gaifman graphs of all relations in $\Delta$ are $2K_2$-free.
  Then \mincsp{\Delta} is in FPT.
\end{theorem}

We are ready to present the algorithm.

\begin{theorem} \label{thm:triple-multicut-fpt}
  \mdt is fixed-parameter tractable.
\end{theorem}
\begin{proof}[\iflong Proof \fi \ifshort Proof Sketch \fi]
Let $(G, \cT, k)$ be an instance of \mdt.
By iterative compression, we obtain
$X_V \subseteq V(G)$ and $X_{\cT} \subseteq \cT$
such that $|X_V| + |X_{\cT}| \leq k + 1$ and
all components of $G - X_V$ intersect
triples in $\cT \setminus X_{\cT}$ in at most one vertex.
Moreover, by branching on the intersection,
we can assume that a hypothetical optimal solution $(Z_V, Z_{\cT})$ to 
$(G, \cT, k)$ is disjoint from $(X_V, X_{\cT})$.
Let $X = X_V \cup \bigcup_{uvw \in X_{\cT}} \{u, v, w\}$ and
guess the partition of the vertices in $X$ into connected components of $G - Z_V$.
Identify vertices that belong to the same component, and enumerate 
them via the bijective mapping $\alpha : X \to \{1,\dots,d\}$.
Observe that for every triple $uvw \in X_{\cT}$,
values $\alpha(u)$, $\alpha(v)$ and $\alpha(w)$ are distinct
since $X_{\cT} \cap Z_{\cT} = \emptyset$.
Create an instance $I_\alpha$ of \textsc{Boolean MinCSP} as follows.
\begin{enumerate}
\itemsep0em
\item \label{step:variables} 
Introduce variables $v_i$ and $\hat{v_i}$
for every $v \in V(G)$ and $i \in [d]$.
\item \label{step:vertices} 
For every vertex $v \in V(G)$,
add soft constraint
$\bigwedge_{i < j} (\neg v_i \lor \neg v_j) \land \bigwedge_i (v_i \to \hat{v}_i)$.
\item \label{step:terminals}
For every vertex $v \in X$, add
crisp constraints $v_{\alpha(v)}$, $\hat{v}_{\alpha(v)}$, and
$\neg v_j$, $\neg \hat{v}_j$ for all $j \neq \alpha(v)$.
\item \label{step:edges}
For every edge $uv \in E(G)$ and $i \in [d]$, 
add crisp constraints $\hat{u_i} \to v_i$ and $\hat{v_i} \to u_i$.
\item \label{step:triples}
For every triple $uvw \in \cT$ and $i \in [d]$, add soft constraints 
$(\neg \hat{u_i} \lor \neg \hat{v_i}) \land (\neg \hat{v_i} \lor \neg \hat{w_i}) \land (\neg \hat{u_i} \lor \neg \hat{w_i})$.
\end{enumerate}
This completes the reduction.
Observe that all relations used in $I_\alpha$ are bijunctive.
Their Gaifman graphs are
cliques with pendant edges attached to all vertices (case~\ref{step:vertices}),
edgeless (case~\ref{step:terminals}),
single edges (case~\ref{step:edges}), or
triangles (case~\ref{step:triples}).
These graphs are $2K_2$-free,
therefore, we can decide whether $(I_\alpha, k)$
is a yes-instance in fpt time
by Theorem~\ref{thm:bij-mincsp}.

The intuitive idea behind the reduction is that deleting
a vertex from the graph corresponds to
deleting the constraint of type~\ref{step:vertices}
for that vertex from the \textsc{MinCSP} instance.
If the constraint for $v \in V(G)$ is present, then
$s$ may reach at most one variable in $\{\hat{v}_i : i \in [d]\}$
by a path of implications.
We can interpret $s$ reaching $\hat{v}_i$ as
placing $v$ into the $i$th connected component of the resulting graph.
Crisp constraints of type~\ref{step:terminals} ensure that
the vertices of $X$ are partitioned into components according to $\alpha$.
Crisp constraints of type~\ref{step:edges} ensure that
neighbouring vertices are placed into the same components.
If $s$ does not reach any variable in $\{\hat{v}_i : i \in [d]\}$,
then $v$ ends up in the same component as in $G - X_V$,
and, by iterative compression, it is not part of any violated triple.
On the other hand, if a constraint of type~\ref{step:vertices}
for a vertex $v$ is deleted, then the Boolean assignment mapping
all $v_i$ and $0$ and all $\hat{v}_i$ to $1$
is compatible with all other constraints involving any of these variables.
Finally, constraints of type~\ref{step:triples} ensure 
that no pair of variables in a triple is assigned the same value.
\ifshort
We defer further details to the full version.
\fi

\iflong
  To show correctness formally, we first assume that 
  $(G, \cT, k)$ is a yes-instance and $(Z_V, Z_{\cT})$ is a solution
  that respects $\alpha$, i.e.
  if $\alpha(x) \neq \alpha(y)$, then
  $Z_V$ disconnects $x$ and $y$ in $G$.
  Define an assignment $\varphi : V(I_\alpha) \to \{0,1\}$ as follows.
  For all $i \in [d]$
  set $\varphi(v_i) = \varphi(\hat{v}_i) = 1$ if $v$ is not in $Z_V$ 
  and is connected to $\alpha^{-1}(i)$ in $G - Z_V$,
  and let $\varphi(v_i) = \varphi(\hat{v}_i) = 0$ otherwise.
  For $v \in Z_V$ and $i \in [d]$, set
  $\varphi(v_i) = 1$ and $\varphi(\hat{v}_i)=0$.
  We claim that $\varphi$ satisfies all crisp constraints in $I_\alpha$,
  and violates at most $|Z_V| + |Z_{\cT}| \leq k$ soft constraints.
  Constraints introduced in step~\ref{step:terminals}
  are satisfied because $(Z_V, Z_{\cT})$ agree with $\alpha$.
  To see that constraints introduced in step~\ref{step:edges}
  are satisfied, we need to consider several cases.
  If $u,v \in V(G) \setminus Z_V$ and both are reachable in $G - Z_V$
  from $\alpha^{-1}(i)$, then the constraints are satisfied
  because $\varphi(u_i) = \varphi(\hat{u}_i) = 1 = \varphi(v_i) = \varphi(\hat{v}_i)$.
  If $u,v \in V(G) \setminus Z_V$ and both are unreachable in $G - Z_V$
  from $\alpha^{-1}(i)$, then the constraints are satisfied
  because $\varphi(u_i) = \varphi(\hat{u}_i) = 0 = \varphi(v_i) = \varphi(\hat{v}_i)$.
  This completes the cases with $u,v \in V(G) \setminus Z_V$
  because $uv$ is an edge in $G$.
  Finally, if $v \in Z_V$, then $\varphi({v}_i) = 1$ and $\varphi(\hat{v}_i) = 0$,
  so $\varphi$ satisfies all constraints.
  Now consider soft constraints, starting with step~\ref{step:vertices}.
  These constraints are satisfied by $\varphi$ for all $v \in V(G) \setminus Z_V$:
  since $Z_V$ respects $\alpha$,
  $v$ is reachable from $\alpha^{-1}(i)$ for at most one $i \in [d]$;
  moreover, $v_i = \hat{v}_i$ is satisfied by $\varphi$ by definition.
  Hence, $\varphi$ violates at most $|Z_V|$ such constraints.
  Further, consider a triple $uvw \in \cT \setminus Z_{\cT}$.
  We claim that $\varphi$ satisfies all constraints
  introduced in step~\ref{step:triples} for $uvw$.
  Suppose for the sake of contradiction that 
  $\varphi(\hat{u}_i) = \varphi(\hat{v}_i) = 1$ for some $i \in [d]$.
  Then, by definition of $\varphi$,
  $\alpha^{-1}(i)$ reaches both $u$ and $v$ in $G - Z_V$,
  contradicting the fact that $(Z_V, Z_{\cT})$ is a solution.
  Finally, for every $uvw \in Z_{\cT}$ there can be at most one $i \in [d]$ 
  for which the introduced constraint is violated by $\varphi$,
  since each such violation requires 
  two out of the variables $u_i, v_i, w_i$ to be assigned $1$.
  Hence, $\varphi$ violates at most $|Z_{\cT}|$ such constraints,
  and $|Z_V| + |Z_{\cT}|$ constraints in total.
  
  Now suppose $\varphi'$ satisfies all crisp constraints and
  violates at most $k$ soft constraints in $I_\alpha$.
  Define $Z'_V$ as the set of variables $v \in V(G)$
  such that the constraint for $v$ introduced 
  in step~\ref{step:vertices} is violated $\varphi$.
  Define $Z'_{\cT}$ as the set of triples $uvw \in \cT$
  such that a constraint for $uvw$ introduced
  in step~\ref{step:triples} is violated by $\varphi'$.
  Observe that $|Z'_V| + |Z'_{\cT}| \leq k$.
  We claim that $(Z'_V, Z'_\cT)$ is a solution to $(G, \cT, k)$.
  First, observe that $(Z'_V, Z'_\cT)$ respects $\alpha$,
  i.e. $\alpha(u) \neq \alpha(v)$ implies that
  $u$ and $v$ are disconnected in $G - Z'_V$.
  Indeed, if this is not the case, then
  crisp constraints of $I_\alpha$ imply
  that $\varphi(\hat{u}_\alpha(v)) = 1$,
  which violates the crisp constraint $\neq \hat{u}_{\alpha(v)}$.
  Now we show that for every triple $uvw \in \cT \setminus Z'_{\cT}$,
  vertices $u$, $v$ and $w$ occur in distinct connected
  components of $G - Z'_V$.
  Suppose for the sake of contradiction
  that there is a triple $uvw \in \cT \setminus Z'_{\cT}$
  such that $u$ and $v$ are connected in $G - Z'_V$.
  If there exists $i \in [d]$ such that
  $\alpha^{-1}(i)$ reaches $u$ and $v$, then
  the crisp constraints imply that 
  $\varphi'(\hat{u}_i) = \varphi'(\hat{v}_i) = 1$, and 
  the clause $(\neg \hat{u}_i \lor \neg \hat{v}_i)$ is violated,
  which implies $uvw \in Z'_{\cT}$, a contradiction.
  Otherwise, $u$ and $v$ are disconnected from $X$ in $G - Z'_V$.
  Then, there is a path from $u$ to $v$ in $G - (Z'_V \cup X)$, and,
  consequently, in $G - X$, which implies $uvw \in X_{\cT}$.
  However, values $\alpha(u)$, $\alpha(v)$ and $\alpha(w)$
  are distinct, and $Z'_V$ respects $\alpha$, which is a contradiction.
  This completes the proof.
\fi
\end{proof}

\section{Disjunctive and Steiner Multicut}
\label{sec:djcut-and-steiner}

We show that two generalizations of \textsc{Vertex Multicut},
\textsc{Disjunctive Multicut} and \textsc{Steiner Multicut},
are constant-factor fpt-approximable,
proving Theorems~\ref{thm:djcut-fpta}~and~\ref{thm:steiner-fpa},
respectively.
Section~\ref{ssec:djcut-main-loop} presents the main loop
of the \textsc{Disjunctive Multicut} algorithm,
while Section~\ref{ssec:djcut-simplification}
is dedicated to the most technical subroutine of the algorithm
that involves \emph{randomized covering of shadow}~\cite{marx2014fixed}.
In Section~\ref{ssec:steiner} we present a simpler and more efficient
algorithm for~\textsc{Steiner Multicut}
that avoids shadow covering using the idea of~\cite{lokshtanov2021fpt}.

\subsection{Main Loop of the Disjunctive Multicut Algorithm}
\label{ssec:djcut-main-loop}

Let $G$ be a graph with vertices $V(G) = V^\infty(G) \uplus V^1(G)$ 
partitioned into undeletable and deletable, respectively.
A subset $L \subseteq \binom{V(G)}{2}$ of pairs
is a \emph{request list}, and
a set of vertices $X \subseteq V(G)$ \emph{satisfies} $L$ if
there is a pair $st \in L$ separated by $X$.
This includes the possibility that $s \in X$ or $t \in X$.
For a graph $G$ and a collection of request lists $\cL$,
we let $\cost(G, \cL)$ be the minimum size of a set 
$X \subseteq V^1(G)$ that satisfies all lists in $\cL$.
\djcut asks, given an instance $(G, \cL)$,
whether $\cost(G, \cL) \leq k$.

\djcut problem generalizes not only \textsc{Multicut}
(which is a special case with $d = 1$) 
but also \textsc{$d$-Hitting Set}.
To see the latter, take an edgeless graph $G$ and
make every request a \emph{singleton}, i.e. a pair $ss$ 
for a vertex $s \in V(G)$.
The only way to satisfy a singleton $ss$ is 
to delete the vertex $s$ itself,
and the only way to satisfy a list of singletons is 
to delete one of the vertices in it.

The intuitive idea behind the approximation algorithm for \djcut
is to iteratively simplify the instance $(G, \cL, k)$, 
making it closer to \textsc{Bounded Hitting Set} after each iteration.
Roughly, we make progress if the maximum number of non-singleton requests in a list decreases.
In each iteration, the goal is to find a set of $O(k)$ vertices 
whose deletion, combined with some branching steps, simplifies every request list.
This process can continue for $O(d)$ steps until we obtain
an instance of \textsc{Bounded Hitting Set}, 
which can be solved in fpt time by branching.
The instance may increase in the process,
but finally we obtain a solution of cost $f(d) \cdot k$ for some function $f$.
We do not optimize for $f$ in our proofs.
Observe also that in the context of constant-factor fpt approximability,
some dependence on $d$ is unavoidable since 
the problem with unbounded $d$ generalizes \textsc{Hitting Set}.

Formally, for a request list $L$, let $\mu_1(L)$ and $\mu_2(L)$
be the number of singleton and non-singleton cut requests in $L$, respectively.
Define the measure for a list $L$ as 
$\mu(L) = \mu_1(L) + 3\mu_2(L) = |L| + 2\mu_2(L)$,
and extend it to a collection of list requests $\cL$ 
by taking the maximum, i.e. $\mu(\cL) = \max_{L \in \cL} \mu(L)$.
Observe that $\mu(\cL) \leq 3d$ for any instance of \djcut.
Further, let $V(L) = \bigcup_{st \in L} \{s,t\}$ denote
the set of vertices in a list $L$, and $\nu(\cL) = \mu_1(L) + 2\mu_2(L)$
be an upper bound on the maximum number of variable occurrences in a list of $\cL$.
The workhorse of the approximation algorithm is the following lemma.

\begin{lemma} \label{lem:lmcut-iteration}
  There is a randomized algorithm \emph{\simplify}
  that takes an instance
  $(G, \cL, k)$ of \djcut as input, and 
  in $O^*(2^{O(k)})$ time produces 
  a graph $G'$ and a collection of requests $\cL'$
  such that
  $|V(G')| \leq |V(G)|$, 
  $\nu(\cL') \leq \nu(\cL)$,
  $|\cL'| \leq k^2|\cL|$,
  and
  $\mu(\cL') \leq \mu(\cL) - 1$.
  Moreover, the following holds.
  \begin{itemize}
  \itemsep0em
  \item If $\cost(G, \cL) \leq k$, then, with probability $2^{-O(k^2)}$,  we have $\cost(G', \cL') \leq 2k$.
  \item If $\cost(G, \cL) > 3k$, then we have $\cost(G', \cL') > 2k$.
  \end{itemize}
\end{lemma}

Randomization in Lemma~\ref{lem:lmcut-iteration} comes from the use
of the \emph{random covering of shadow} of~\cite{marx2014fixed,chitnis2015directed}.
They also provide a derandomized version of this procedure,
so our algorithm can be derandomized as well.
We postpone the proof of Lemma~\ref{lem:lmcut-iteration} 
until Section~\ref{ssec:djcut-simplification}
since it requires introduction of some technical machinery.
For now, we show how to prove 
Theorem~\ref{thm:djcut-fpta} using the result of the lemma.

\begin{algorithm}
  \caption{Main Loop.}
  \begin{algorithmic}[1]
    \Procedure{SolveDJMC}{$G, \cL, k$}
    \While{$\mu_2(\cL) > 0$}
    \State $(G, \cL) \gets \simplify(G, \cL, k)$
    \If{\simplify rejects}
    \State \textbf{reject}
    \EndIf
    \State $k \gets 2k$
    \EndWhile
    \State $W \gets \{vv : v \in V(G) \}$
    \If{\textsc{SolveHittingSet}($W, \cL, k$) accepts}
    \State \textbf{accept}
    \Else
    \State \textbf{reject}
    \EndIf
    \EndProcedure
  \end{algorithmic}
  \label{alg:mainloop}
\end{algorithm}

\begin{theorem} \label{thm:djcut-fpta}
  \djcut is fixed-parameter tractable.
\end{theorem}
\begin{proof}
  Let $(G, \cL, k)$ be an instance of \djcut.
  Repeat the following steps until $\mu_2(\cL) = 0$.
  Apply the algorithm of Lemma~\ref{lem:lmcut-iteration} to
  $(G, \cL, k)$, obtaining a new graph $G$ and
  a new collection of lists $\cL$,
  and let $(G, \cL, k) := (G', \cL', 2k)$.
  When $\mu_2(\cL) = 0$,
  let $W = \{ vv : v \in V(G) \}$ be the set of
  singleton cut requests for every vertex in $V(G)$.
  Check whether $(W, \cL, k)$ is a yes-instance of 
  \textsc{Hitting Set} -- if yes, accept, otherwise reject. 
  See Algorithm~\ref{alg:mainloop} for the pseudocode.

  To argue correctness,
  let $(G, \cL, k)$ be the input instance
  and $(G', \cL', k')$ be the instance obtained after simplification.
  By induction and Lemma~\ref{lem:lmcut-iteration}, we have
  $|V(G')| \leq |V(G)|$ and $\nu(\cL') \leq \nu(\cL)$. 
  Since $\nu(\cL') \leq \nu(\cL) \leq 2d$ and $\mu_2(\cL) = 0$, 
  every list in $\cL$ has at most $2d$ requests.
  Let $r$ be the number of calls to \simplify performed by the algorithm.
  Note that $r \leq \mu(\cL) \leq 3d$ since the measure decreases by at least one 
  with each iteration, and define $k' = 2^r k$.
  The lists in $\cL'$ only contain singletons, thus $(G', \cL', k')$ is essentially 
  an instance of \textsc{Hitting Set} with sets of size $2d$.
  Moreover, $|\cL'| \leq k^{2r} |\cL|$, so the number of lists is polynomial in $|\cL|$.
  We can solve $(G', \cL', k')$ in $O^*((2d)^{k'})$ time by branching
  (see, for example, Chapter~3~in~\cite{cygan2015parameterized}).
  For the other direction, suppose $\cost(G, \cL) \leq k$.
  By Lemma~\ref{lem:lmcut-iteration} and induction,
  we have $\cost(G', \cL') \leq 2^r k \leq k'$ with probability $2^{-O(rk^2)}$,
  and the algorithm accepts.
  If $\cost(I) > 3k$, then $\cost(G', \cL') > k'$ and the algorithm rejects.
\end{proof}

\subsection{Simplification Procedure}
\label{ssec:djcut-simplification}

In this section we prove Lemma~\ref{lem:lmcut-iteration}.
We start by iterative compression and guessing.
Then we delete at most $k$ vertices from the graph and modify it,
obtaining an instance amenable to the main technical tool 
of the section -- the \emph{shadow covering} technique.

\subsubsection{Initial Phase}
\label{sec:initial-phase}

Let $(G, \cL, k)$ be an instance of \djcut.
By iterative compression, assume we have a set $X \subseteq V(G)$
that satisfies all lists in $\cL$ and $|X| = c \cdot k + 1$,
where $c := c(d)$ is the approximation factor.
Assume $Z$ is an optimal solution to $G$, i.e. 
$|Z| \leq k$ and $Z$ satisfies all lists in $\cL$.
Guess the intersection $W = X \cap Z$,
and let $G' = G - W$,
$X' = X \setminus W$, and
$Z' = Z \setminus W$.
Construct $\cL'$ starting with $\cL$ and removing 
all lists satisfied by $W$.
Further, guess the partition $\cX = (X_1, \dots, X_\ell)$
of $X'$ into the connected components of $G' - Z'$, and 
identify the variables in each subset $X_i$ into a single vertex $x_i$,
and redefine $X'$ accordingly.
Note that the probability of our guesses being correct
up to this point is $2^{-O(k \log k)}$.
Also, these steps can be derandomized by creating
$2^{O(k \log k)}$ branches.

Now compute a minimum $\cX$-multiway cut in $G'$,
i.e. a set $M \subseteq V^1(G')$ that separates
every pair of vertices $x_i$ and $x_j$ in $X'$.
Note that $Z'$ is a $\cX$-multiway cut by the definition of $\cX$,
so $|M| \leq |Z'| \leq k$.
Such a set $M$ can be computed in $O^*(2^k)$ time
using the algorithm of~\cite{cygan2013multiway}.
If no $\cX$-multiway cut of size at most $k$ exists, 
then abort the branch and make another guess for $\cX$.
If an $\cX$-multiway cut $M$ of size at most $k$ is obtained,
remove the vertices in $M$ from $G$ and 
along with the lists in $\cL'$ satisfied by $M$.
This completes the initial phase of the algorithm.
Properties of the resulting instance are summarized below.

\begin{lemma} \label{lem:initial-phase}
  After the initial phase we obtain 
  a graph $G'$, a family of list requests $\cL'$,
  and subset of vertices $X' \subseteq V(G')$
  such that
  $|V(G')| \leq |V(G)|$, 
  $\nu(\cL') \leq \nu(\cL)$,
  $\mu(\cL') \leq \mu(\cL)$, and 
  $|X'| \in O(k)$.
  The set $X'$ satisfies all lists in $\cL'$ and
  intersects each connected component of $G'$ in at most one vertex.
  Moreover, the following hold.
  \begin{itemize}
  \itemsep0em
  \item $\cost(G, \cL) \leq k + \cost(G', \cL')$.
  \item If $\cost(G, \cL) \leq k$, then,
    with probability $2^{-O(k \log k)}$,
    we have $\cost(G', \cL') \leq k$. 
    Moreover, there is a set $Z' \subseteq V(G')$, $|Z'| \leq k$ 
    that satisfies all lists in $\cL'$ and is disjoint from $X'$.
  \end{itemize}
\end{lemma}
\begin{proof}
  All statements apart from the last two
  are immediate from the construction.
  For the first statement,
  note that an optimal solution to $(G', \cL', k)$
  combined with the $\cX$-multiway cut $M$ has size at most $2k$ and 
  satisfies all lists in $\cL'$.

  To see that $\cost(G, \cL) \leq k$ implies $\cost(G', \cL') \leq k$
  with probability $2^{-O(k \log k)}$,
  observe that, assuming our guesses for $W$ and $\cX$ are correct,
  there is an optimal solution $Z$ to $(G, \cL, k)$
  such that $W = Z \cap X$ and 
  $Z$ partitions vertices of $X$ into 
  connected components according to $\cX$.
  Then, $Z' = Z \setminus W$ is a solution to $(G', \cL')$,
  $|Z'| \leq k$, and $Z' \cap X = \emptyset$.
\end{proof}

\subsubsection{Random Covering of Shadow}
\label{sec:shadow-cover}

Random covering of shadow is a powerful tool
introduced by~\cite{marx2014fixed} and 
sharpened by~\cite{chitnis2015directed}.
We use the latter work as our starting point.
Although~\cite{chitnis2015directed} present their theorems in terms of directed graphs,
their results are applicable to our setting by considering 
undirected edges as bidirectional,
i.e. replacing every edge $uv$ with a pair of antiparallel arcs $(u,v)$ and $(v, u)$.
Consider a graph $G$ with vertices partitioned 
into deltable and undeletable subsets, i.e. $V(G) = V^1(G) \uplus V^\infty(G)$.
Let $\cF = (F_1, \dots, F_q)$ be a family of connected subgraphs of $G$.
An \emph{$\cF$-transversal} is a set of vertices $T$
that intersects every subgraph $F_i$ in $\cF$.
If $T$ is an $\cF$-transversal, we say that $\cF$ is \emph{$T$-connected}.
For every $W \subseteq V(G)$, the \emph{shadow of $W$ (with respect to $T$)}
is the subset of vertices disconnected from $T$ in $G - W$.
We state it for the case $T \subseteq V^\infty(G)$ which suffices for our applications.

\begin{theorem}[Random Covering of Shadow, Theorem~3.5~in~\cite{chitnis2015directed}] \label{thm:shadow-cover}
  There is an algorithm \emph{\randomcover} that takes a graph $G$,
  a subset $T \subseteq V^{\infty}(G)$ and an integer $k$ as input,
  and in $O^*(4^k)$ time outputs a set $S \subseteq V(G)$
  such that the following holds.
  For any family $\cF$ of $T$-connected subgraphs,
  if there is an $\cF$-transversal of size at most $k$ in $V^1(G)$,
  then with probability $2^{-O(k^2)}$,
  there exists an $\cF$-transversal $Y \subseteq V^1(G)$ 
  of size at most $k$ such that
  \begin{enumerate}
  \itemsep0em
  \item $Y \cap S = \emptyset$, and
  \item $S$ covers the shadow of $Y$ with respect to $T$.
  \end{enumerate}
\end{theorem}

The following consequence is convenient for our purposes.

\begin{corollary} \label{cor:solution-cover}
  Let $S$ and $Y$ be the shadow-covering set and the $\cF$-transversal
  from Theorem~\ref{thm:shadow-cover}, respectively.
  Define $R = V(G) \setminus S$ to be the complement of $S$.
  Then $Y \subseteq R$ and, for every vertex $v \in R$, 
  either $v \in Y$ or $v$ is connected to $T$ in $G - Y$.
\end{corollary}
\begin{proof}
  By definition, $R$ is the complement of $S$, hence 
  $Y \cap S = \emptyset$ implies that $Y \subseteq R$.
  Since $S$ covers the shadow of $Y$ with respect to $T$, 
  the set $R$ is outside the shadow.
  Hence, a vertex $v \in R$ is either connected to $T$ in $G - Y$ or it is contained in $Y$.
\end{proof}

Note that if a vertex $v \in N(S)$ and $v$ is undeletable,
then $v \in R$ and $v \notin Y$,
hence $v$ is connected to $T$ in $G - Y$.
Since $Y \cap S = \emptyset$,
every vertex in $N(v) \cap S$ is also connected to $T$ in $G - Y$,
so we can remove $N(v) \cap S$ from $S$ (and add it to $R$ instead).
By applying this procedure to exhaustion, 
we may assume that no vertex in $N(S)$ is undeletable.

With the random covering of shadow at our disposal, 
we return to \djcut.
By Lemma~\ref{lem:initial-phase},
we can start with an instance $(G, \cL, k)$
and a set $X \subseteq V(G)$ such that
$|X| \in O(k)$, $X$ satisfies all lists in $\cL$,
every connected component of $G$ intersects $X$ in at most one vertex,
and there is an optimal solution $Z$ disjoint from $X$.
Let $\cT := \cT(G, \cL, X, Z)$ be the set of cut requests in $\bigcup \cL$ 
satisfied by both $X$ and $Z$.
Define $\cF$ as the set of $st$-walks for all $st \in \cT$.
Observe that an $\cF$-transversal is precisely a $\cT$-multicut.
Apply the algorithm from Theorem~\ref{thm:shadow-cover} to $(G, X, k)$.
Since $X$ and $Z$ are $\cF$-transversals and $|Z| \leq k$ by assumption,
Theorem~\ref{thm:shadow-cover} and Corollary~\ref{cor:solution-cover} imply that
we can obtain a set $R \subseteq V(G)$ in fpt time such that,
with probability $2^{-O(k^2)}$,
there is an $\cF$-transversal $Y \subseteq R$ of size at most $k$,
and every vertex in $R \setminus Y$ is connected to $X$ in $G - Y$.

For every vertex $v \in V^{1}(G) \setminus X$, 
define a set of vertices $R_v \subseteq R \setminus X$ as follows:
\begin{itemize}
\itemsep0em
\item if $v$ is disconnected from $X$, then let $R_v = \emptyset$;
\item if $v \in N(X)$ or $v \in R$, then let $R_v = \{v\}$;
\item otherwise, let $R_v = R \cap N(H)$, 
  where $H$ is the component of $G[S]$ containing $v$.
\end{itemize}
Note that, by definition, the set $R_v$ is an $Xv$-separator in $G$.
Moreover, we have ensured that $N(S)$ does not contain undeletable vertices,
so $R_v$ does not contain undeletable vertices.
In a certain sense, the sets $R_v$ are the only $Xv$-separators 
that $Y$ needs to use.
This idea is made precise in the following lemma.

\begin{lemma} \label{lem:decided-separators}
  Let $G$ be a graph.
  Let $X$ and $Y$ be disjoint subsets of $V(G)$ such that 
  $X$ intersects every connected component of $G$ in at most one vertex.
  Suppose $R \subseteq V(G)$ is such that $Y \subseteq R$ and 
  all vertices in $R \setminus Y$ are connected to $X$ in $G - Y$.
  If a vertex $s$ is disconnected from $X$ in $G - Y$, 
  then $R_s \subseteq Y$.
\end{lemma}
\begin{proof}
  Assume that $s$ is connected to a vertex $x \in X$ in $G$,
  as otherwise the conclusion holds vacuously.
  Let $Y$ be an $xs$-separator for some $x \in X$.
  If $s \in N(x)$, then clearly $R_s = \{s\} \subseteq Y$.
  Otherwise, if $s \in R$, then $R_s = \{s\} \subseteq Y$ 
  because every vertex in $R \setminus Y$ is connected to $X$ in $G - Y$.
  Finally, if $s \notin N(x) \cup R$, then we claim that
  $R_s = R \cap N(H)$ is contained in $Y$, 
  where $H$ is the connected component of $s$ in $G - R$.
  Suppose for the sake of contradiction that $R_s \nsubseteq Y$ and 
  there is $v \in R_s \setminus Y$.
  Then $v \in R \setminus Y$, so
  $v$ is connected to $x$.
  However, since $v \in N(H)$, it has a neighbour in $H$,
  and is connected to $s$ in $G - R$ and in $G - Y$,
  contradicting that $Y$ is an $xs$-separator.
\end{proof}

Now we compute a simplified collection of lists $\cL'$.
Start by adding all lists in $\cL$ to $\cL'$.
Remove every singleton request $xx$ such that $x \in X$ from every list of $\cL'$.
For every list $L \in \cL'$ not shortened this way,
let $st \in L$ be a non-singleton cut request satisfied by $X$. 
Consider $R_s$ and $R_t$ and
apply one of the following rules.
\begin{enumerate}[(R1)]
\itemsep0em
\item If $|R_s| > k$ and $|R_t| > k$, remove $st$ from $L$.
\item If $|R_s| \leq k$ and $|R_t| > k$, replace $L$
  with sets $(L \setminus \{st\}) \cup \{aa\}$ for all $a \in R_s$.
\item If $|R_s| > k$ and $|R_t| \leq k$, replace $L$ 
  with sets $(L \setminus \{st\}) \cup \{bb\}$ for all $b \in R_t$.
\item If $|R_s| \leq k$ and $|R_t| \leq k$, replace $L$ with
  sets $(L \setminus \{st\}) \cup \{aa, bb\}$ 
  for all $a \in R_t$, $b \in R_t$.
\end{enumerate}
Finally, make vertices in $X$ undeletable, obtaining a new graph $G'$.
This completes the simplification step.
Note that each list in $\cL$ is processed once,
so the running time of the last step is polynomial.

Now we prove some properties of $G'$ and $\cL'$ obtained above.
Note that $|V(G')| = |V(G)|$.
Since every list in $\cL$ is processed once
and with at most $k^2$ new lists,
the size of $|\cL'|$ grows by a factor of at most $k^2$.
To see that $\nu(\cL') \leq \nu(\cL)$ and $\mu(\cL') \leq \mu(\cL) - 1$,
observe that every reduction rule replaces a list $L$ with new lists
with either one less non-singleton request (so $\mu_2$ decreases by at least $1$),
and adds up to two singleton requests (so $\mu_1$ increases by at most $2$).
Moreover, in every list of $\cL$ there is a cut request satisfied by $X$,
so no list of $\cL$ remains unchanged in $\cL'$.
We prove the remaining observations in separate lemmas.

\begin{lemma} \label{lem:cost-complete}
  If $\cost(G, \cL) \leq k$, then, with probability $2^{-O(k^2)}$, 
  we have $\cost(G', \cL') \leq 2k$.
\end{lemma}
\begin{proof}
  Let $Z$ be a solution to $(G, \cL, k)$,
  i.e. $|Z| \leq k$ and $Z$ satisfies all requests in $\cL$.
  By Lemma~\ref{lem:initial-phase}, we may assume that $Z \cap X = \emptyset$.
  Let $R$ be the complement of the shadow covering set obtained 
  by Theorem~\ref{thm:shadow-cover} and Corollary~\ref{cor:solution-cover}.
  With probability $2^{-O(k^2)}$, there exists an $\cF$-transversal
  $Y \subseteq R$ of size at most $k$.
  We claim that $Y \cup Z$ satisfies all lists in $\cL'$.
  Since $|Y \cup Z| \leq 2k$, this suffices to prove the lemma.

  Consider a list $L \in \cL$. 
  We show that $Z \cup Y$ 
  satisfies every list obtained from $L$ by the reduction rules.
  If $L$ contains $xx$ for some $x \in X$, then 
  $Z$ satisfies $L \setminus \{xx\}$ because $X \cap Z = \emptyset$.
  If $L$ contains a request satisfied by $Z$ but not $X$,
  then this request remains in every list derived from $L$,
  and $Z$ satisfies all these lists.
  There is one remaining case: when $X$ and $Z$ satisfy 
  the same non-singleton request $st$ in $L$.
  Then $Y$ also satisfies $st$ since $st \in \cT$ and $Y$ is a $\cT$-multicut.
  We claim that $Y$ satisfies all lists derived from $L$ in this case.
  Note that there is a unique $x \in X$ that is contained on every $st$-walk,
  and $Y$ separates $x$ from $s$ or $t$.
  Assume by symmetry that $Y$ is an $xs$-separator.
  By Lemma~\ref{lem:decided-separators},
  we obtain $R_s \subseteq Y$, and since
  $|Y| \leq k$, we have $|R_s| \leq k$.
  Thus, every list derived from $L$
  contains a single request $aa$ for some $a \in R_s$,
  and $Y$ satisfies every such list.
\end{proof}

Now we show the remaining direction.

\begin{lemma} \label{lem:cost-sound}
  If $\cost(G, \cL) > 2k$, then
  we have $\cost(G', \cL') > 2k$. 
\end{lemma}
\begin{proof}
  We prove the contrapositive.
  Suppose $\cost(G', \cL') \leq 2k$ and 
  $Z'$ is an optimal solution to $(G', \cL', 2k)$.
  We claim that $Z'$ is also a solution to $(G, \cL, 2k)$.
  Consider a list $L \in \cL$.
  It suffices to show that if $Z'$ satisfies a list $L'$ derived from $L$
  by one of the reduction rules, then $Z'$ satisfies $L$ as well.
  If $L'$ is derived from $L$ be removing $xx$ for some $x \in X$,
  then $Z'$ satisfies $L \setminus \{xx\}$
  because the vertices in $X$ are undeletable in $G'$, so $Z' \cap X = \emptyset$.
  If $L'$ is derived from $L$ by removing some request $st \in \cT$,
  then we need to consider several cases.
  If $Z'$ satisfies $L \setminus \{st\}$, then it clearly satisfies $L$.
  Otherwise, we claim that $R_s \subseteq Z'$ or $R_t \subseteq Z'$.
  Suppose all lists derived from $L$ after removing $st$
  have singletons $aa$ for all $a \in R_s$ added to them.
  Since $Z'$ satisfies all these lists and does not satisfy $L \setminus \{st\}$,
  we have $R_s \subseteq Z'$.
  The cases when singletons $bb$ for all $b \in R_t$ 
  are added to the derived lists
  follow by an analogous argument.
  This completes the case analysis.
\end{proof}

We are now ready to prove Lemma~\ref{lem:lmcut-iteration}.

\begin{proof}[Proof of Lemma~\ref{lem:lmcut-iteration}]
  Suppose $(G, \cL, k)$ is a yes-instance of \djcut.
  By Lemma~\ref{lem:initial-phase}, 
  after the initial phase we obtain $G'$, $\cL'$
  such that
  $|V(G')| \leq |V(G)|$,
  $\nu(\cL') \leq \nu(\cL)$,
  $\mu(\cL') \leq \mu(\cL)$,
  and
  $\cost(G', \cL') \leq k$.
  Moreover, we obtain a set $X' \subseteq V(G')$, $|X'| \in O(k)$,
  that satisfies all lists in $\cL'$,
  intersects every component of $G'$ in at most one vertex,
  and is disjoint from an optimum solution $Z'$ to $(G', \cL', k)$.
  Now we apply random covering of shadow and the list reduction rules
  to $G', \cL', X'$,
  obtaining a new graph $G''$ and a new set of lists $\cL''$.
  By Lemma~\ref{lem:cost-complete},
  with probability $2^{O(-k^2)}$, we have $\cost(G'', \cL'') \leq 2k$.
  This proves one statement of Lemma~\ref{lem:lmcut-iteration}.

  For the second statement of Lemma~\ref{lem:lmcut-iteration},
  suppose $\cost(G'', \cL'') \leq 2k$.
  By Lemma~\ref{lem:cost-sound},
  we have $\cost(G', \cL') \leq 2k$.
  By Lemma~\ref{lem:initial-phase},
  $\cost(G' \cL') \leq 2k$
  implies that 
  $\cost(G, \cL) \leq 2k + k \leq 3k$, 
  and we are done.
\end{proof}

\subsection{An Improved Algorithm for Steiner Multicut}
\label{ssec:steiner}

\ifshort
We present a simpler and faster algorithm for 
\textsc{(Vertex) Steiner Multicut}, proving the following theorem.

\begin{theorem} \label{thm:steiner-fpa}
  \textsc{Steiner Multicut} with requests of constant size
  is $2$-approximable in $O^*(2^{O(k)})$ time.
\end{theorem}
\fi

\iflong
We present a simpler and faster algorithm for 
\textsc{(Vertex) Steiner Multicut}, proving Theorem~\ref{thm:steiner-fpa}.
\fi
Our approximation algorithm builds upon
the $O^*(2^{O(k)})$-time $2$-approximation
for \textsc{Vertex Multicut} of~\cite{lokshtanov2021fpt}.
Note that it is a special case of \textsc{Steiner Multicut} with $p = 2$.
We only need to change one subroutine in their algorithm,
so we describe the complete procedure informally,
invoking relevant results from~\cite{lokshtanov2021fpt}
and proving that our modified subroutine allows us to handle 
the cases with $p > 2$.
Our goal is to reduce the problem to
\textsc{Strict Steiner Multicut},
which is a special case of \textsc{Steiner Multicut}
where the input comes with a designated vertex $x$
such that $\{x\}$ satisfies all requests,
and the goal is to find a multicut of size at most $k$
that does not include $x$.
As we show in the sequel, this problem can be solved in
single-exponential fpt time.

Let $(G, \cT, k)$ be an instance of \textsc{Steiner Multicut}.
Start by iterative compression and branching,
which allows to assume access to a set $X$
of size at most $2k+1$ that satisfies all subsets in $\cT$.
Further, we may assume that a hypothetical 
optimal solution $Z$ is disjoint from $X$.
Let $\cX$ be the partition of $X$ into connected components of $G - Z$,
and $Z' \subseteq Z$ be a minimal subset of $Z$ that is a $\cX$-multiway cut.
One can guess $|Z'|$ in polynomial time.
Using the method of two important separators~(Lemma~3.2~of~\cite{lokshtanov2021fpt})
and guessing partial information about $\cX$ using a divide-and-conquer approach
(as in Theorem~3.2~of~\cite{lokshtanov2021fpt}),
one can find a subset $M$ of at most $2|Z'|$ vertices
such that $X$ is partitioned into the connected components of $G - M$ according to $\cX$.
The remainder of the problem can be solved by $|Z - Z'| = k - Z'$ deletions.
Since $M$ is a $\cX$-multiway cut,
all vertices of $X$ that intersect a connected component of $G - M$ can be identified.
Thus, the problem reduces to solving several instances of
\textsc{Strict Steiner Multicut}, one for each connected component of $G - M$.
\iflong
  To summarize, Theorem~3.2~of~\cite{lokshtanov2021fpt} leads to the following observation.

  \begin{observation} \label{obs:steiner-to-strict-steiner}
  If \textsc{Strict Steiner Multicut} is solvable in $O^*(2^{O(k)})$ time,
  then \textsc{Steiner Multicut} is approximable within a factor of $2$ in 
  $O^*(2^{O(k)})$ time.
\end{observation}
\fi
\ifshort
  We conclude that if \textsc{Strict Steiner Multicut} is solvable in $O^*(2^{O(k)})$ time,
  then \textsc{Steiner Multicut} is approximable within a factor of $2$ in $O^*(2^{O(k)})$ time.
\fi

For the case with $p = 2$, \cite{lokshtanov2021fpt} invoke the algorithm for 
\textsc{Digraph Pair Cut}~of~\cite{kratsch2020representative}.
We generalize this algorithm to handle $p > 2$.
For this, we need to invoke a definition.
\iflong
  \begin{definition}[Section~3.2~in~\cite{kratsch2020representative}] 
    \label{def:closest-cut}
    Let $G = (V, E)$ be a graph and fix $x \in V(G)$.
    A set $W \subseteq V$ is \emph{closest to $v$}
    if $v \notin W$ and $W$ is the unique minimum $vW$-separator.
  \end{definition}
\fi
\ifshort
Let $G = (V, E)$ be a graph and fix $x \in V(G)$.
A set $W \subseteq V$ is \emph{closest to $v$}
if $v \notin W$ and $W$ is the unique minimum $vW$-separator.
\fi
Intuitively, the closest $vW$-separator is distinguished
among all $vW$-separators by the property that
its deletion minimizes the subset of vertices reachable from $v$.
The uniqueness of the closest minimum $vW$-separator
follows by submodularity of cuts.
Such a set can be computed in polynomial time
(see~e.g.~Section~3.2~of~\cite{kratsch2020representative}).

\begin{lemma} \label{lem:strict-steiner}
\textsc{Strict Steiner Multicut} with sets of size $p$ 
is solvable in $O^*(p^k)$ time.
\end{lemma}
\iflong
  \begin{proof}
    Let $(G, \cT, k)$ be an instance of \textsc{Steiner Multicut}
    and assume $x \in V(G)$ satisfies all subsets in $\cT$.
    Our goal is to compute a set of vertices $W \subseteq V(G)$
    such that $|W| \leq k$, $x \notin W$ and 
    $W$ satisfies all subsets in $\cT$.

    Initialize a set $Y = \emptyset$ and continue with the following steps.
    \begin{enumerate}
    \item \label{alg:separate} 
      Compute the minimum $xY$-separator $W$ closest to $x$.
    \item \label{alg:reject}  
      If $|W| > k$, reject the instance.
    \item \label{alg:accept} 
      If $W$ satisfies all subsets in $\cT$, accept the instance.
    \item \label{alg:branch} 
      Otherwise, pick a subset $T_i = \{t_1, \dots, t_p\}$
      that is not satisfied by $W$.
      Branch in $p$ directions, adding $t_i$ to $Y$ in each branch.
    \end{enumerate}

    Since $W$ in step~\eqref{alg:separate} is chosen to be a closest separator,
    the maxflow between $x$ and $Y$ is increased in step~\eqref{alg:branch}
    by adding a vertex to $Y$.
    Hence, the depth of the recursion tree is at most $k$,
    and the branching factor is $p$, which yields the running time $O^*(p^k)$.

    To prove correctness, the key observation is
    that a set $W \subseteq V(G)$ satisfies a subset $T_i \in \cT$
    if and only if at least one vertex in $t \in T_i$
    is separated from $x$ in $W$.
    Clearly, if all vertices of $T_i$ are 
    connected to $x$ in $G - W$, then $W$ does not satisfy $T_i$.
    Moreover, if $W$ is an $xt$-separator for some $t \in T_i$,
    then then either there is 
    a vertex $s \in T_i \setminus \{t\}$ 
    reachable from $x$ but not from $t$ in $G - W$,
    or the set $T_i$ is completely separated from $x$ in $G - W$.
    In the first case, $W$ is an $st$-separator for $\{s,t\} \subseteq T_i$, 
    so it satisfies $T_i$.
    In the second case, recall that $x$ satisfies $T_i$,
    so there is a pair of vertices in $T_i$ such that all paths
    connecting them contain $x$.
    Thus, separating $x$ and $T_i$ satisfies $T_i$.

    It remains to show that if $(G, \cT, k)$ is a yes-instance,
    then it has a solution that is closest to $x$.
    Suppose $Z$ is a solution to $(G, \cT, k)$
    and let $Z'$ be the unique minimum $vZ'$-separator.
    Clearly, $|Z'| \leq |Z| \leq k$ and if $Z$ separates $x$
    from $t \in V(G)$, then so does $Z'$.
    Hence, $Z'$ is a solution as well. 
  \end{proof}

  Combining Lemmas~\ref{obs:steiner-to-strict-steiner}~and~\ref{lem:strict-steiner}
  proves Theorem~\ref{thm:steiner-fpa}.
\fi

\ifshort

\begin{proof}[Proof Sketch]
  Let $(G, \cT, k)$ be an instance of \textsc{Steiner Multicut}
  and assume $x \in V(G)$ satisfies all subsets in $\cT$.
  Our goal is to compute a set of vertices $W \subseteq V(G)$
  such that $|W| \leq k$, $x \notin W$ and 
  $W$ satisfies all subsets in $\cT$.

  Initialize a set $Y = \emptyset$.
  Compute the minimum $xY$-separator $W$ closest to $x$.
  If $|W| > k$, reject the instance.
  If $W$ satisfies all subsets in $\cT$, accept the instance.
  Otherwise, pick a subset $T_i = \{t_1, \dots, t_p\}$
  that is not satisfied by $W$.
  Branch in $p$ directions, adding $t_i$ to $Y$ in each branch.

  Since $W$ i is chosen to be a closest separator,
  the maxflow between $x$ and $Y$ is increased
  by adding a vertex to $Y$.
  Hence, the depth of the recursion tree is at most $k$,
  and the branching factor is $p$, which yields the running time $O^*(p^k)$.

  The key observation for showing correctness is
  that a set $W \subseteq V(G)$ satisfies a subset $T_i \in \cT$
  if and only if at least one vertex in $t \in T_i$
  is separated from $x$ in $W$.
  We omit verification in the short version of the paper.
  \end{proof}
\fi

\iflong
  \section{Singleton Expansion}
  \label{sec:singleton-expansion}

  In this section we study the complexity of the \emph{singleton expansion}
  of equality languages, i.e., informally, the effect of enriching an
  equality language $\Gamma$ by assignment constraints $(v=i)$.
  Viewed as a constraint language, allowing constraints $(v=i)$ for some constant $i$
  corresponds to adding the singleton unary relation $R_i=\{(i)\}$ to
  $\Gamma$, hence the term singleton expansion.
  For completeness, we consider both adding just a constant number of
  such relations, and infinitely many.

  More formally, let $c \in \NN$. Define $\Gamma_c=\{\{(1)\}, \ldots, \{(c)\}\}$
  and $\Gamma_\NN=\{\{(i)\} \mid i \in \NN\}$. For an equality language
  $\Gamma$, define $\Gamma^+=\Gamma \cup \Gamma_\NN$ and for $c \in \NN$
  define $\Gamma_c^+=\Gamma \cup \Gamma_c$. For an equality language $\Gamma$,
  a \emph{singleton expansion of $\Gamma$} is either the language
  $\Gamma^+$ or $\Gamma_c^+$ for $c \in \NN$.
  For every equality language $\Gamma$ and every singleton expansion
  $\Gamma'$ of $\Gamma$, we study the complexity of $\mincsp{\Gamma'}$.
  
  We note that if $\Gamma=\{=\}$, then $\mincsp{\Gamma^+}$ naturally
  captures the problem \textsc{Edge Multiway Cut}. Let $(G,\cT,k)$ be an
  instance of \textsc{Edge Multiway Cut} with terminal set $\cT=\{t_1,\ldots,t_p\}$.
  Then we create an equivalent instance of $\mincsp{\Gamma^+}$ over
  variable set $V(G) \setminus \cT$ as follows. First, we assume there
  are no edges in $G[\cT]$, as any such edge must be deleted anyway.
  Next, for every edge $t_iv \in E(G)$ with $v \in V(G) \setminus \cT$ 
  and $t_i \in \cT$, add a soft constraint $(v=i)$, and for
  every edge $uv \in E(G-\cT)$ add a soft constraint $(u=v)$. 
  Clearly, this reduction can also be employed in the reverse,
  i.e., for every instance of $\mincsp{\Gamma^+}$ we can
  create an equivalent instance of \textsc{Edge Multiway Cut}.
  In a similar way, $\mincsp{=,\Gamma_2}$ corresponds to
  \textsc{$st$-Min Cut} and is in P, and \mincsp{=,\Gamma_s} for $s \in \NN$
  corresponds to \textsc{$s$-Edge Multiway Cut}, the restriction of
  \textsc{Edge Multiway Cut} to $s$ terminals, which for $s \geq 3$
  is NP-hard but FPT. 

  In this sense, studying singleton expansions of equality languages allows us
  to consider MinCSPs that may be intermediate between \textsc{Multiway Cut}
  and \textsc{Multicut}. 

  Unfortunately, our main conclusion is that nothing novel occurs in this range. 
  Let us first provide the explicit characterization of all properties
  under consideration for \mincsp{\Gamma^+}, since these cases are
  relatively manageable. Say that $\Gamma$ is \emph{positive conjunctive} if
  $\Gamma$ is both conjunctive and constant, i.e., every relation $R \in \Gamma$
  is defined as a conjunction of positive literals,
  and \emph{positive conjunctive and connected} if
  it is additionally split, i.e., the literals $(x_i=x_j)$ in the definition of a
  relation $R(x_1,\ldots,x_r) \in \Gamma$, if viewed as edges $x_ix_j$
  in a graph, form a connected graph. 
  
  \begin{theorem} \label{thm:gammaplus}
    Let $\Gamma$ be an equality language. The following hold.
    \begin{itemize}
    \item \csp{\Gamma^+} is in P if $\Gamma$ is Horn, and NP-hard otherwise
    \item \mincsp{\Gamma^+} is NP-hard
    \item \mincsp{\Gamma^+} has a polynomial-time constant-factor approximation if 
      $\Gamma$ is strictly negative or if $\Gamma$ is positive
      conjunctive, and otherwise it has no constant-factor approximation
      under UGC
    \item \mincsp{\Gamma^+} is FPT if either every relation in $\Gamma$
      is $\rel{NEQ}_3$ or split, or $\Gamma$ is strictly negative,
      otherwise it is W[1]-hard
    \item \mincsp{\Gamma^+} has a constant-factor FPT approximation
      if $\Gamma$ is negative, and otherwise it is \textsc{Hitting Set}-hard
    \end{itemize}
  \end{theorem}

  The cases for $\Gamma_c^+$, $c \in \NN$ are more involved, and
  require additional definitions. Let us provide the statements in
  multiple steps. Recall that \mincsp{\Gamma} and \mincsp{\Delta} are
  \emph{equivalent} if there are cost-preserving reductions in between
  the problems in both directions. 
  We first observe that if $\Gamma$ implements $\neq$
  and $=$, then we can ``emulate'' arbitrarily many constants via
  auxiliary variables (Lemma~\ref{lem:eq-neq-gives-constants}),
  and \mincsp{\Gamma'} maps back to \mincsp{\Gamma} for any singleton
  expansion $\Gamma'$ of $\Gamma$.  Now Theorem~\ref{thm:eq-csp-dichotomy}
  and Lemma~\ref{lem:neither-const-nor-neg} give the following.
  
  \begin{lemma} \label{olem:first}
    Let $\Gamma'$ be a singleton expansion of an equality language $\Gamma$.
    If $\Gamma$ is not Horn or constant, then \csp{\Gamma'} is NP-hard.
    If $\Gamma$ is Horn but not strictly negative or constant,
    then \mincsp{\Gamma'} is equivalent to \mincsp{\Gamma}.
    Otherwise $\Gamma$ is strictly negative or constant (and not both).
  \end{lemma}

  For strictly negative languages $\Gamma$, all the positive
  properties mentioned in Theorem~\ref{thm:gammaplus} of course
  translate to \mincsp{\Gamma'} for any singleton extension $\Gamma'$
  of $\Gamma$.
  
  \begin{lemma} \label{olem:strictly-negative}
    Let $\Gamma'$ be a singleton expansion of a finite strictly negative equality
    language $\Gamma$. 
    Then $\csp{\Gamma'}$ is in P and $\mincsp{\Gamma'}$ is NP-hard,
    FPT, and has a constant-factor approximation.
  \end{lemma}

  Finally, we assume that $\Gamma$ is constant.
  For $c \in \NN$, let the \emph{$c$-slice} of $\Gamma$ be the language
  \[
    \Delta = \{ R \cap [c]^{r(R)} \mid R \in \Gamma\}
  \]
  where $r(R)$ is the arity of $R$. For $c=2$, we will also interpret
  this as a Boolean language. The language $\Delta$ is
  \emph{trivial} if for every relation $R \in \Delta$, of arity $r$,
  either $R=\emptyset$ or $R=[c]^r$. We generalize the notions of
  \emph{positive conjunctive} and \emph{positive conjunctive and
    connected} to $\Delta$ in the natural way. Finally, a Boolean
  language is \emph{affine} if every relation in it can be modelled as
  the set of solutions to a system of affine linear equations over GF(2).
  We state the remaining cases.

  \begin{theorem}
    \label{thm:constant-gammacplus}
    Let $\Gamma$ be a constant equality language and $c \in \NN$.
    Let $\Delta$ be the $c$-slice of $\Gamma$.
    If $c=1$ then \mincsp{\Gamma_c^+} is in P. Otherwise the following hold.
    \begin{itemize}
    \item If $\Gamma$ has a retraction to domain $[c]$ then
      \mincsp{\Gamma_c^+} is equivalent to \mincsp{\Delta_c^+}.
    \item If $\Gamma$ has no retraction to domain $[c]$
      but is Horn, then \csp{\Gamma_c^+} is in P but
      \mincsp{\Gamma_c^+} is \textsc{Hitting Set}-hard
    \item Otherwise \csp{\Gamma_c^+} is NP-hard.
    \end{itemize}
    Furthermore, assume that $\Gamma$ has a retraction to domain $[c]$.
    If $c \geq 3$, then this implies that $\Delta$ is positive conjunctive.
    Furthermore the following hold.
    \begin{itemize}
    \item If $\Delta$ is trivial, then \mincsp{\Gamma_c^+} is in P
    \item If $\Delta$ is non-trivial, positive conjunctive and connected,
      then \mincsp{\Gamma_c^+} is in P for $c=2$, and NP-hard but FPT
      and constant-factor approximable for $c \geq 3$
    \item If $\Delta$ is positive conjunctive but not connected,
      then \mincsp{\Gamma_c^+} is W[1]-hard
      but constant-factor approximable
    \item If $c=2$ and $\Delta$ is affine but not positive
      conjunctive, then \csp{\Gamma_2^+} is in P but 
      \mincsp{\Gamma_c^+} has a cost-preserving
      reduction from \textsc{Nearest Codeword}
    \item Otherwise $c=2$ and \csp{\Gamma_2^+} is NP-hard 
    \end{itemize}    
  \end{theorem}

  This theorem in particular refines Theorem~\ref{ithm:pos-conj} from Section~\ref{sec:intro}.

  \subsection{The first step}

  Let us begin by showing that if $\Gamma$ implements $=$ and $\neq$ then the
  singleton expansion over $\Gamma$ adds no additional power.

  \begin{lemma} \label{lem:eq-neq-gives-constants}
    Let $\Gamma$ be an equality language that implements $=$ and $\neq$.
    Then there is a cost-preserving reduction from $\mincsp{\Gamma^+}$ to $\mincsp{\Gamma}$. 
  \end{lemma}
  \begin{proof}
    Let $I$ be an instance of $\mincsp{\Gamma^+}$ and let $C \subset \NN$
    be the set of constants $i$ used in assignment constraints $v=i$ in $I$.
    Create a new set of variables $T=\{t_i \mid i \in C\}$ and add a
    crisp constraint $t \neq t'$ for all distinct pairs $t, t' \in T$.
    Replace every constraint $(x=i)$ in $I$ by an implementation of
    $(x=t_i)$, and keep all other constraints unchanged. Let $I'$ be the
    output instance produced. We show that this reduction is cost-preserving.
    On the one hand, let $\alpha$ be an assignment to $V(I)$. Define an
    assignment $\alpha'$ to $V(I')$ by extending $\alpha$ by $\alpha(t_i)=i$
    for every $i \in C$. Then $\alpha$ and $\alpha'$ have the same cost.
    On the other hand, let $\alpha'$ be an assignment to $V(I')$ with
    finite cost. Since $I'$ consists of only equality-language
    constraints, we may apply any bijunction over $\NN$ to $\alpha'$
    and retain an assignment that satisfies the same set of constraints.
    In particular, since $\alpha'$ has finite cost it must hold that
    $\alpha'(t) \neq \alpha'(t')$ for every distinct pair $t, t' \in T$.
    We may then apply a bijunction to $\alpha'$ such that $\alpha'(t_i)=i$
    for every $i \in C$. Now letting $\alpha$ be the restriction of $\alpha'$
    to the variables $V(I)$ produces an assignment for $I$ of the same
    cost as $\alpha'$. Finally, assume that $I'$ has no finite-cost solutions.
    Then by the above, the same holds for $I$, as otherwise a
    finite-cost solution to $I$ could be transformed to a finite-cost
    solution to $I'$. 
  \end{proof}

  We may thus focus on languages that are either constant or strictly negative:
  By Theorem~\ref{thm:eq-csp-dichotomy}, if $\Gamma$ is not constant and
  not Horn, then $\csp{\Gamma}$ is NP-hard, and by Lemma~\ref{lem:neither-const-nor-neg},
  if $\Gamma$ is not constant
  and not strictly negative but is Horn, then Lemma~\ref{lem:eq-neq-gives-constants} applies.
  Recall in particular that every strictly negative language is Horn.

  Our main focus will be on languages that are constant, as they allow
  a richer range of interesting behaviour, but even strictly negative equality
  languages yield non-trivial problems under singleton expansion.

  Let us formalize the following simple consequence.

  \begin{corollary} \label{cor:horn-with-constants}
    If $\Gamma$ is Horn, then $\csp{\Gamma_c^+}$ and $\csp{\Gamma^+}$
    are in P for every $c$. 
  \end{corollary}
  \begin{proof}
    Apply Lemma~\ref{lem:eq-neq-gives-constants} to $\Gamma'=\Gamma \cup \{=,\neq\}$.
    The cost-preserving reduction in particular preserves satisfiability. 
  \end{proof}

  This finishes Lemma~\ref{olem:first}.
  
  \subsection{Strictly negative languages}

  Assume that $\Gamma$ is strictly negative. We note that in this case, adding a
  single unary singleton relation makes the optimization problem hard.

  \begin{lemma} \label{lem:vertex-cover}
    $\mincsp{\neq,x=1}$ is NP-hard. 
  \end{lemma}
  \begin{proof}
    There is a cost-preserving reduction from \textsc{Vertex Cover}.
    Let $G$ be a graph. Create an instance $I$ of ${\mincsp{\neq,x=1}}$
    as follows. The variable set is $V(G)$. For every edge $uv \in E(G)$, 
    create a crisp constraint $(u \neq v)$. For every vertex $v \in V(G)$,
    create a soft constraint $(v=1)$. On the one hand, let $S$ be a
    vertex cover for $G$. Then we may assign variables of $S$ distinct
    values from $2, 3, \ldots$, thereby breaking the soft constraint
    $(v=1)$ for every $v \in S$, but satisfying every crisp disequality
    constraint. Conversely, let $\alpha$ be an assignment to $I$ of
    finite cost. Let $S$ be the set of variables $v \in V(I)$ such 
    that $\alpha(v) \neq 1$. Then $S$ is a vertex cover of $G$.
  \end{proof}

  On the other hand, if $\Gamma$ is strictly negative then its
  singleton expansion has both an FPT algorithm and constant-factor
  approximation.
    
  \begin{lemma} \label{lem:negative-fpt}
    Let $\Gamma$ be a strictly negative language. Then $\mincsp{\Gamma^+}$ is FPT
    and has a constant-factor approximation.
  \end{lemma}
  \begin{proof}
    Let $I$ be an instance of $\mincsp{\Gamma^+}$ with parameter $k$. 
    We branch over obvious obstructions. First, assume there is a
    variable $v \in V(I)$ that is subject to contradictory assignment
    constraints $(v=i)$ and $(v=j)$, $i \neq j$. Then we recursively
    branch into the two cases of either setting $\alpha(v)=i$, in which case $(v=j)$
    is violated (possibly along with further assignment constraints on $v$),
    or concluding that $(v=i)$ is violated and deleting those constraints. In both branches, the parameter $k$
    decreases. Next, form a tentative assignment $\alpha$ by setting
    $\alpha(v)=i$ for every variable $v$ subject to an assignment
    constraint $(v=i)$, and setting all other variables to mutually
    distinct values, distinct from all values occurring in assignments
    in $I$. If $\alpha$ satisfies $I$, then we are done. Otherwise,
    let $R(x_1,\ldots,x_r)$ be a constraint violated by $\alpha$. 
    Since $R$ is strictly negative, it has a definition as a conjunction of
    strictly negative clauses. Let $C$ be such a clause that is violated by
    $\alpha$. Then $C$ is a disjunction over a finite number of negative
    literals $(x_i \neq x_j)$, and for every such literal $(x_i \neq x_j)$
    there are assignment constraints $(x_i=a)$ and $(x_j=a)$ in $I$ for
    some shared value $a$. Then we can branch on either $R$ being
    violated or an assignment constraint $(x_i=a_i)$, $i \in [r]$ being
    violated. This yields $r+1=O(1)$ branches, and since $k$ decreases
    in each branch the result is an FPT branching algorithm.

    For the approximation, similar arguments apply. Let $v$ be a
    variable subject to multiple contradictory assignment constraints.
    Let $i \in \NN$ be the constant for which the number of copies of a
    constraint $(v=i)$ is maximized. Let $m_1$ be the number of
    assignment constraints $(v=j)$ for $j \neq i$, and let $m_2$
    be the number of constraints $(v=i)$. Delete all constraints $(v=j)$
    for $j \neq i$, and delete $\min(m_1,m_2)$ constraints $(v=i)$.
    Let $X_1$ be set of assignment constraints deleted in total in this phase.
    Then any assignment to the instance will violate at least half of
    the constraints in $X_1$, and in the remaining instance every
    variable $v$ occurs in assignment constraints $(v=i)$ for at most
    one value $i$. Let $\alpha$ be as above, selecting an assignment
    that satisfies all assignment constraints and assigns unique
    distinct values to all variables not occurring in assignment
    constraints. Then, as above, if there is a constraint $R(X)$ not
    satisfied by $\alpha$ then there is an explicit contradiction
    between $R(X)$ and at most $r(R)=O(1)$ assignment constraints.
    Since at least one of these constraints must be violated in any
    assignment, we get an $O(1)$-approximation by simply deleting both
    $R(X)$ and one copy of every assignment constraint with scope
    intersecting $X$. Repeating this until no constradictions remain
    gives an $r(\Gamma)$-approximation, where $r(\Gamma)$ is the largest
    arity of a relation in $\Gamma$. 
  \end{proof}

  We summarise the properties of singleton expansions of strictly negative
  languages as follows.

  \begin{lemma}[Lemma~\ref{olem:strictly-negative}, repeated]
    Let $\Gamma'$ be a singleton expansion of a finite strictly negative equality
    language $\Gamma$. 
    Then $\csp{\Gamma'}$ is in P and $\mincsp{\Gamma'}$ is NP-hard
    but FPT, and has a constant-factor approximation.
  \end{lemma}
  \begin{proof}
    \csp{\Gamma'} is in P by Cor.~\ref{cor:horn-with-constants}
    (or indeed by Lemma~\ref{lem:negative-fpt}). 
    For NP-hardness, we invoke Lemma~\ref{lem:vertex-cover}. In
    particular since no proper relation is both strictly negative and
    constant and $\Gamma$ contains at least one relation by assumption,
    Lemma~\ref{lem:neither-const-nor-neg} implies that $\Gamma$ implements
    $\neq$. Then $\mincsp{\Gamma_1^+}$ is already NP-hard by
    Lemma~\ref{lem:vertex-cover}, and adding further unary singleton
    relations to the language clearly does not change this fact. 
    Finally, the FPT algorithm and constant-factor approximation of
    Lemma~\ref{lem:negative-fpt} also applies to $\mincsp{\Gamma'}$. 
  \end{proof}

  \subsection{Constant languages}

  We now assume that $\Gamma$ is constant but not strictly negative.
  For studying these cases, we need to employ the algebraic machinery  
  for studying CSPs.

  Let $R \subseteq D^r$ be a relation, and $f \colon D^c \to D$ an
  operation on a domain $D$. We say that $f$ \emph{preserves} $R$ if,
  for any $t_1, \ldots, t_c \in R$ we have
  $f(t_1,\ldots,t_c) \in R$, where $f$ is applied component-wise.
  Let $\Gamma$ be a constraint language over $D$.
  A \emph{polymorphism} of $\Gamma$ is an operation over $D$
  that preserves every relation $R \in \Gamma$.
  A polymorphism $f \colon D^c \to D$ is \emph{essentially unary}
  if there exists an index $i \in [c]$ and an operation $g \colon D \to D$
  such that $f(x)=g(x_i)$ for every $x \in D^c$. 
  A polymorphism is \emph{essential} if it is not essentially unary. 
  Polymorphisms are a standard tool in studying finite constraint
  languages, since they characterize the expressive power of the
  language up to pp-definability. Bodirsky~\cite{Bodirsky-Hab} shows
  that under mild assumptions they can be used for infinite languages
  and languages over infinite domains as well.

  \begin{theorem}[{\cite[Theorem~5.2.3]{Bodirsky-Hab}}]
    Let $\Gamma$ be an $\omega$-categorical structure. A relation $R$
    has a pp-definition in $\Gamma$ if and only if $R$ is preserved
    by all polymorphisms of $\Gamma$.
  \end{theorem}

  In particular, every equality constraint language is
  $\omega$-categorical. Furthermore, for any equality constraint
  language $\Gamma$ and any $c \in \NN$, the language $\Gamma_c^+$ is
  $\omega$-categorical (see Section~3.1 of Bodirsky~\cite{Bodirsky-Hab}).
  We will not need to consider any subtleties around the language $\Gamma^+$
  which has an infinite number of additional relations, since all
  hardness claims over languages $\Gamma^+$ will follow from finite
  subsets $\Gamma_c^+$ of $\Gamma^+$. 

  We give two simple statements for reference.

  \begin{lemma} \label{lem:essential-or-hard}
    Let $\Gamma$ be an equality constraint language and $c \in \NN$ a
    constant, $c \geq 2$. If $\Gamma_c^+$ has no essential polymorphisms
    then $\csp{\Gamma_c^+}$ is NP-hard.
  \end{lemma}
  \begin{proof}
    Define $Q(a,b,c)$ over $\NN$ as the relation $(a=b \lor b=c)$. 
    As noted by Bodirsky~{\cite[Lemma~5.3.2]{Bodirsky-Hab}}, $Q$ is
    preserved by all essentially unary operations but has no essential
    polymorphisms. As noted above, $\Gamma_c^+$ is $\omega$-categorical,
    and therefore pp-defines $Q$. However, $\csp{Q,x=1,x=2}$ is NP-hard:
    Note that $\exists y_1, y_2: (y_1=1) \land (y_y=2) \land Q(y_1,x,y_2)$ pp-defines the unary
    relation $x \in \{1,2\}$. Thus it is enough that $\csp{Q}$ is
    NP-hard as a Boolean language, which is standard (e.g., see Exercise~3.24~in~\cite{chen2006rendezvous}).  
  \end{proof}

  We also give a direct proof of the following implementation result.
  Recall that $\rel{ODD}_3 \subset \NN^3$ accepts any tuple that takes
  either one or three distinct values.

  \begin{lemma} \label{lem:retract-or-odd3}
    Let $\Gamma$ be a constant equality language. Let $f \colon \NN \to \NN$
    be the retraction defined by $f(1)=1$ and $f(x)=2$ otherwise.
    If $\Gamma$ is not preserved by $f$, then $\Gamma_2^+$ implements $\rel{ODD}_3$. 
  \end{lemma}
  \begin{proof}
    Let $R \in \Gamma$ be a relation not preserved by $f$, of arity $r$. 
    For a tuple $\bt \in \NN^r$, write $f(\bt)=(f(t_1),\ldots,f(t_r))$
    for the result of applying $f$ to $\bt$. Since $\Gamma$ is constant,
    $\ba := (1,\ldots,1) \in R$, and since $R$ is not preserved by $f$
    there is a tuple $\bc \in R$ such that $f(\bc) \notin R$. 
    Note that we have a refinement order, where $\bc$ strictly refines
    $f(\bc)$ which strictly refines $\ba$. Let $\bc$ be a least refined
    tuple in $R$ such that $f(\bc) \notin R$, and let $r'$ be the number
    of distinct values used in $\bc$. Without loss of generality, 
    assume that $\bc$ uses values $1$ through $r'$. Note that $r' \geq 3$.
    Define a relation $R'$ as $R'(x_1,\ldots,x_{r'})=R(x_{c_1},\ldots,x_{c_r})$.
    Then $R'$ accepts $(1,\ldots,1)$ and $(1,\ldots,r')$,
    but no tuple between these two in the refinement order. 
    Thus $R''(x,y,z) = \exists_{x_4,\ldots,x_{r'}} R'(x,y,z,x_4,\ldots,x_{r'})$ 
    is an implementation of $\rel{ODD}_3$.
  \end{proof}

  We show that $\mincsp{\rel{ODD}_3,\Gamma_2}$ is \textsc{Hitting Set}-hard.

  \begin{lemma} \label{lem:odd3-constants-hard}
    There is a cost-preserving reduction from \textsc{Hitting Set}
    to $\mincsp{\rel{ODD}_3,\Gamma_2}$.
  \end{lemma}
  \begin{proof}
    We give a variant of the hardness proof for
    \mincsp{\rel{ODD}_3,=,\neq} (Lemma~\ref{lem:hitting-set-to-odd3}).

    \begin{claim} \label{claim:odd3-nary}
      The language $\{\rel{ODD}_3,\Gamma_2\}$ pp-defines, for every arity
      $\ell$, the relation $R(x_1,\ldots,x_\ell)$ that accepts every
      tuple except $(1,\ldots,1)$ and $(2, \ldots,2)$. 
    \end{claim}
    \begin{claimproof}
      Create a pp-definition with local variables $z_1$, $z_2$
      and $Y=\{y_2,\ldots,y_\ell\}$ using the constraints
      \[
        (z_1=1) \land (z_2=2) \land 
        \rel{ODD}_3(x_1,x_2,y_2) \land
        \bigwedge_{i=2}^{\ell-1} \rel{ODD}_3(y_i,x_{i+1},y_{i+1}) \land
        \rel{ODD}_3(z_1,z_2,y_\ell).
      \]
      We claim that this pp-defines the relation $R$.
      First, assume that $x_1=x_2=\ldots=c$ for some value $c$.
      Then $y_\ell=c$ by induction. If $c \in \{1,2\}$, then
      the formula is unsatisfiable, otherwise the assignment where
      $y=c$ for every $y \in Y$ satisfies the formula. 
      Next, assume that $x_i \neq x_{i+1}$. Then we may set $y_{i+1}$
      to any free value, and by induction we can set all variables $y_j$
      for $j>i$ to distinct values not equal to 1 or 2 and not used by
      any variable $x_i$. Then also the constraint $\rel{ODD}_3(z_1,z_2,y_\ell)$
      is satisfied.
    \end{claimproof}

    Now let $I=(n, \cF, k)$ be an instance of \textsc{Hitting Set} (using
    the same encoding as in Lemma~\ref{lem:hitting-set-to-odd3}). 
    We create an instance $I'$ of $\mincsp{\rel{ODD}_3,\Gamma_2}$
    on variable set $V=\{v_1,\ldots,v_n\}$ as follows. 
    First, create a soft constraint $(v_i=1)$ for every $i \in [n]$.
    Next, for every set $F \in \cF$ we construct the relation of
    Claim~\ref{claim:odd3-nary} over the variable set $\{v_i \mid i \in F\}$. 
    We let every such construction consist of crisp constraints only. 
    We claim that this defines a cost-preserving reduction.
    On the one hand, let $S \subseteq [n]$ be a hitting set for $\cF$.
    Then we assign $\alpha(v_i)=3$ for $i \in S$ and $\alpha(v_i)=1$
    otherwise. This satisfies every constraint of $I'$ except the soft
    constraints $(v_i=1)$ for $i \in S$, i.e., the cost of $\alpha$ is
    precisely $|S|$. On the other hand, let $\alpha$ be an assignment to $I'$
    of some finite cost $c$, and let $S = \{i \in [n] \mid \alpha(v_i) \neq 1\}$
    be the indices corresponding to violated soft constraints in $I'$.
    Then $S$ is a hitting set of $\cF$. 
  \end{proof}

  With the preliminaries above in place, we can provide the main
  algebraic characterization we use from previous work.

  The following is a direct rephrasing of a result of Bodirsky et
  al.~\cite{BodirskyCP10equality}. We use the phrase
  \emph{$f$ takes (at most) $k$ values}
  to mean that the image of $f$ has cardinality (at most) $k$.
  For the definition of quasilinear operation,
  see~\cite{BodirskyCP10equality}; we will only need that
  quasilinear operations take at most two values. 

  \begin{theorem}[Theorem 8 of \cite{BodirskyCP10equality}]
    \label{thm:bodirsky:thm8}
    Let $\Gamma$ be an equality language with at least one unary
    polymorphism that is not constant or injective. If $\Gamma$ has any
    essential polymorphism, then the essential polymorphisms of  
    $\Gamma$ are described by one of the following cases.
    \begin{enumerate}
    \item All quasilinear operations
    \item For some $k \geq 2$, all operations which take at most $k$
      different values
    \end{enumerate}
  \end{theorem}

  The remaining case looks as follows. 

  \begin{lemma} \label{lem:support-iplus}
    Let $\Gamma$ be a constant equality language that does not have any
    unary polymorphism that is not injective or a constant. Then
    $\Gamma_2^+$ implements $\rel{ODD}_3$. Furthermore, if $\Gamma$ has an
    essential polymorphism then $\Gamma$ is Horn. 
  \end{lemma}
  \begin{proof}
    The first statement follows from Lemma~\ref{lem:retract-or-odd3}.
    The second statement is Theorem~15 of Bodirsky et al.~\cite{BodirskyCP10equality}.
  \end{proof}

  \subsubsection{At most two singletons}

  We first treat the case of $\Gamma_c^+$ for $c \leq 2$, since it
  behaves somewhat differently from the rest. 
  We begin with a trivial observation.

  \begin{proposition} \label{lem:single-trivial}
    If $\Gamma$ is a constant equality constraint language, then $\mincsp{\Gamma_1^+}$ is in P.
  \end{proposition}

  Thus we consider the language $\Gamma_2^+$. Our main strategy will be
  a reduction to a Boolean MinCSP problem, all of which have been fully
  characterized~\cite{KimKPW23flow3,bonnet2016mincsp}.

  \begin{lemma} \label{lem:cases-two-constants}
    Let $\Gamma$ be a finite, constant equality language. 
    One of the following applies.
    \begin{itemize}
    \item $\Gamma_2^+$ has a retraction to domain $\{1,2\}$
      and $\mincsp{\Gamma_2^+}$ is equivalent under cost-preserving reductions to
      $\mincsp{\Delta}$ for some Boolean language $\Delta$
    \item $\Gamma$ is Horn and $\csp{\Gamma_2^+}$ is in P,
      but $\mincsp{\Gamma_2^+}$ is \textsc{Hitting Set}-hard
    \item $\csp{\Gamma_2^+}$ is NP-hard
    \end{itemize}
  \end{lemma}
  \begin{proof}
    First assume that $\Gamma$ is preserved by the retraction $f$ defined in Lemma~\ref{lem:retract-or-odd3}.
    Define the language 
    \[
      \Gamma'=\{R \cap \{1,2\}^{r(R)} \mid R \in \Gamma\}
    \]
    as the slice of $\Gamma$ that only uses two values, and
    interpret $\Gamma'$ as a language over domain $\{1,2\}$.
    Let $\Delta$ be the corresponding Boolean language, under the domain
    renaming $1 \mapsto 0$ and $2 \mapsto 1$, with the singleton
    relations $(x=0)$ and $(x=1)$ added to $\Delta$. 
    We claim that $\mincsp{\Gamma_2^+}$ and $\mincsp{\Delta}$ are
    equivalent problems. Indeed, let $I$ be an instance of $\mincsp{\Gamma_2^}$
    and let $I'$ be the instance of $\mincsp{\Delta}$ resulting from replacing
    every constraint in $I$ to the corresponding constraint over $\Delta$. 
    Let $\varphi$ be an assignment to $V(I)$. By applying the retraction
    $f$ to $\varphi$ followed by the domain renaming mapping,
    we get an assignment $\varphi' \colon V(I) \to \{0,1\}$
    which violates at most as many constraints in $I'$ as $\varphi$
    violates in $I$. Correspondingly, any assignment $\varphi' \colon V(I') \to \{0,1\}$ 
    can be used directly as an assignment to $V(I)$, and violates precisely
    the same set of constraints in $I$ as $\varphi'$ violates in $I'$. 
    Clearly, the same reduction also works when interpreted as a
    reduction from $\mincsp{\Delta}$ to $\mincsp{\Gamma_2^+}$. 

    Otherwise, if $\Gamma$ is not preserved by $f$ then by Lemma~\ref{lem:retract-or-odd3}, $\Gamma_2^+$
    implements $\rel{ODD}_3$ and $\mincsp{\Gamma_2^+}$ is at least
    \textsc{Hitting Set}-hard by Lemma~\ref{lem:odd3-constants-hard},
    and the question is whether $\csp{\Gamma_2^+}$ is in P.
    If $\Gamma$ has no essential polymorphism, then $\csp{\Gamma_2^+}$
    is NP-hard by Lemma~\ref{lem:essential-or-hard}.
    Hence assume that $\Gamma$ has some essential polymorphism.
    If $\Gamma$ has a unary polymorphism that is not constant or an
    injection, then by Theorem~\ref{thm:bodirsky:thm8},
    $\Gamma$ is preserved by every quasilinear operation. However, the
    retraction $f$ of Lemma~\ref{lem:retract-or-odd3} is quasilinear
    (see Bodirsky  et al.~\cite{BodirskyCP10equality}), so this case is
    impossible. In the remaining case, $\Gamma$ is Horn by
    Lemma~\ref{lem:support-iplus} and $\csp{\Gamma_2^+}$ is in P
    by Cor.~\ref{cor:horn-with-constants}.
  \end{proof}

  Finally, the possibilities for the last case here are quite limited.

  \begin{lemma} \label{lem:cases-boolean}
    Let $\Gamma$ be a constant equality language such that $\Gamma_2^+$
    has a retraction to domain $\{1,2\}$ and let $\Delta$ be the
    corresponding Boolean language. The following hold.
    \begin{itemize}
    \item If $\Delta$ is positive conjunctive and connected then
      \mincsp{\Gamma_2^+} is in P
    \item If $\Delta$ is positive conjunctive but not connected,
      then \csp{\Gamma_2^+} is in P and 
      \mincsp{\Gamma_2^+} is W[1]-hard but has a constant-factor
      approximation
    \item If $\Delta$ is not positive conjunctive but affine,
      then \csp{\Gamma_2^} is in P but
      \mincsp{\Gamma_2^+} is \textsc{Nearest Codeword}-hard
    \item Otherwise \csp{\Gamma_2^+} is NP-hard.
    \end{itemize}
  \end{lemma}
  \begin{proof}
    Let $\Gamma'$ be the intersection of $\Gamma$ with domain $\{0,1\}$,
    interpreted as a Boolean language (as in Lemma~\ref{lem:cases-two-constants}).
    The possible cases for such languages follow from Bonnet et al.~\cite{bonnet2016mincsp}
    and the structure of Post's lattice of co-clones.
    Specifically, $\Gamma'$ is 0-valid, 1-valid and preserved by negation. 
    If $\Gamma'$ has no further polymorphism, then \csp{\Gamma_2^+} is
    NP-hard; if $\Gamma'$ is preserved by the 3-ary XOR operation, then
    \csp{\Gamma_2^+} is in P, but if $\Gamma'$ has no further
    polymorphism then \mincsp{\Gamma_2^+} is as hard to FPT-approximate
    as \textsc{Nearest Codeword}; and in every remaining case, $\Gamma'$
    is positive conjunctive~\cite{bonnet2016mincsp}.
    Finally, if a Boolean language $\Gamma'$ is positive conjunctive,
    then either every relation in $\Gamma'$ is connected, in which case
    the relation is submodular and $\mincsp{\Gamma_2^+}$ is in P,
    or $\Gamma'$ implements $R_{=,=}$ and $\mincsp{\Gamma_2^+}$ is
    W[1]-hard. However, as in Lemma~\ref{lem:negative-fpt} we
    can split every relation $R \in \Gamma'$ into separate equality
    constraints, and reduce to \textsc{st-Min Cut}, up to a
    constant-factor approximation loss.
  \end{proof}

  For completion, let us provide example languages for each of the
  cases of Lemma~\ref{lem:cases-boolean}. If $\Gamma=\{Q\}$ where 
  $Q(a,b,c) \equiv (a=b \lor b=c)$ is from Lemma~\ref{lem:essential-or-hard},
  then $\Gamma$ is closed under the retraction to domain $\{1,2\}$ but
  $\csp{\Gamma_2^+}$ is NP-hard. Next, consider the relation
  $R(a,b,c,d)$ which accepts any assignment where every block has even
  cardinality (i.e., $R$ accepts tuples $(1,1,1,1)$, $(1,1,2,2)$,
  $(1,2,1,2)$ and $(1,2,2,1)$). Then $\mincsp{\Gamma_2^+}$ corresponds 
  to $\mincsp{\Delta}$ for the Boolean language $\Delta=\{x=0, x=1, a+b+c+d=0 \pmod 2\}$
  which is \textsc{Nearest Codeword}-hard~\cite{bonnet2016mincsp}.
  The final two cases are represented by the languages 
  $\Gamma=\{R_{=,=}\}$ and $\Gamma=\{=\}$. 
  As a final example, consider $\Gamma=\{R\}$ where $R(a,b,c,d)$ is the 4-ary relation
  what accepts any tuple $(a,b,c,d)$ where $a=b$ and $|\{a,b,c,d\}| \neq 3$.
  Then $\Gamma$ is not itself positive conjunctive,
  but $\Gamma$ is constant and closed under the retraction $f \colon \NN \to \{1,2\}$
  and the Boolean slice $\Gamma'$ of $\Gamma$ simply contains the relation $(a=b)$
  with two additional irrelevant arguments.

  \subsubsection{At least three singletons}

  The cases where $c \geq 3$ (and the case of $\Gamma^+$) are more regular.

  \begin{lemma} \label{lem:c3-cases}
    Let $c \geq 3$ and let $\Gamma$ be a constant equality language.
    One of the following applies.
    \begin{itemize}
    \item \mincsp{\Gamma_c^+} is equivalent to \mincsp{\Delta_c^+}    
      where $\Delta$ is the $c$-slice of $\Gamma$, and $\Delta$ is
      positive conjunctive
    \item $\Gamma$ is Horn but the last case does not apply;
      \csp{\Gamma_c^+} is in P, but \mincsp{\Gamma_c^+} is \textsc{Hitting Set}-hard
    \item Neither case applies, and \csp{\Gamma_c^+} is NP-hard.
    \end{itemize}
    Similarly, one of the following applies.
    \begin{itemize}
    \item $\Gamma$ is positive conjunctive
    \item $\Gamma$ is Horn but not positive conjunctive; 
      $\csp{\Gamma^+}$ is in P, but $\mincsp{\Gamma^+}$ is \textsc{Hitting Set}-hard
    \item Neither case applies, and $\csp{\Gamma^+}$ is NP-hard.
    \end{itemize}
  \end{lemma}
  \begin{proof}
    Consider the polymorphisms of $\Gamma_c^+$. They are precisely the
    intersection of the polymorphisms of $\Gamma$ and of $\Gamma_c$,
    i.e., every polymorphism $f$ of $\Gamma$ such that 
    $f(i,\ldots,i)=i$ for every $i \in [c]$.  Then any such $f$ is an
    operation that takes at least $c$ values. First, assume that $\Gamma$
    has at least one unary polymorphism that is not constant or injective,
    so that Theorem~\ref{thm:bodirsky:thm8} applies. In this case,
    either $\Gamma_c^+$ has no essential polymorphisms or $\Gamma$ is 
    preserved by every operation that takes $c$ values. 
    In the former case, $\csp{\Gamma_c^+}$ is NP-hard by
    Lemma~\ref{lem:essential-or-hard}. In the latter case, it first of all
    follows that $\Gamma_c^+$ is preserved by a retraction to domain $[c]$.
    Furthermore, consider the resulting slice language
    \[
      \Gamma'=\{R \cap [c]^{r(R)} \mid R \in \Gamma\}
    \]
    (where $r(R)$ denotes the arity of $R$) as a language over domain $[c]$.
    Then every relation $R \in \Gamma'$ 
    is preserved by every operation over $[c]$. The only such relations
    are $=$ and relations pp-defined over $=$. Thus, the slice language
    $\Gamma'$ is positive conjunctive.
    Otherwise, by assumption Lemma~\ref{lem:support-iplus} applies, so that
    $\Gamma_c^+$ implements $\rel{ODD}_3$ and \mincsp{\Gamma_c^+} is
    \textsc{Hitting Set}-hard by Lemma~\ref{lem:odd3-constants-hard}.
    Furthermore, either $\Gamma$ has no essential polymorphisms or $\Gamma$ is Horn.
    In the former case $\csp{\Gamma_c^+}$ is NP-hard by Lemma~\ref{lem:essential-or-hard};
    in the latter case $\csp{\Gamma_c^+}$ is in P by Cor.~\ref{cor:horn-with-constants}.

    The case of $\Gamma^+$ is similar. Clearly the hardness cases for
    $\Gamma_c^+$, $c \geq 3$ carry over to $\Gamma^+$. Hence if $\Gamma$
    is not Horn then $\csp{\Gamma^+}$ is NP-hard. Furthermore, let
    $R \in \Gamma$ be any relation not preserved by all operations. 
    Then every tuple of $R$ uses at most $r(R)$ distinct values,
    and $R \cap [r(R)]^{r(R)}$ is not preserved by all operations. 
    Hence $\Gamma_{r(R)}^+$ is already \textsc{Hitting Set}-hard by the above.
    The only remaining case is that $\Gamma$ itself is pp-definable over
    $\{=\}$, i.e., positive conjunctive.
  \end{proof}

  We also observe the tractable cases for positive conjunctive languages. 

  \begin{lemma} \label{lem:c3-posconj}
    Let $\Gamma$ be a finite, positive conjunctive equality language,
    possibly trivial, and let $\Gamma'$ be a singleton expansion of
    $\Gamma$ with at least three singleton relations. 
    Then \csp{\Gamma'} is in P and \mincsp{\Gamma'} has a
    constant-factor approximation. Furthermore the following hold. 
    \begin{enumerate}
    \item If $\Gamma$ is trivial, then \mincsp{\Gamma'} is in P
    \item If $\Gamma$ is connected but not trivial,
      then \mincsp{\Gamma'} is NP-hard but FPT
    \item If $\Gamma$ is not connected, 
      then \mincsp{\Gamma'} is W[1]-hard.
    \end{enumerate}
  \end{lemma}
  \begin{proof}
    Since every conjunctive language is Horn, \csp{\Gamma'} is in P 
    by Cor.~\ref{cor:horn-with-constants}.   
    If $\Gamma$ is trivial, then instances of \mincsp{\Gamma'}
    consist of trivial constraints (always satisfied, or never satisfied)
    which can be discarded, and unary constraints $(x=i)$. 
    An optimal assignment to such an instance can be computed easily
    with a greedy algorithm. 
    Otherwise, let $R \in \Gamma$ be non-trivial. Since $R$ is
    positive conjunctive, $R$ is Horn and not strictly negative and
    $R$ implements $=$ by Lemma~\ref{lem:neither-const-nor-neg}.
    Since \mincsp{=,\Gamma_3} captures \textsc{Edge Multiway Cut} with
    three terminals, which is NP-hard, it follows that \mincsp{\Gamma'}
    is NP-hard. If $\Gamma$ is positive conjunctive and connected,
    then $\Gamma \cup \{=,\neq\}$ is split,
    hence \mincsp{\Gamma',=,\neq} is FPT by Theorem~\ref{thm:fpt-class}
    and Lemma~\ref{lem:eq-neq-gives-constants}.
    Finally, let $R \in \Gamma$ be a relation that is positive
    conjunctive but not connected. Then $R$ is defined by a conjunction of 
    positive literals $(x_i=y_i)$, and the graph induced over its
    arguments has at least two non-trivial connected components
    since otherwise $R$ is split. Then $R$ implements the relation
    $R'(x_1,y_1,x_2,y_2) \equiv (x_1=y_1) \land (x_2=y_2)$
    using existential quantification over irrelevant variables
    and \mincsp{\Gamma'} is W[1]-hard by a simple reduction from
    \textsc{Split Paired Cut} (see Section~\ref{ssec:split-paired-cut}). 
    
    For approximation, since $\Gamma$ is finite there is a finite bound
    $d \in \NN$ such that every relation $R \in \Gamma$ is a conjunction
    of at most $d$ terms. Then we can split every constraint
    \[
      R(x_1,\ldots,x_r)=(x_{i_1} = x_{j_1}) \land \ldots \land (x_{i_d}=x_{j_d})
    \]
    in an instance of \mincsp{\Gamma'} into the $d$ separate constraints $(x_{i_p}=x_{j_p})$,
    $p \in [d]$. This increases the cost of the instance at most by a
    factor $d$. Now we have an instance with just equality and
    assignments, which reduces to \textsc{Edge Multiway Cut} which has a
    constant-factor approximation~\cite{CalinescuKR00mwc}.
  \end{proof}

  \section{Discussion}
  \label{sec:discuss}
  
  We classify the parameterized complexity of $\mincsp{\Gamma}$ for all
  finite equality constraint languages $\Gamma$ and their singleton expansions.
  In particular, we show that for an equality language $\Gamma$, $\mincsp{\Gamma}$ is in FPT
  if $\Gamma$ only contains split relations and $\rel{NEQ}_3$,
  it is W[1]-hard and fpt-approximable within a constant factor if 
  $\Gamma$ is negative but contains a relation that is neither
  split nor $\rel{NEQ}_3$, and is \textsc{Hitting Set}-hard otherwise.
  We also show that the complications introduced by singleton expansion
  can be handled by well-known fpt algorithms.

  Note that singleton expansion is distinct from adding constants to the underlying
  structure $(\NN,=)$, as the latter creates a much more powerful
  language. For example, just adding a single constant to the structure
  allows first-order reducts to implement arbitrary Boolean relations
  (using, e.g., $v=0$ and $v \neq 0$ as the two domain values).
  Hence, studying $\mincsp{\Gamma}$ for reducts of the structure $(\NN,=)$
  with any finite number of constants added generalizes the task of
  producing a parameterized dichotomy for $\mincsp{\Gamma}$
  over all finite-domain language $\Gamma$, which is a highly
  challenging task. There is a CSP dichotomy for this
  setting~\cite{BodirskyM18unary}, but it explicitly reduces back to 
  the CSP dichotomy for finite-domain languages, and no such result is
  known for parameterized complexity.
  The complexity of $\mincsp{\Gamma}$ for the reducts of $(\NN,=)$ 
  extended with finitely many constants is therefore an interesting but challenging
  question as it generalizes $\mincsp{\Gamma}$ for all finite-domain
  constraint languages $\Gamma$. We note that this question is open
  even for the case of $(\NN,=)$ plus a single constant, as mentioned above.

  The natural next step is to study \emph{temporal constraint languages}
  where the relations are first-order definable over $(\QQ; <)$.
  Note that relations $\leq$ and $\neq$ belong to this class,
  and $\mincsp{\leq, \neq}$ is equivalent to
  \textsc{Symmetric Directed Multicut},
  an important open problem on digraphs~\cite{EibenRW22ipec,kim2022weighted}.
  One way forward would be to study \emph{constant-factor fpt approximation} instead
  since there the classes of tractable languages are preserved by
  (equality-free) pp-definitions, and thus one can employ
  powerful tools from the algebraic approach to CSP.  

\fi

\bibliographystyle{plainurl}
\bibliography{references}

\end{document}